\documentclass[draftclsnofoot,onecolumn]{IEEEtran}
\newif\ifThesis
\newif\ifVerboseICIntro
\newif\ifTITJournal
\TITJournaltrue
\usepackage{amsmath,dsfont}
\usepackage{amssymb}
\usepackage{mathrsfs} % Used for \mathscr
\usepackage{ulem} %Used for \msout or \sout
\usepackage{epsfig,epsf}
\usepackage{color}
\def\3To1BC{$3-$to$-1$}

\def\define{:{=}~}
\def\naturals{\mathbb{N}}
\def\reals{\mathbb{R}}

\def\Var{\mbox{Var}}

\def\UHK{\mathscr{U}\!\mathcal{S}\mathcal{B}-}

\def\underlinem{\underline{m}}

\def\underlineY{\underline{Y}}
\def\underliney{\underline{y}}
\def\underlined{\underline{d}}
\def\underlineR{\underline{R}}
\def\GroupSum{\oplus_{4}}

\def\underlineSetY{\underline{\OutputAlphabet}}

\def\MessageSetM{\mathcal{M}}
\def\underlineMessageSetM{\underline{\MessageSetM}}
\def\hatm{\hat{m}}
\def\cl{\mbox{cl}}
\def\cocl{\mbox{cocl}}

\def\SetOfDistributions{\mathbb{D}}

\def\TimeSharingRVSet{\mathcal{Q}}
\def\TimeSharingRV{Q}

\def\SemiPrivateRVSet{\mathcal{U}}
\def\underlineSemiPrivateRV{\underline{\SemiPrivateRV}}

\def\underlineSemiPrivateRVSet{\underline{\SemiPrivateRVSet}}

\def\InputRV{X}

\def\SemiPrivateRV{U}

\def\Expectation{\mathbb{E}}

\def\fieldpij{\mathcal{F}_{\Prime_{j}}}
\def\fieldpii{\mathcal{F}_{\Prime_{i}}}
\def\fieldpik{\mathcal{F}_{\Prime_{k}}}

\def\integers{\mathbb{Z}}

\def\fieldpi{\mathcal{F}_{\Prime}}

\def\underlineX{\underline{X}}

\def\underlinee{\underline{e}}
\def\Prime{\theta}

\def\threeIC{$3-$IC }
\def\2IC{$2-$IC}
\def\decoder{d}
\def\ulined{\underline{\decoder}}
\def\messsagem{m}
\def\ulinem{\underline{\messsagem}}
\def\underlineMessageSet{\underline{\MessageSet}}
\def\cost{\tau}
\def\ulinecost{\underline{\cost}}

\def\timeshare{q}

\def\nlettercost{\bar{\costfn}^{n}}

\def\three-1{3\mbox{-}1}

\def\\3to1IC{$3-$to$-1$ IC}
\def\threeto1{$3-$to$-1$}

\def\ulineMessageSet{\underline{\MessageSet}}
\def\encoder{e}
\def\ulinee{\underline{\encoder}}

\def\ulineInputAlphabet{\underline{\InputAlphabet}}
\def\ulineOutputAlphabet{\underline{\OutputAlphabet}}
\def\2IC{$2-$IC}
\def\OutputRV{Y}
\def\ulineOutputRV{\underline{\OutputRV}}

\def\ulineInputRV{\underline{\InputRV}}

\def\ulineoutput{\underline{y}}
\def\ulineinput{\underline{x}}
\def\OutRV{\OutputRV}

\def\InRV{\InputRV}
\def\output{y}
\def\inp{x}

\def\costfn{\kappa}
\def\ulinecostfn{\underline{\costfn}}

\def\MessageSet{\mathcal{M}}
\def\msg{m}

\newcommand{\BinaryField}{\mathbb{F}_{2}}
\newcommand{\InputAlphabet}{\mathcal{X}}
\newcommand{\OutputAlphabet}{\mathcal{Y}}

\newtheorem{thm}{\textbf{Theorem}}
\newtheorem{corollary}{\textbf{Corollary}}
\newtheorem{definition}{\textbf{Definition}}
\newtheorem{lemma}{\textbf{Lemma}}
\newtheorem{example}{\textbf{Example}}
\newtheorem{prop}{\textbf{Proposition}}
\newcommand{\comment}[1]{}
\newcommand{\msout}[1]{\text{\sout{\ensuremath{#1}}}}
\begin{document}

\title{An Achievable rate region for the $3-$user interference channel based on coset codes}

\author {Arun Padakandla, Aria G. Sahebi and
    S. Sandeep Pradhan,
~\IEEEmembership{Member,~IEEE}%
\thanks{The authors are with the Department of Electrical
and Computer Engineering, University of Michigan, Ann Arbor
48109-2122, USA.}
\thanks{This work was supported by NSF grant CCF-1116021.}}

\maketitle

\begin{abstract}
We consider the problem of communication over a three user discrete memoryless interference channel ($3-$IC). The current known coding techniques for communicating over an arbitrary $3-$IC are based on message splitting, superposition coding and binning using independent and identically distributed (iid) random codebooks. In this work, we propose a new ensemble of codes - partitioned coset codes (PCC) - that possess an appropriate mix of empirical and algebraic closure properties. We develop coding techniques that exploit algebraic closure property of PCC to enable interference alignment over general $3-$IC. We analyze the performance of the proposed coding technique to derive an achievable rate region for the general discrete $3-$IC. Additive and non-additive examples are identified for which the derived achievable rate region is the capacity, and moreover, strictly larger than current known largest achievable rate regions based on iid random codebooks.
\end{abstract}

\section{Introduction}
An interference channel (IC) is a model for communication between multiple transmitter receiver (Tx-Rx) pairs that share a common communication medium. Each transmitter wishes to communicate specific information to its corresponding receiver. Since the Tx-Rx pairs share a common communication medium, every user's signal causes interference to every other user. Communication over an IC is therefore facilitated by a coding technique that manages interference efficiently, in addition to combating channel noise.

Carleial proposed the technique of \textit{message splitting via
  superposition coding} \cite{197801TIT_Car} to manage
interference. Carleial's technique is based on each receiver decoding
a part of the interferer's signal and peeling it off to enhance its
ability to decode the desired signal. Han and Kobayashi
\cite{198101TIT_HanKob} enhanced Carleial's technique with joint
decoding and derived an achievable rate region for the IC with two
receivers ($2-$IC) that is the current known largest. This coding
technique and its corresponding achievable rate region will be
referred to as CHK-technique and CHK rate region, respectively. 

More recently, a newer technique of \textit{aligning interference} has
been proposed for managing interference over additive IC with three or
more receivers. The technique of aligning interference is based on
carefully choosing codebooks such that the interfering signals
\textit{align} and appear as if they were coming from a single
user. This technique was proposed for the MIMO X-channel by Maddah
Ali et. al. \cite{200808TIT_AliMotKha}, and for the multi-user IC by
Jafar and Cadambe \cite{200808TIT_CadJaf}. The technique of aligning
interference has subsequently been proposed in several settings
\cite{201009TIT_BreParTse}, \cite{201408TIT_HonCai},
\cite{201407TIT_KriJaf} \cite{6516907} using algebraic codes. 

Our current understanding of interference alignment techniques is
limited in several aspects. Firstly, these techniques are applicable
only to additive IC's. Secondly, from an information theoretic point
of view, the single-letter distributions induced by the codes are uniform, resulting
in achievability of rates corresponding to only uniform
distributions. Thirdly, the particular form of (i) encoding, decoding
(syndrome or lattice) and (ii) the information theoretic tools
constrains us to analyze performance only of additive IC's. 

It is natural to ask whether the technique of interference alignment
is applicable to only additive IC's? More generally, do codes endowed
with structure enable alignment and thereby facilitate communication
over IC's that are not additive? This article addresses these
questions. In particular, we develop a coding technique based on a new
ensemble of codes - partitioned coset codes (PCC) - possessing
algebraic and empirical properties to enable alignment over arbitrary
discrete memoryless IC's with three receivers ($3-$IC). We analyze the
performance of the proposed coding technique to derive a new
achievable rate region for the $3-$IC.

How does the proposed coding technique and the corresponding achievable rate region compare with the current known best? We employ the current known techniques of message \textbf{s}plitting, \textbf{s}uperposition coding and \textbf{b}inning based on \textbf{u}nstructured codes to derive a characterization of $\UHK$region, the current known largest achievable rate region for the general $3-$IC. An important contribution of this article is the identification of additive as well as non-additive instances of $3-$IC for which the proposed coding technique based on PCC yields a strictly larger achievable rate region than the $\UHK$region. We emphasize that our findings for the non-additive instance validates the utility of the theory developed in this article.

The new elements of our work are the following. Firstly, we employ joint typicality encoding and decoding of coset codes to propose alignment techniques for arbitrary $3-$IC's. Secondly, we employ the technique of binning of coset codes to induce arbitrary distributions over corresponding alphabet sets and thereby prove achievability of rates corresponding to arbitrary distributions. Thirdly, we develop coding techniques over looser algebraic objects such as Abelian groups. These elements enable us to derive a new achievable rate region for the general $3-$IC in terms of single-letter information theoretic quantities.

The technique of employing structured codes to obtain larger
achievable rate regions was initiated in the context of a distributed
source coding (DSC) problem by K\"orner and Marton
\cite{197903TIT_KorMar}. Recently, this approach has been employed for
several problem settings.  Philosof and Zamir \cite{200906TIT_PhiZam}
employ coset codes for efficient communication over doubly dirty MACs
and Gaussian version of this problem was studied using lattice codes
in \cite{201108TIT_PhiZamEreKhi}. \cite{200812IEEEI_KocKhiEreZam} and
\cite{200911TIT_KocZam} propose lattice-based schemes for
communicating over Gaussian multi-terminal networks. An achievable
rate region based on Abelian group codes was provided for 
the general DSC problem in \cite{201103TIT_KriPra}.
Linear codes have been employed for
efficient computation over multiple access channels (MAC) in
\cite{200710TIT_NazGas}. In the context of
the interference channel, Maddah-Ali
et. al. \cite{200808TIT_AliMotKha, 200607ISIT_AliMotKha}, and Cadambe
and Jafar \cite{200808TIT_CadJaf} propose the technique of
\textit{interference alignment}, wherein interference is restricted to
a subspace and thereby harness the available of degrees of freedom in
an IC with several Tx-Rx pairs more efficiently. Bresler, Parekh and
Tse \cite{201009TIT_BreParTse} employ lattice codes to align
interference and thereby characterize the capacity of Gaussian ICs
within a constant number of bits. The use of lattice codes has also
been proposed in \cite{200809Allerton_SriJafVisJafSha,
  201109TITarXiv_JafVis,6283726,6516907} for efficient interference management over
Gaussian ICs with three or more Tx-Rx pairs. \cite{201105TIT_BanGam}
considers saturation technique for general ICs. 

This article is organized as follows. In section
\ref{SubSec:NaturalExtensionOfHanKobayashiTo3to1IC}, we characterize
the current known largest achievable rate region based on unstructured
codes for the general
$3-$IC and prove its strict sub-optimality in section
\ref{Sec:StrictSub-optimalityOfHanKobayashiRateRegionFor3-To-1IC}. We
provide new achievable rate regions based on PCC built over fields and groups in
sections \ref{Sec:AchievableRateRegionsFor3To1ICUsingNestedCosetCodes}
and \ref{Sec:AchievableRateRegionsFor3To1ICUsingAbelianGroups},
respectively. We begin with preliminaries in section
\ref{Sec:Preliminaries}. 

\section{Preliminaries: notation and definitions}
\label{Sec:Preliminaries}

\subsection{Notation}
\label{SubSec:Notation}
We let $\naturals, \reals$ denote the set of natural numbers and real numbers, respectively. Calligraphic letters such as $\InputAlphabet,\OutputAlphabet$ exclusively denote finite sets. For $K \in \naturals$, we let $[K]\define \left\{ 1,2\cdots,K \right\}$. In this article, we will need to define multiple objects, mostly triples, of the same type. In order to reduce clutter, we use an \underline{underline} to
denote aggregates of objects of similar type. For example, (i) if
$\OutputAlphabet_{1},\OutputAlphabet_{2}, \OutputAlphabet_{3}$ denote (finite) sets, we
let $\underlineSetY$ either denote the Cartesian product
$\OutputAlphabet_{1} \times \OutputAlphabet_{2} \times \OutputAlphabet_{3}$ or
abbreviate the collection $(\OutputAlphabet_{1},\OutputAlphabet_{2},\OutputAlphabet_{3})$ of sets, the particular reference being clear from context, (ii) if $y_{j} \in \OutputAlphabet_{j}:j\in [3]$, we let $\ulineoutput \in
\underlineSetY$ abbreviate $(y_{1},y_{2},y_{3}) \in \ulineOutputAlphabet$, (iii) if
$d_{j}:\OutputAlphabet_{j}^{n} \rightarrow \MessageSetM_{j}:j\in [3]$ denote (decoding)
maps, then we let $\underlined(\underliney^{n})$ denote $(d_{1}(y_{1}^{n}),
d_{2}(y_{2}^{n}), d_{3}(y_{3}^{n}))$. If $j \in \left\{1,2\right\}$, then $\msout{j} \in \left\{1,2\right\} \setminus \left\{j \right\}$ is the other index. Unless otherwise mentioned, we let $\Prime$ denote an integral power of a prime. Throughout, $\fieldpi$ will denote the finite field of cardinality $\Prime$. $\oplus$ denotes the addition operation in the corresponding finite field. We employ the notion of typicality as in \cite{201301arXivMACDSTx_PadPra}. In particular, if $U,V$ are random variables distributed with respect to $p_{UV}$, then $T_{\eta}(U,V) \in \mathcal{U}^{n}\times \mathcal{V}^{n}$ denotes the typical set with respect to $p_{UV}$ and deviation parameter $\eta$. For any $v^{n} \in \mathcal{V}^{n}$, $T_{\eta}(U|v^{n}) = \left\{ u^{n}:(u^{n},v^{n}) \in  T_{\eta}(U,V)\right\}$ denotes the conditional typical set. $*$ denotes binary convolution, i.e., $\alpha * \beta = \alpha (1-\beta) + (1-\alpha)\beta$. $|a|^{+}$ is defined as $\max\{ 0,a\}$.
\subsection{Definitions: $3-$IC, $3-$to$-1$IC, achievability, capacity region}
\label{Sec:3UserICDefinitions}
A \threeIC consists of three finite input alphabet
sets $\InputAlphabet_{1},\InputAlphabet_{2},\InputAlphabet_{3}$ and three finite output alphabet sets $\OutputAlphabet_{1}, \OutputAlphabet_{2}, \OutputAlphabet_{3}$. The discrete time channel is (i) time invariant, (ii) memoryless, and (iii) used without feedback. Let
$W_{\ulineOutputRV|\ulineInputRV}(\ulineoutput|\ulineinput)=W_{\OutRV_{1}\OutRV_{2}\OutRV_{3}|\InRV_{1}\InRV_{2}\InRV_{3}}(\output_{1},\output_{2},\output_{3}|\inp_{1},\inp_{2},\inp_{2})$ denote
probability of observing symbol $\output_{j} \in \OutputAlphabet_{j}$ at output $j$, given $\inp_{j} \in \InputAlphabet_{j}$ is input by encoder $j$. Inputs are constrained with respect to bounded cost functions $\costfn_{j} : \InputAlphabet_{j} \rightarrow [0, \infty):j \in [3]$. The cost function
is assumed to be additive, i.e., cost of transmitting vector $\inp_{j}^{n} \in \InputAlphabet_{j}^{n}$
is $\nlettercost_{j}(\inp_{j}^{n}) \define \frac{1}{n}\sum_{t=1}^{n}\costfn_{j}(\inp_{jt})$. We refer to this \threeIC as
$(\ulineInputAlphabet,\ulineOutputAlphabet,W_{\ulineOutputRV|\ulineInputRV},\ulinecostfn)$.

\begin{definition}
 \label{Defn:3ICCode}
 A \threeIC code $(n,\underline{\mathcal{M}},\ulinee,\ulined)$ consist of (i) index
sets $\MessageSet_{1}, \MessageSet_{2},\MessageSet_{3}$ of messages, (ii) encoder maps
$\encoder_{j}:\MessageSet_{j} \rightarrow
\InputAlphabet_{j}^{n}:j \in [3]$, and (iii) decoder maps $\decoder_{j} : \OutputAlphabet_{j}^{n}
\rightarrow \MessageSet_{j}:j \in [3]$.
\end{definition}

\begin{definition}
 \label{Defn:ICCodeErrorProbability}
 The error probability of a \threeIC code
$(n,\underline{\mathcal{M}},\ulinee,\ulined)$ conditioned on message triple
$(m_{1},m_{2},m_{3}) \in \ulineMessageSet$ is
\begin{equation}
 \label{Eqn:ErrorProbabilityOf3ICCode}
 \xi(\ulinee,\ulined|\ulinem) \define 1-\sum_{\ulineoutput^{n} : \ulined
(\ulineoutput^{n}) = \ulinem}
W^{n}_{\ulineOutputRV|\ulineInputRV}(\ulineoutput^{n}|\encoder_{1}(\msg_{1}),\encoder_{2}
(\msg_{2}),\encoder_{3}(\msg_{3})).\nonumber
\end{equation}
The average error probability of
a \threeIC code $(n,\underline{\mathcal{M}},\underlinee,\underlined)$ is
$\bar{\xi}(\ulinee,\underlined)\define \sum_{\underlinem
 \in \underlineMessageSet}
\frac{1}{|\underline{\mathcal{M}}|}
\xi(\ulinee,\underlined|\underlinem)$. Average cost per symbol of transmitting message $\underlinem \in
\underlineMessageSetM$ is $\ulinecost(\ulinee|\underlinem) \define
\left(\nlettercost_{j}(\encoder_{j}(\msg_{j})):j\in [3]\right)$ and
average cost per symbol of \threeIC code $(n,\underline{\mathcal{M}},\ulinee,\underlined)$ is $\ulinecost(\ulinee)
\define \frac{1}{|\underline{\mathcal{M}}|}\sum_{\underlinem
 \in \underlineMessageSetM}\ulinecost(\ulinee|\underlinem)$.
\end{definition}

\begin{definition}
\label{Defn:3ICAchievabilityAndCapacity}
A rate-cost sextuple
$(R_{1},R_{2},R_{3},\cost_{1},\cost_{2},\cost_{3})\in [0,\infty)^{6}$
is said to be 
achievable if for every $\eta > 0$, there exists $N(\eta)\in \naturals$ such that for all
$n > N(\eta)$, there exists a \threeIC code $(n, \underline{\mathcal{M}}^{(n)},\ulinee^{(n)},\underlined^{(n)})$ such that (i)
$\frac{\log|\MessageSet_{j}^{(n)}|}{n} \geq R_{j}-\eta:j \in [3]$, (ii)
$\bar{\xi}(\ulinee^{(n)},\underlined^{(n)})
\leq \eta$, and (iii) average cost $\ulinecost(e^{(n)})_{j} \leq \cost_{j}+\eta$. The capacity region is
$\mathbb{C}(\ulinecost) \define {\left\{ \underlineR \in \reals^{3}:
(\underlineR,\ulinecost)\mbox{ is achievable} \right\}}$.
\end{definition}

We now consider \\3to1IC, a class of $3-$IC's that was studied in
\cite{200912arXiv_CadJaf}. \\3to1IC enables us to prove strict
sub-optimality of coding techniques based on unstructured codes. A
\\3to1IC is a \threeIC wherein two of the users enjoy interference
free point-to-point links. Formally, a \threeIC
$(\ulineInputAlphabet,\ulineOutputAlphabet,W_{\ulineOutputRV|\ulineInputRV},\ulinecost)$
is a \\3to1IC if (i)
$W_{Y_{2}|\ulineInputAlphabet}(y_{2}|\ulineinput)\define
\sum_{(y_{1},y_{3}) \in \OutputAlphabet_{1} \times
  \OutputAlphabet_{3}}W_{\ulineOutputRV|\ulineInputRV}(\ulineoutput|\ulineinput)$
is independent of $(x_{1},x_{3}) \in \InputAlphabet_{1} \times
\InputAlphabet_{3}$, and (ii)
$W_{Y_{3}|\ulineInputAlphabet}(y_{3}|\ulineinput)\define
\sum_{(y_{1},y_{2}) \in \OutputAlphabet_{1} \times
  \OutputAlphabet_{2}}W_{\ulineOutputRV|\ulineInputRV}(\ulineoutput|\ulineinput)$
is independent of $(x_{1},x_{2}) \in \InputAlphabet_{1} \times
\InputAlphabet_{2}$ for every collection of input and output symbols $(\ulineinput,\ulineoutput) \in \ulineInputAlphabet \times \ulineOutputAlphabet$. For a \\3to1IC, the channel transition probabilities factorize as \[W_{\ulineOutputRV|\ulineInputRV}(\ulineoutput|\ulineinput)=W_{\OutputRV_{1}|\ulineInputRV}(y_{1}|\ulineinput)W_{\OutputRV_{2}|\InputRV_{2}}(y_{2}|x_{2})W_{\OutputRV_{3}|\InputRV_{3}}(y_{3}|x_{3})\] for some conditional probability mass functions (pmfs) $W_{\OutputRV_{1}|\ulineInputRV}$, $W_{\OutputRV_{2}|\InputRV_{2}}$ and $W_{\OutputRV_{3}|\InputRV_{3}}$. We also note that $X_{1}X_{3}-X_{2}-Y_{2}$ and $X_{1}X_{2}-X_{3}-Y_{3}$ are Markov chains for any distribution $p_{X_{1}}p_{X_{2}}p_{X_{3}}W_{\ulineOutputRV|\ulineInputRV}$.\footnote{Any interference channel wherein only one of the users is subjected to interference is a $3-$to$-1$ IC by a suitable permutation of the user indices.}

In the following section, we describe the coding technique of message splitting and superposition using unstructured codes and employ this to derive the $\UHK$region for $3-$to$-1$ IC.
\ifVerboseICIntro{Our first task is to derive the $\UHK$region which is an achievable rate region for a $3-$to$-1$ IC using all current known coding techniques. The reader will recognize that the structure of a $3-$to$-1$ IC lends the technique of precoding via binning superfluous. 

The currently known largest achievable rate region for \threeIC is obtained by a natural extension of that proved achievable by Han and Kobayashi for \2IC \cite{198101TIT_HanKob}. We begin with a brief review of Han-Kobayashi coding technique for the \2IC in section \ref{SubSec:Han-KobayashiTechniqueFor2IC} and, as a first step, state its natural extension to a \\3to1IC in section \ref{SubSec:NaturalExtensionOfHanKobayashiTo3to1IC}.}\fi
\section{Message splitting and superposition using unstructured codes}
\label{SubSec:NaturalExtensionOfHanKobayashiTo3to1IC}
Before we consider the case of a $3-$to$-1$ IC, it is appropriate to state how does one optimally stitch together current known coding techniques - message splitting, \textbf{s}uperposition coding and precoding via \textbf{b}inning - for communicating over $3-$IC? Each encoder must make available parts of its signal to each user it interferes with. Specifically, encoder $j$ splits its signal into four parts - one public, two semi-private and one private. The corresponding decoder $j$ decodes all of these parts. The other two decoders, say $i$ and $k$, for which encoder $j$'s signal is interference, decode the public part of user $j$'s signal. The public part is decoded by all receivers, and is therefore encoded using a cloud center codebook at the base layer. Moreover, each semi-private part of encoder $j$'s signal is decoded by exactly one among the decoders $i$ and $k$. The semi-private parts are encoded at the intermediate level using one codebook each. These codebooks, referred to as semi-satellite codebooks, are conditionally coded over the cloud center codebook. The semi-satellite codebooks are precoded for each other via binning. The private part is encoded at the top layer using a satellite codebook. The satellite codebook is conditionally coded over the cloud center and semi-satellite codebooks. Each decoder decodes the eight parts using a joint typicality decoder. Finally, the encoders and decoders share a time sharing sequence to enable them to synchronize the choice of codebooks at each symbol interval. We henceforth refer to the above coding technique as the $\UHK$technique.

One can characterize $\UHK$region - an achievable rate region corresponding to the above coding technique - via random coding. Indeed, such a characterization is quite involved. Since our objective is to illustrate sub-optimality of $\UHK$technique, it suffices to obtain a characterization of $\UHK$region for $3-$to$-1$ ICs.

For the case of \\3to1IC, user $1$'s signal does not cause interference to users $2$ and $3$, and therefore will not need it to split its message. This can be proved using the Markov chains $X_{1}X_{3}-X_{2}-Y_{2}$ and $X_{1}X_{2}-X_{3}-Y_{3}$. Moreover, signal of user $2$ does not interfere with user $3$'s reception and vice versa. Therefore, users $2$ and $3$ will only need to split their messages into two parts - a private part and a semi-private part that is decoded by user $1$. Using this approach we obtain the following achievable rate region.

\begin{definition}
 \label{Defn:HanKobayashiTestChannelsFor3To1IC}
Let $\SetOfDistributions_{u}(\ulinecost)$ denote the collection of pmfs
$p_{\TimeSharingRV\SemiPrivateRV_{2}\SemiPrivateRV_{3}\ulineInputRV\ulineOutputRV}$
defined on $\TimeSharingRVSet \times \SemiPrivateRVSet_{2} \times \SemiPrivateRVSet_{3}
\times \ulineInputAlphabet \times \ulineOutputAlphabet$, where $\TimeSharingRVSet,
\SemiPrivateRVSet_{2},\SemiPrivateRVSet_{3}$ are finite sets, such that (i)
$p_{\ulineOutputRV|\ulineInputRV\SemiPrivateRV_{2}\SemiPrivateRV_{3}\TimeSharingRV}=W_{
\ulineOutputRV|\ulineInputRV}$, (ii) the triplet $\InputRV_{1},
(\SemiPrivateRV_{2},\InputRV_{2})$ and $(\SemiPrivateRV_{3},\InputRV_{3})$ are
conditionally mutually independent given $\TimeSharingRV$, (iii) $\mathbb{E}\left\{
\costfn_{j}(X_{j}) \right\} \leq \cost_{j}:j \in [3]$. For
$p_{\TimeSharingRV\SemiPrivateRV_{2}\SemiPrivateRV_{3}\ulineInputRV\ulineOutputRV} \in
\SetOfDistributions_{u}(\ulinecost)$, let
$\alpha_{u}(p_{\TimeSharingRV\SemiPrivateRV_{2}\SemiPrivateRV_{3}
\ulineInputRV\ulineOutputRV})$ 
denote the set of rate triples $(R_{1},R_{2},R_{3}) \in [0,\infty)^{3}$ that satisfy
\begin{eqnarray}
 \label{Eqn:HanKobayashiRateRegion2ICForParticularTestChannel}
\label{Eqn:3To1ICUpperBoundOnR1}&&0 \leq R_{1} < I(X_{1};Y_{1}|\TimeSharingRV,
\SemiPrivateRV_{2},\SemiPrivateRV_{3}), ~~~
0 \leq R_{j} < I(\SemiPrivateRV_{j}X_{j};Y_{j}|\TimeSharingRV):j =2,3 \\
&&R_{1}+R_{2} < I(\SemiPrivateRV_{2}X_{1};Y_{1}|\TimeSharingRV
\SemiPrivateRV_{3})+I(X_{2};Y_{2}|\TimeSharingRV \SemiPrivateRV_{2}),~~~
R_{1}+R_{3} < I(\SemiPrivateRV_{3}X_{1};Y_{1}|\TimeSharingRV
\SemiPrivateRV_{2})+I(X_{3};Y_{3}|\TimeSharingRV \SemiPrivateRV_{3})\nonumber\\
\label{Eqn:SumRateBoundOn3To1IC}&&R_{1}+R_{2}+R_{3} <
I(\SemiPrivateRV_{2}\SemiPrivateRV_{3}X_{1};Y_{1}|\TimeSharingRV)+I(X_{2};Y_{2}
|\TimeSharingRV \SemiPrivateRV_{2})+I(X_{3};Y_{3}|\TimeSharingRV \SemiPrivateRV_{3}),
\end{eqnarray}
and
\begin{equation}
 \label{Eqn:HanKobayashiRateRegionFor3to1IC}
\alpha_{u}(\ulinecost) = \cl \left(
\underset{\substack{p_{\TimeSharingRV\SemiPrivateRV_{2}\SemiPrivateRV_{3}
\ulineInputRV\ulineOutputRV} ~\in ~\SetOfDistributions_{u}(\ulinecost)}
}{\bigcup}\alpha_{u}(p_{\TimeSharingRV\SemiPrivateRV_{2}\SemiPrivateRV_{3}
\ulineInputRV\ulineOutputRV}) \right).\nonumber
\end{equation}
\end{definition}
\begin{thm}
\label{Thm:HanKobayashiRateRegionFor3To1IC}
For \\3to1IC $(\ulineInputAlphabet,\ulineOutputAlphabet,W_{\ulineOutputRV|\ulineInputRV},\ulinecostfn)$,
$\alpha_{u}(\ulinecost)$ is achievable, i.e., $\alpha_{u}(\ulinecost) \subseteq \mathbb{C}(\ulinecost)$.
\end{thm}

\section{Strict sub-optimality of $\UHK$region for $3-$to$-1$ IC}
\label{Sec:StrictSub-optimalityOfHanKobayashiRateRegionFor3-To-1IC}
This section contains our first main finding of this {\ifThesis chapter\fi}{\ifTITJournal article\fi} - strict sub-optimality of $\UHK$technique. In particular, we identify a binary additive \\3to1IC for which we prove strict sub-optimality of $\UHK$technique. We begin with the description of the \\3to1IC. It maybe noted that a similar example was studied in \cite{201009TIT_BreParTse}, wherein CHK technique restricted to Gaussian test channels were shown to be strictly sub-optimal. While our finding is in a similar spirit, our proof takes into account all possible test channels under the CHK technique. In \cite{197903TIT_KorMar}, it was proven that linear codes are strictly more efficient than unstructured codes for the DSC problem.

\begin{example}
\label{Ex:A3To1ICForWhichLinearCodesOutperformUnstructuredCodes}
Consider a binary additive \\3to1IC illustrated in figure \ref{Fig:A3To1ICForWhichLinearCodesOutperformUnstructuredCodes} with $\InputAlphabet_{j}=\OutputAlphabet_{j}=\left\{ 0,1 \right\}:j \in [3]$ with channel transition probabilities $W_{\ulineOutputRV|\ulineInputRV}(\ulineoutput|\ulineinput)=BSC_{\delta_{1}}(y_{1}|x_{1}\oplus x_{2} \oplus x_{3})BSC_{\delta_{2}}(y_{2}|x_{2})BSC_{\delta_{3}}(y_{3}|x_{3})$, where $BSC_{\eta}(0|1)=BSC_{\eta}(1|0)=1-BSC_{\eta}(0|0)=1-BSC_{\eta}(1|1)=\eta$ denotes the transition probabilities of a BSC with cross over probability $\eta \in [0,\frac{1}{2}]$. Inputs of users $2$ and $3$ are not constrained, i.e., $\costfn_{j}(0)=\costfn_{j}(1)=0$ for $j=2,3$. User $1$'s input is constrained with respect to a Hamming cost function, i.e., $\costfn_{1}(x)=x$ for $x \in \left\{ 0,1 \right\}$ to an average cost of $\tau \in (0,\frac{1}{2})$ per symbol. Let $\mathbb{C}(\tau)$ denote the capacity region of this $3-$to$-1$ IC.
\end{example}
\begin{figure}
\centering
\includegraphics[height=1.8in,width=1.9in]{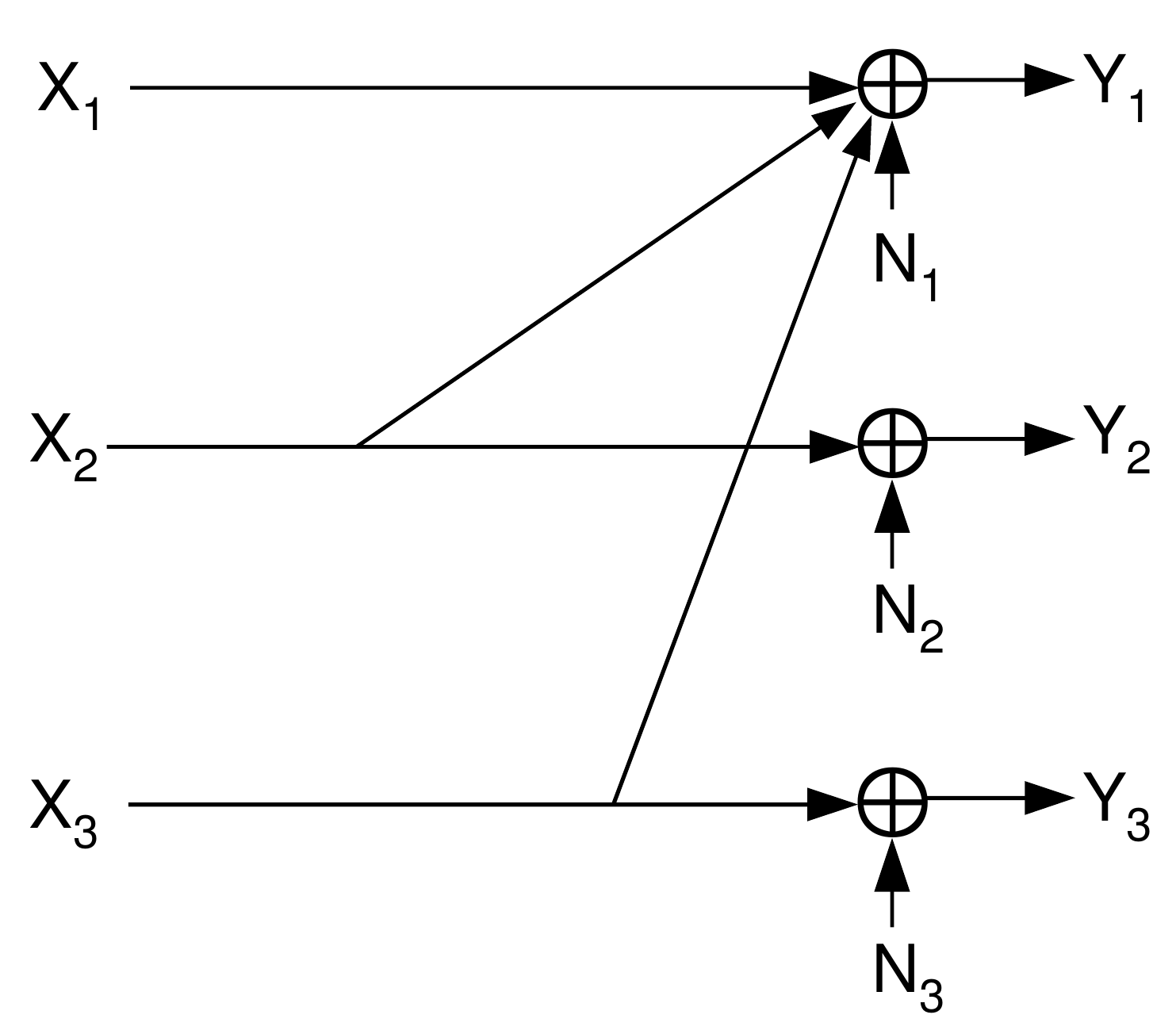}
\caption{A binary additive $3-$to$-1$ IC described in example \ref{Ex:A3To1ICForWhichLinearCodesOutperformUnstructuredCodes}.}
\label{Fig:A3To1ICForWhichLinearCodesOutperformUnstructuredCodes}
\end{figure}

Clearly, $\mathbb{C}(\tau) \subseteq \beta(\tau,\frac{1}{2},\frac{1}{2},\underline{\delta})$, where
\begin{eqnarray}
 \label{Eqn:StrictSubOptimalityOfUSBTechnique3To1ICOuterBound}
 \beta (\underline{\tau},\underline{\delta} ) \define \left\{ (R_{1},R_{2},R_{3}) \in [0,\infty)^{3} : R_{j} \leq h_{b}(\delta_{j}*\tau_{j})-h_{b}(\delta_{j}):j=1,2,3  \right\}.
\end{eqnarray}Let us focus on achievability. We begin with a few simple observations for the above channel. Let us begin with the assumption $\delta\define \delta_{2}=\delta_{3}$. As illustrated in figure \ref{Fig:A3To1ICForWhichLinearCodesOutperformUnstructuredCodes}, users $2$ and $3$ enjoy interference free unconstrained binary symmetric channels (BSC) with cross over probability $\delta=\delta_{2}=\delta_{3}$. They can therefore communicate at their respective capacities $1-h_{b}(\delta)$. Constrained to average Hamming weight of $\tau$, user $1$ cannot hope to achieve a rate larger than $h_{b}(\tau * \delta_{1})-h_{b}(\delta_{1})$.\footnote{If receiver $1$ is provided with the codewords transmitted by users $2$ and $3$, the effective channel it sees is a BSC with cross over probability $\delta_{1}$.} What is the maximum rate achievable by user $1$ while users $2$ and $3$ communicate at their respective capacities?

User $1$ cannot achieve rate $h_{b}(\tau * \delta_{1})-h_{b}(\delta_{1})$
\textit{and} decode the pair of codewords transmitted by user $2$ and $3$ if $h_{b}(\tau *
\delta_{1})-h_{b}(\delta_{1})+2(1-h_{b}(\delta)) > 1-h_{b}(\delta_{1})$ or equivalently
$1+h_{b}(\tau * \delta_{1}) > 2h_{b}(\delta)$. Under this condition, $\UHK$technique forces decoder $1$ to be contented to decoding univariate components -
represented through semi-private random variables $\SemiPrivateRV_{2},\SemiPrivateRV_{3}$
- of user $2$ and $3$'s signals. We state that as long as the univariate components
leave residual uncertainty in the interfering signal, i.e., $H(X_{2}\oplus
X_{3}|U_{2},U_{3}) > 0$, the rate achievable by user 1 is strictly smaller than its
maximum $h_{b}(\tau * \delta_{1})-h_{b}(\delta_{1})$.\footnote{The reader will be
able to reason this by relating this situation to a point-to-point (PTP) channel with partial
state observed at the receiver.} 

We now describe a simple linear coding technique, based on the works of \cite{200912arXiv_CadJaf, 201408TIT_HonCai, 200809Allerton_SriJafVisJafSha}, that enables user $1$ to achieve its maximum rate $h_{b}(\tau * \delta_{1})-h_{b}(\delta_{1})$ even under the condition $1+h_{b}(\tau * \delta_{1}) > 2h_{b}(\delta)$! Let us assume $\tau*\delta_{1} \leq \delta$. We choose a linear code, or a coset thereof, that achieves the capacity of a BSC with cross over probability $\delta$. We equip users $2$ and $3$ with the same code, thereby constraining the sum of their transmitted codewords to this linear code, or a coset thereof, of rate $1-h_{b}(\delta)$. Since $\tau*\delta_{1} \leq \delta$, decoder $1$ can first decode the interfering signal - sum of codewords transmitted by encoders $2$ and $3$ - treating the rest as noise, peel it off, and then decode the desired signal. User $1$ can therefore achieve its maximum rate $h_{b}(\tau * \delta_{1})-h_{b}(\delta_{1})$ if $\tau*\delta_{1} \leq \delta$. 

In proposition \ref{Prop:StrictSub-OptimalityOfHanKobayashi}, we prove that if $1+h_{b}(\delta_{1} * \tau) > h_{b}(\delta_{2})+h_{b}(\delta_{3})$, then $(h_{b}(\tau*\delta_{1})-h_{b}(\delta_{1}),1-h_{b}(\delta_{2}),1-h_{b}(\delta_{3})) \notin \alpha_{u}((\tau,0,0))$. We therefore conclude that if $\tau,\delta_{1},\delta_{2},\delta_{3}$ are such that $1+h_{b}(\delta_{1} * \tau) > h_{b}(\delta_{2})+h_{b}(\delta_{3})$ and $\min\left\{ \delta_{2},\delta_{3}\right\}\geq \delta_{1} * \tau$, then $\UHK$technique is strictly suboptimal for the \\3to1IC presented in example \ref{Ex:A3To1ICForWhichLinearCodesOutperformUnstructuredCodes}.

\begin{prop}
 \label{Prop:StrictSub-OptimalityOfHanKobayashi}
For the \\3to1IC of example \ref{Ex:A3To1ICForWhichLinearCodesOutperformUnstructuredCodes}, if $\tau*\delta_{1} \leq \min\left\{\delta_{2},\delta_{3}\right\}$, then $\mathbb{C}(\tau)= \beta(\tau,\frac{1}{2},\frac{1}{2},\underline{\delta})$, where $\beta(\underline{\tau},\underline{\delta})$ is given by (\ref{Eqn:StrictSubOptimalityOfUSBTechnique3To1ICOuterBound}). If $h_{b}(\delta_{2})+h_{b}(\delta_{3}) < 1+h_{b}(\tau * \delta_{1})$, then $(h_{b}(\tau*\delta_{1})-h_{b}(\delta_{1}),1-h_{b}(\delta_{2}),1-h_{b}(\delta_{3})) \notin \alpha_{u}(\tau,0,0)$.
\end{prop}
Please refer to appendix \ref{AppSec:SubOptAdditiveIC} for a proof.
In particular, if $\delta_{1}=0.01$ and $\delta_{2} \in (0.1325,0.21)$, then $\alpha_{u}(\frac{1}{8},0,0) \neq \mathbb{C}(\frac{1}{8})$.

\section{Achievable rate region using PCC built over finite fields}
\label{Sec:AchievableRateRegionsFor3To1ICUsingNestedCosetCodes}

\begin{comment}{
How do structured codes enable one to decode the bivariate interference component more efficiently? This is best described by the use of linear codes in decoding the sum interference component. Let us assume that codebooks of users $2$ and $3$ are built over the the binary field $\BinaryField$ and the sum of user $2$ and $3$ codewords is the interference component at receiver $1$. If user $2$ and $3$ build independent codebooks of rate $R_{2}$ and $R_{3}$ respectively, the range of interference patterns has rate $R_{2}+R_{3}$. Enabling user $1$ decode the sum interference pattern constrains the sum $R_{2}+R_{3}$. Instead if codebooks of users $2$ and $3$ are sub-codes of a common linear code, then enabling user $1$ to decode the sum will only constrain $\max\left\{ R_{2},R_{3}\right\}$.
}\end{comment}

\ifTITJournal{In this section we present our second main finding - a
  new achievable rate region for $3-$IC - in the context of finite
  fields. In other words, we propose a coding technique based on PCC
  built over finite fields. Characterizing its information-theoretic
  performance enables us to derive an achievable rate region,
  henceforth referred to as PCC-region.\footnote{We employ the same
    terminology for the rate region achievable using PCC built over
    Abelian groups in section
    \ref{Sec:AchievableRateRegionsFor3To1ICUsingAbelianGroups}.} We
  derive PCC rate region in three pedagogical steps. In the first
  step, presented in section
  \ref{SubSec:DecodingSumOfTransmittedCodewordsUsingNestedCosetCodes},
  we employ PCC to manage interference seen by only one of the
  receivers. This simplified setting aids the reader recognize and
  absorb all the key elements of the framework proposed herein. For
  this step, we provide a complete proof of achievability. In this
  section, we also identify a \textit{non-additive} $3-$to$-1$ IC
  (Example \ref{Ex:A3To1-OR-IC}) for which we \textit{analytically}
  prove (i) strict sub-optimality of $\UHK$technique and (ii)
  optimality of PCC rate region. We provide several examples that illustrate the
  central theme of this article - codes endowed with algebraic closure
  properties enable interference alignment over arbitrary $3-$ICs, not just additive, symmetric instances - and thereby justifies the framework developed herein.

In the second step, presented in section \ref{Sec:AchievableRateRegion3ICUsingNestedCosetCodes}, we employ PCC to manage interference seen by every receiver and thereby provide a characterization of PCC rate region. In the third step we provide a unification of PCC rate region and $\UHK$ rate region along the lines of \cite[Section VI]{198305TIT_AhlHan}.}\fi

\ifThesis{\subsection{Step I : Decoding sum of codewords chosen from PCC over an arbitrary $3-$IC}\label{SubSec:DecodingSumOfTransmittedCodewordsUsingNestedCosetCodes}}\fi
\ifTITJournal{\subsection{Step I : Managing interference seen by one receiver using PCC built over fields}\label{SubSec:DecodingSumOfTransmittedCodewordsUsingNestedCosetCodes}}\fi

%The linear coding technique proposed for example
%\ref{Ex:A3To1ICForWhichLinearCodesOutperformUnstructuredCodes} seems
%to hinge on the additive nature of the channel therein. One of our
%main contributions is in being able to generalize this technique to
%arbitrary channels. 

\begin{definition}
 \label{Eqn:TestChannelsCodingOver3To1ICUsingNestedCosetCodes}
Let $\SetOfDistributions_{f}(\ulinecost)$ denote the collection of
distributions
$p_{\TimeSharingRV\SemiPrivateRV_{2}\SemiPrivateRV_{3}\ulineInputRV\ulineOutputRV} \in
\SetOfDistributions_{u}(\ulinecost)$ defined over $\TimeSharingRVSet \times \SemiPrivateRVSet_{2} \times \SemiPrivateRVSet_{3}  \times \underline{\InputAlphabet} \times \underline{\OutputAlphabet}$, such that $\SemiPrivateRVSet_{2}=\SemiPrivateRVSet_{3}$ is a finite field. For
$p_{\TimeSharingRV\SemiPrivateRV_{2}\SemiPrivateRV_{3}\ulineInputRV\ulineOutputRV} \in
\SetOfDistributions_{f}(\ulinecost)$, let
$\alpha^{\three-1}_{f}(p_{\TimeSharingRV\SemiPrivateRV_{2}\SemiPrivateRV_{3}
\ulineInputRV\ulineOutputRV})$ be defined as the set of rate triples $(R_{1},R_{2},R_{3})
\in [0,\infty)^{3}$ that satisfy
\begin{eqnarray}
R_1 &<& \min \{0, H(U_{j}|\TimeSharingRV)-H(U_{2} \oplus U_{3}|\TimeSharingRV
Y_1):j=2,3\} + I(X_1;U_{2} \oplus U_{3},Y_1|\TimeSharingRV ),\nonumber\\R_j &<&
I(U_{j},X_j;Y_j|\TimeSharingRV ):j=2,3,   \nonumber\\
R_1+R_j &<& I(X_j;Y_j|\TimeSharingRV U_{j}) +I(X_1;U_{2}\oplus
U_{3},Y_1|\TimeSharingRV )+H(U_{j}|\TimeSharingRV )   - H(U_{2}\oplus
U_{3}|\TimeSharingRV Y_1):j=2,3,\nonumber
\end{eqnarray}
and
\begin{equation}
 \label{Eqn:AchievableRateRegion3to1ICUsingNestedCosetCodes}
\alpha^{\three-1}_{f}(\ulinecost) = \cocl \left(
\underset{\substack{p_{\TimeSharingRV\SemiPrivateRV_{2}\SemiPrivateRV_{3}
\ulineInputRV\ulineOutputRV} \in \SetOfDistributions_{f}(\ulinecost)}
}{\bigcup}\alpha_{f}^{\three-1}(p_{\TimeSharingRV\SemiPrivateRV_{2}\SemiPrivateRV_{3}
\ulineInputRV\ulineOutputRV}) \right).\nonumber
\end{equation}
\end{definition}
\begin{thm}
\label{Thm:AchievableRateRegionFor3To1ICUsingCosetCodes}
For $3-$IC $(\ulineInputAlphabet,\ulineOutputAlphabet,W_{\ulineOutputRV|\ulineInputRV},\ulinecostfn)$,
$\alpha^{\three-1}_{f}(\ulinecost)$ is achievable, i.e., $\alpha_{f}^{\three-1}(\ulinecost) \subseteq \mathbb{C}(\ulinecost)$.
\end{thm}
Note that $\alpha_{f}^{3-1}(\underline{\tau})$ is a continuous function of the $3-$IC
$(\underline{\mathcal{X}},\underline{\mathcal{Y}},W_{\underline{Y}|\underline{X}},\kappa)$.
A proof is provided in appendix \ref{AppSec:ProofOfAch}. 
Here we provide a simplified description of the coding technique. Towards that end, consider a pmf $p_{\TimeSharingRV \SemiPrivateRV_{2}
\SemiPrivateRV_{3}\ulineInputRV \ulineOutputRV} \in
\SetOfDistributions_{f}(\ulinecost)$ with $\TimeSharingRVSet=\phi$\footnote{Since the time sharing random variable $\TimeSharingRV$ is employed in a standard way, we choose to omit it in this description.} and $\SemiPrivateRVSet_{2}=\SemiPrivateRVSet_{3}=\fieldpi$. \ifThesis Except for key differences, the coding technique proposed
herein is identical to $\UHK$technique for the \\3to1IC (section \ref{SubSec:NaturalExtensionOfHanKobayashiTo3to1IC}). Let us revisit the linear coding technique proposed for example \ref{Ex:A3To1ICForWhichLinearCodesOutperformUnstructuredCodes} to identify these key differences. Note that the structure and encoding rule of user $1$ in example \ref{Ex:A3To1ICForWhichLinearCodesOutperformUnstructuredCodes} and section \ref{SubSec:NaturalExtensionOfHanKobayashiTo3to1IC} are identical. We therefore employ the same code structure and encoding rules for user $1$. In particular, e\fi\ifTITJournal{E}\fi ncoder $1$ builds a single codebook $\mathcal{C}_{1}=(x_{1}^{n}(\msg_{1}):\msg_{1} \in \MessageSet_{1})$ of rate $R_{1}$ over $\InputAlphabet_{1}$ and the codeword indexed by the message is transmitted on the channel. 

The structure and encoding rules for users $2$ and $3$ are identical and we describe it
using a generic index $j \in \left\{2,3\right\}$. As in section
\ref{SubSec:NaturalExtensionOfHanKobayashiTo3to1IC}, we employ a two layer - cloud center
and satellite - code for user $j$ and split its message $M_{j} \in \MessageSet_{j}$ into
two parts. Let (i) $M_{j1} \in \MessageSet_{j1} \define [\Prime^{t_{j}}]$ denote its
semi-private part, and (ii) $M_{j\InputRV} \in \MessageSet_{j \InputRV} \define
[\exp\{nL_{j}\}]$ denote its private part. While in section \ref{SubSec:NaturalExtensionOfHanKobayashiTo3to1IC} user $1$ decoded the pair of cloud
center codewords, the first
key difference we propose is that user $1$ decodes the sum of user $2$ and $3$ cloud center
codewords. Let a coset $\lambda_{j} \subseteq \SemiPrivateRVSet_{j}^{n}$ of a linear code
$\overline{\lambda}_{j} \subseteq \SemiPrivateRVSet_{j}^{n}$ denote user $j$'s cloud
center codebook.\footnote{The use of a coset code instead of a linear
  code enables ease of analysis. In particular, the key property of
  statistical pairwise independence \cite{Gal-ITRC68} of distinct codewords of randomly
  chosen coset codes is facilitated by choosing a random bias
  shift. This is employed in the many proof elements, for example that
  of lemma
  \ref{Lem:3to1ICAnyPairOfCodewordsInTheTwoCloudCenterCodewordsAreIndependent}.}
In particular, let $g_{j} \in \SemiPrivateRVSet_{j}^{s_{j} \times n}$ 
denote generator matrix of $\overline{\lambda}_{j}$ and coset $\lambda_{j}$ correspond to
shift $b_{j}^{n} \in \SemiPrivateRVSet_{j}^{n}$. \ifThesis{The second key difference we propose is
that }\fi\ifTITJournal{We let the }\fi cloud center codebooks of users' $2$ and $3$ overlap, i.e., the larger of
$\overline{\lambda}_{2},\overline{\lambda}_{3}$ contains the other. For
example, if $s_{j_{2}} \leq s_{j_{3}}$,
then $\overline{\lambda}_{j_{2}}\subseteq \overline{\lambda}_{j_{3}}$. We therefore let
$g_{j_{3}}^{T} = \left[ g_{j_{2}}^{T} ~~g_{j_{3}/j_{2}}^{T} \right]$.

Since codewords of a uniformly distributed coset code are uniformly distributed, we need
to partition the coset code $\lambda_{j}$ into $\Prime^{t_{j}}$ bins to induce a non-uniform distribution over the auxiliary alphabet $\SemiPrivateRVSet_{j}$. \ifThesis{We therefore employ partitioned coset codes (section \ref{Sec:PartitionedCosetPTP_STxCodes}). The third key difference is therefore a
partition of $\lambda_{j}$ into $\Prime^{t_{j}}$ bins to enable one to induce a non-uniform
distribution. For the benefit of a reader who has not studied through section \ref{Sec:PartitionedCosetPTP_STxCodes}, we describe and define partitioned coset codes again. }\fi In particular, for each codeword $u_{j}^{n}(a^{s_{j}})\define
a^{s_{j}}g_{j}\oplus b_{j}^{n}$, where $a^{s_{j}} \in \SemiPrivateRVSet_{j}^{s_{j}}$, a binning function $i_{j}(a^{s_{j}}) \in [\Prime^{t_{j}}]$ is defined that indexes the bin containing
$u_{j}^{n}(a^{s_{j}})$. We let $c_{j1}(m_{j1}) = \{ a^{s_{j}} \in \mathcal{U}_{j}^{s_{j}}:i_{j}(a^{s_{j}})=m_{j1} \}$ denote the set containing indices corresponding to message $m_{j1}$. The structure of the cloud center codebook plays an important role \ifThesis{in this chapter }\fi and we formalize the same through the following definition.

\begin{definition}
 \label{Defn:PartitionedCosetCodes}
{\ifThesis{Recall that a coset code $\lambda \subseteq \fieldpi^{n}$
    is a coset of a linear code $\overline{\lambda} \subseteq
    \fieldpi^{n}$.}\fi} A coset code $\lambda$ is completely specified
by the generator matrix $g \in \fieldpi^{k \times n}$ and a bias
vector $b_{j}^{n} \in \fieldpi^{n}$. Consider a partition of $\lambda$
into $\Prime^{l}$ bins. Each codeword $a^{k}g \oplus b^{n}$ is
assigned an index $i(a^{k}) \in [\Prime^{l}]$. This coset code
$\lambda$ with its partitions is referred to as an $(n,k,l,g,b^{n},i)$ \textit{partitioned coset code} (PCC) or succinctly as an $(n,k,l)$ PCC. For each $m \in [\Prime^{l}]$, let $c(m) \define \left\{ a^{k} \in \fieldpi^{k}:i(a^{k})=m\right\}$.
\end{definition}

User $j$'th satellite codebook $\mathcal{C}_{j}$, built over $\InputAlphabet_{j}$, consists of $\exp\{ nL_{j} \}$ bins, one for each private message $m_{j\InputRV} \in \mathcal{M}_{jX} \define [\exp\{ nL_{j} \}]$. Let $(x_{j}^{n}(\msg_{j\InputRV},b_{jX}) \in \InputAlphabet^{n}_{j}:b_{jX}
\in  [\exp\{nK_{j}\}])$ denote bin corresponding to message $m_{j\InputRV} \in \mathcal{M}_{jX}$ and let $c_{jX} \define [\exp\{nK_{j}\}]$. Having received message $M_{j}=(M_{j1},M_{jX})$, the encoder identifies all pairs
$ (u_{j}^{n}(a^{s_{j}}),x_{j}^{n}(M_{j\InputRV},b_{jX}))$ of jointly typical codewords with $(a^{s_{j}},b_{jX}) \in c_{j1}(M_{j1})\times c_{jX}$. If it finds one or more such pairs, one of them is chosen and the corresponding satellite codeword is fed as input on the channel. Otherwise, an error is declared.

We now describe the decoding rule. Predictably, the decoding rules of users
$2$ and $3$ are identical and we describe this through a generic index $j \in \left\{
2,3 \right\}$. \ifThesis{Except for a slight modification to handle the bins in the codebooks, decoder $j$'s operation is identical in spirit to a point to point decoder
described in section \ref{SubSec:NaturalExtensionOfHanKobayashiTo3to1IC}. Specifically,
d}\fi\ifTITJournal{D}\fi ecoder $j$ identifies all
$(\hat{\msg}_{j1},\hat{\msg}_{j\InputRV})$ for which there exists
$(a^{s_{j}},b_{jX}) \in  c_{j1}(\hat{m}_{j1}) \times c_{jX}$ such that
$(u_{j}^{n}(a^{s_{j}}),x_{j}^{n}(\hat{\msg}_{j\InputRV},b_{jX}),Y_{j}^{n})$ is jointly
typical with
respect to $p_{\SemiPrivateRV_{j}\InputRV_{j},\OutputRV_{j}}$. If
there is exactly one such pair
$(\hat{\msg}_{j1},\hat{\msg}_{j\InputRV})$, this is declared as
the message of user $j$. Otherwise an error is signaled.

Decoder $1$ constructs the sum $\lambda_{2}\oplus \lambda_{3} \define \left\{ u_{2}^{n}
\oplus u_{3}^{n}:u_{j}^{n} \in \lambda_{j} ,j =2,3\right\}$ of the cloud center codebooks.
Having received $Y_{1}^{n}$, it looks for all potential message $\hat{m}_{1}$ for which
there exists a $u_{\oplus}^{n} \in \lambda_{2} \oplus \lambda_{3}$ such that
$(u_{\oplus}^{n},x_{1}^{n}(\hat{m}_{1}),Y_{1}^{n})$ is jointly typical with
respect to
$p_{\SemiPrivateRV_{2}\oplus \SemiPrivateRV_{3},X_{1},Y_{1}}$. If it finds exactly
one such message $\hat{m}_{1}$, it declares this as the decoded message of user $1$.
Otherwise, it declares an error.

We characterize the performance of the proposed coding technique in the proof by averaging over the ensemble of codebooks. Since the distribution induced on the codebooks is such that codebooks of users $2$ and $3$ are statistically correlated and moreover, contain correlated codewords, this involves new elements.

The coding technique proposed in the proof of theorem \ref{Thm:AchievableRateRegionFor3To1ICUsingCosetCodes} is indeed a generalization of that proposed for example \ref{Ex:A3To1ICForWhichLinearCodesOutperformUnstructuredCodes}, and moreover capacity achieving for the same. We formalize this through the following corollary.

\begin{corollary}
 \label{Cor:PCCAchievesCapacityForAdditive3-To-1IC}
 For the $3-$to$-1$ IC in example \ref{Ex:A3To1ICForWhichLinearCodesOutperformUnstructuredCodes}, if $\tau * \delta_{1} < \min \{ \delta_{2},\delta_{3}\}$, then $\alpha_{f}^{3-1}(\tau,\frac{1}{2},\frac{1}{2})=\mathbb{C}(\tau)$.\ifThesis Moreover, if $\delta\define\delta_{2}=\delta_{3}$ and $h_{b}(\tau*\delta_{1}) \leq h_b(\delta) < \frac{1+h_b(\delta_{1}*\tau)}{2}$, then $\alpha_{u}(\tau,\frac{1}{2},\frac{1}{2}) \neq \mathbb{C}(\tau)$ and $ \mathbb{C}(\tau)= \alpha_{f}^{3-1}(\tau,\frac{1}{2},\frac{1}{2})$.\fi
\end{corollary}

It can be verified that $\beta(\tau,\frac{1}{2},\frac{1}{2},\underline{\delta}) = \alpha_{f}^{3-1}(p_{QU_{2}U_{3}\underline{X}\underline{Y}})$ where $P(U_{j}=X_{j}=0)=P(U_{j}=X_{j}=1)=\frac{1}{2}$, $P(X_{1}=1)=\tau$ and $\mathcal{Q}=\phi$, the empty set, where $\beta(\underline{\tau},\underline{\delta})$ is given in (\ref{Eqn:StrictSubOptimalityOfUSBTechnique3To1ICOuterBound}).

In the sequel, we illustrate through three examples the central claim
of this \ifThesis{thesis }\fi\ifTITJournal{article }\fi that the utility of codes endowed with algebraic structure, and in particular coset codes, are not restricted to particular symmetric and additive problems. Furthermore, these examples establish the need (i) to achieve rates corresponding to non-uniform distributions which is accomplished via the technique of binning, (ii) to build coset codes over larger fields, and (iii) to analyze decoding of sums of transmitted codewords over arbitrary channels using typical set decoding.

\begin{example}
\label{Ex:A3To1-OR-IC}
Consider a binary \\3to1IC illustrated in figure \ref{Fig:A3To1ORIC} with $\InputAlphabet_{j}=\OutputAlphabet_{j}=\left\{ 0,1 \right\}:j \in [3]$ with channel transition probabilities $W_{\ulineOutputRV|\ulineInputRV}(\ulineoutput|\ulineinput)=BSC_{\delta_{1}}(y_{1}|x_{1}\oplus (x_{2} \vee x_{3}))BSC_{\delta_{2}}(y_{2}|x_{2})BSC_{\delta_{3}}(y_{3}|x_{3})$, where $\vee$ denotes logical OR.\footnote{$BSC(\cdot|\cdot)$ has been defined in example \ref{Ex:A3To1ICForWhichLinearCodesOutperformUnstructuredCodes}.} Users' inputs are constrained with respect to a Hamming cost function, i.e., $\costfn_{j}(x)=x$ for $x \in \left\{ 0,1 \right\}$, and user $j$th input is constrained to an average cost per symbol of $\tau_{j} \in (0,\frac{1}{2})$ for $j \in [3]$.
\end{example}
% \begin{figure}
% \centering
% \includegraphics[height=1.8in,width=1.9in]{3To1ORIC}
% \caption{A binary $3-$to$-1$ IC described in example \ref{Ex:A3To1-OR-IC}.}
% \label{Fig:A3To1ORIC}
% \end{figure}

Clearly, for the above example, $X_{2}\vee X_{3}$ is the interfering pattern. If $X_{2}$ and $X_{3}$ are viewed as elements in the ternary field, then observe that  $H(X_{2}\vee X_{3}|X_{2}\oplus_{3} X_{3})=0$. The decoder $1$ can reconstruct the interfering pattern after having decoded the ternary sum of the codewords. This motivates the use of \textit{coset} codes for decoding of non-additive interference. 

\begin{prop}
 \label{Prop:ORICStrictSub-OptimalityOfHK}
Consider the \\3to1IC described in example \ref{Ex:A3To1-OR-IC} with $\delta \define \delta_{2}=\delta_{3} \in (0,\frac{1}{2})$ and $\tau \define \tau_{2}=\tau_{3} \in (0,\frac{1}{2})$. Let $\beta \define \delta_{1}*(2\tau -\tau^{2})$. If
\begin{eqnarray}
 \label{Eqn:3To1-OR-ICConditionForAchievabilityInProp}
h_{b}(\tau * \delta)-h_{b}(\delta)\leq  \theta ,
\end{eqnarray}
where $\theta =h_{b}(\tau)-h_{b}((1-\tau)^{2})-(2\tau-\tau^{2})h_{b}(\frac{\tau^{2}}{2\tau-\tau^{2}})-h_{b}(\tau_{1} * \delta_{1})+h_{b}(\tau_{1}*\beta)$, then $\beta(\underline{\tau},\underline{\delta}) = \mathcal{C}(\underline{\tau})=\alpha^{3-1}_{f}(\underline{\tau})$. Moreover, the rate triple $(h_{b}(\tau_{1} * \delta_{1})-h_{b}(\delta_{1}),h_{b}(\tau * \delta)-h_{b}(\delta),h_{b}(\tau * \delta)-h_{b}(\delta)) \notin \alpha_{u}(\ulinecost)$ if
\begin{equation}
 \label{Eqn:3To1-OR-ICConditionForStrictSubOptimalityOfHKInProp}
h_{b}(\tau_{1} * \delta_{1})-h_{b}(\delta_{1})+2(h_{b}(\tau * \delta)-h_{b}(\delta))> h_{b}(\tau_{1}(1-\beta)+(1-\tau_{1})\beta)-h_{b}(\delta_{1}).
\end{equation}
Therefore, if (\ref{Eqn:3To1-OR-ICConditionForAchievabilityInProp}) and (\ref{Eqn:3To1-OR-ICConditionForStrictSubOptimalityOfHKInProp}) hold, $\alpha_{u}(\ulinecost) \subsetneq \alpha^{3-1}_{f}(\underline{\tau}) = \mathcal{C}(\underline{\tau})$.\footnote{The reader is reminded that $\alpha_{u}(\ulinecost)$ is defined in definition \ref{Defn:HanKobayashiTestChannelsFor3To1IC}.}
\end{prop}

Please refer to appendix \ref{AppSec:UpperBoundOnUSBRateRegionFor3To1ORIC} for a proof. Conditions (\ref{Eqn:3To1-OR-ICConditionForAchievabilityInProp}) and (\ref{Eqn:3To1-OR-ICConditionForStrictSubOptimalityOfHKInProp}) are \textit{not} mutually exclusive. It maybe verified that the choice $\tau_{1}=\frac{1}{90}$, $\tau=0.15$, $\delta_{1}=0.01$ and $\delta=0.067$ satisfies both conditions, thereby establishing the utility of structured codes for examples well beyond particular additive ones.

A skeptical reader will wonder whether the utility of PCC depends crucially on the additive multiple access channel (MAC) $Y_{1}=X_{1}\oplus (X_{2}\vee X_{3})\oplus N_{1}$. The following example provides conclusive evidence that this is indeed not the case.

\begin{example}
 \label{Ex:3To1ORICCoupledThroughNonAdditiveMAC}
Consider a binary \\3to1IC illustrated in figure
\ref{Fig:3To1ORICWithNonAdditiveMAC} with
$\InputAlphabet_{j}=\OutputAlphabet_{j}=\left\{ 0,1 \right\}:j \in
[3]$ with channel transition probabilities
$W_{\ulineOutputRV|\ulineInputRV}(\ulineoutput|\ulineinput)=MAC(y_{1}|x_{1},
x_{2} \vee x_{3})BSC_{\delta}(y_{2}|x_{2})BSC_{\delta}(y_{3}|x_{3})$,
where $MAC(0|0,0)=0.989, MAC(0|0,1)=0.01, MAC(0|1,0)=0.02,
MAC(0|1,1)=0.993$ and $MAC(0|b,c)+MAC(1|b,c)=1$ for each $(b,c) \in \{
0,1 \}^{2}$. Users' inputs are constrained with respect to a Hamming
cost function, i.e., $\costfn_{j}(x)=x$ for $x \in \left\{ 0,1
\right\}$. Assume that user $j$th input is constrained to an average cost per symbol of $\tau_{j} \in (0,\frac{1}{2})$, where $\tau\define\tau_{2}=\tau_{3}$.
\end{example}
\begin{figure}
\begin{minipage}{.5\textwidth}
\centering\includegraphics[height=1.8in,width=1.9in]{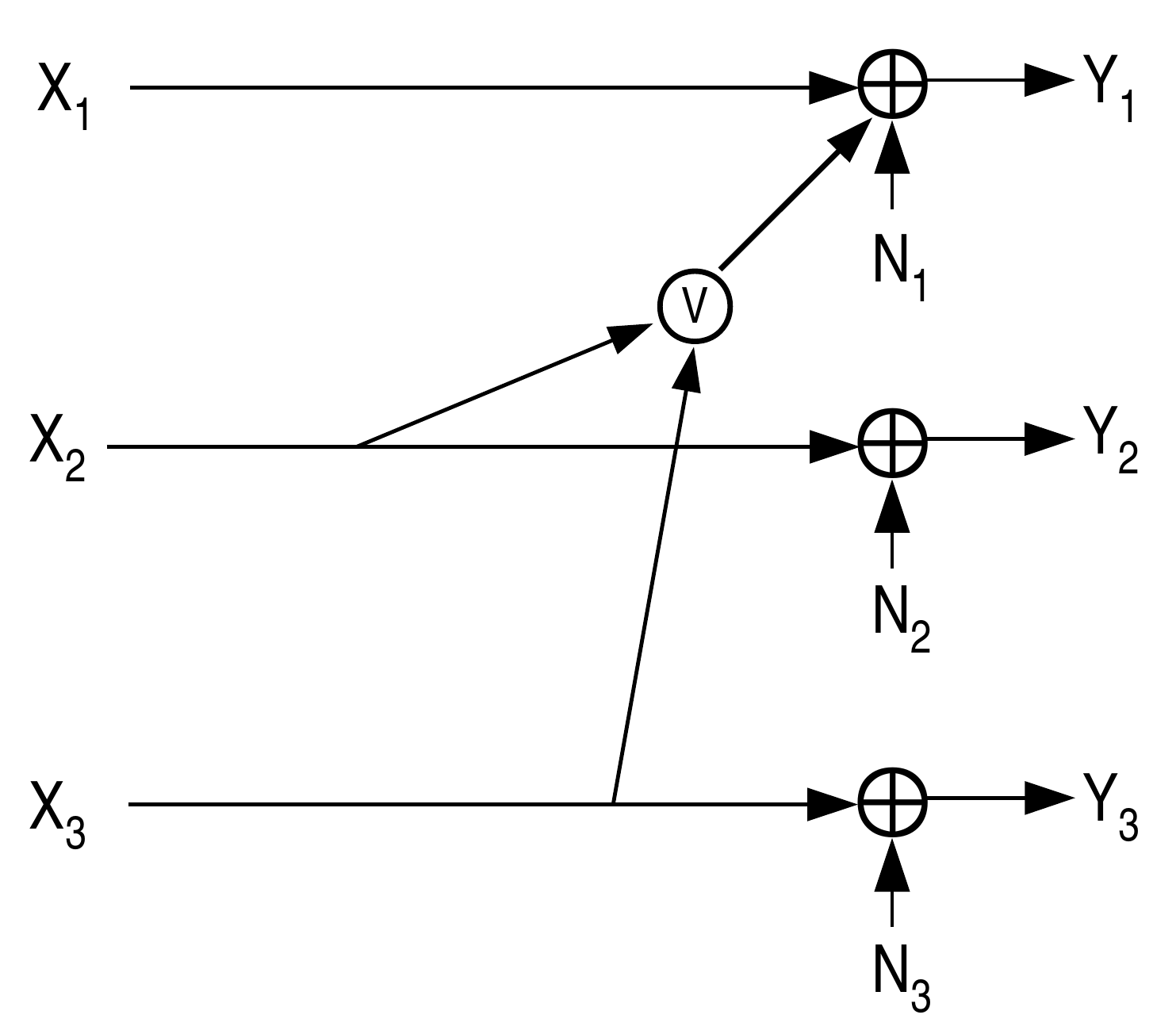}
\caption{A binary $3-$to$-1$ IC described in example \ref{Ex:A3To1-OR-IC}.}
\label{Fig:A3To1ORIC}
\end{minipage}
\begin{minipage}{.5\textwidth}
\centering\includegraphics[height=1.8in,width=1.9in]{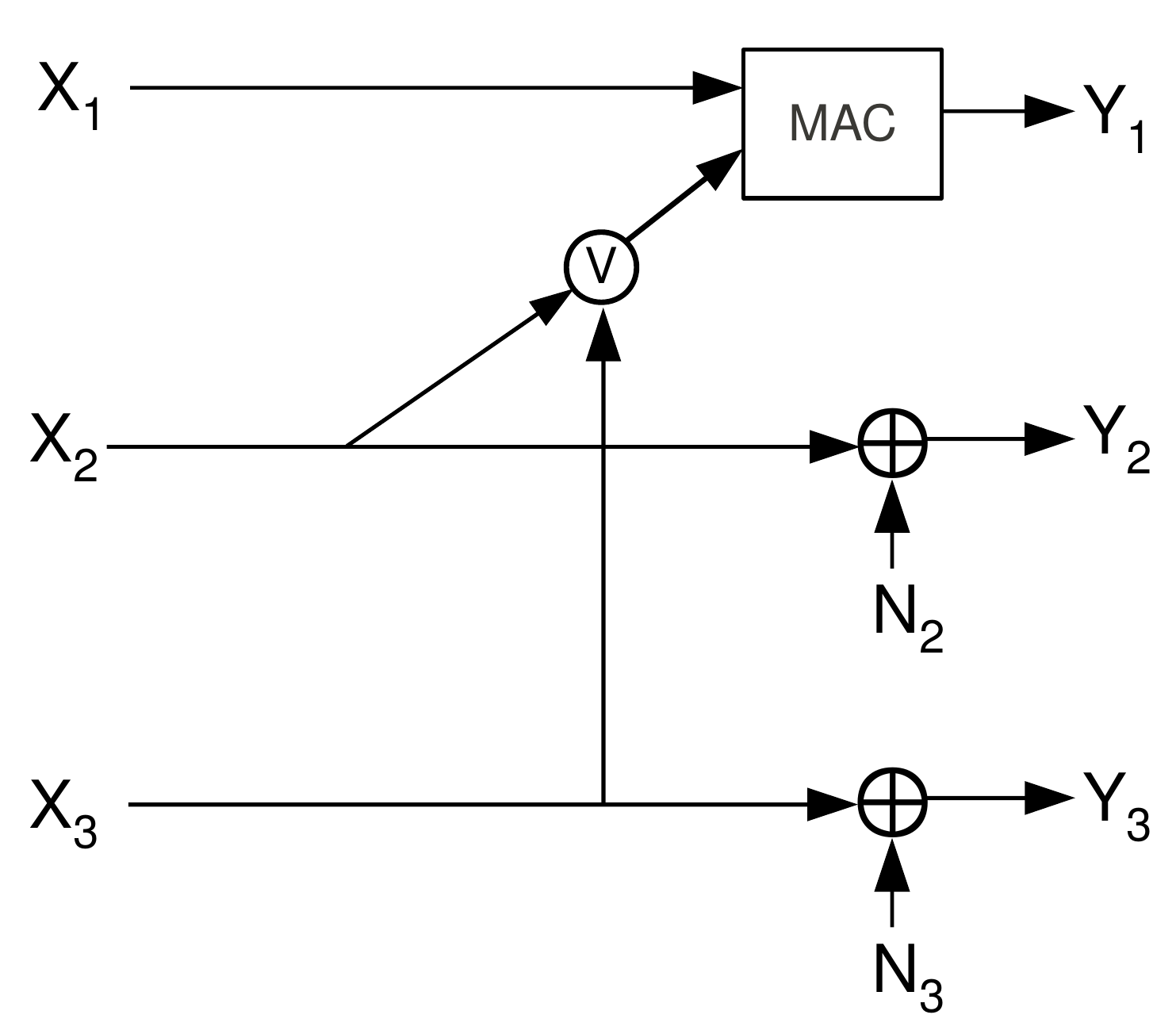}
\caption{A binary $3-$to$-1$ IC described in example \ref{Ex:3To1ORICCoupledThroughNonAdditiveMAC}.}
\label{Fig:3To1ORICWithNonAdditiveMAC}
\end{minipage}
\end{figure}
% \begin{table}\centering
% \begin{tabular}{|c|c|} \hline
% $MAC(0|0,0) = $ & $MAC(1|0,0) = $ \\
% \hline
% $MAC(0|0,1) = $ & $MAC(1|0,1) = $ \\
% \hline
% $MAC(0|1,0) = $ & $MAC(1|1,0) = $ \\
% \hline
% $MAC(0|1,1) = $ & $MAC(1|1,1) = $ \\
% \hline
% \end{tabular}\caption{Transition probabilities of non-additive MAC present in example \ref{Ex:3To1ORICCoupledThroughNonAdditiveMAC}} \label{Table:TransProbsOfNonAdditiveMAC} 
% \end{table}
Our study of example \ref{Ex:3To1ORICCoupledThroughNonAdditiveMAC} closely mimics that of example \ref{Ex:A3To1-OR-IC}. In particular, we derive conditions under which $\underline{C}^{*}\define(C_{1},h_{b}(\tau * \delta))-h_{b}(\delta),h_{b}(\tau * \delta))-h_{b}(\delta)) \in \alpha_{f}^{3-1}(\ulinecost)$ and $\underline{C}^{*} \notin \alpha_{u}(\ulinecost)$, where $\ulinecost \define (\tau_{1},\tau,\tau)$,
\begin{eqnarray}
 \label{Eqn:TestChannelsThatEnsureUsers2And3AchieveCapacity}
\lefteqn{C_{1} \define \underset{p_{\underline{X}\ulineOutputRV} \in \mathcal{D}(\ulinecost )}{\sup} I(X_{1};Y_{1}|X_{2}\vee X_{3}),}\\ \!\!\!\!\!\!\!\!\label{Eqn:TestChannelsNonAdditive3To1IC}&\!\!\!\!\!\!\!\!\mathcal{D}(\ulinecost) \define \left\{ \begin{array}{c}p_{\ulineInputRV\ulineOutputRV} \mbox{ is a pmf on }\ulineInputAlphabet \times \ulineOutputAlphabet\mbox{ such that (i) } p_{\ulineOutputRV|\ulineInputRV}=W_{\ulineOutputRV|\ulineInputRV}\mbox{ is the channel transition probabilities}\\\mbox{of example \ref{Ex:3To1ORICCoupledThroughNonAdditiveMAC}, (ii) }p_{\ulineInputRV}=p_{X_{1}}p_{X_{2}}p_{X_{3}}, p_{X_{j}}(1)=\tau\mbox{ for }j=2,3 \mbox{ and (iii) }p_{X_{1}}(1)\leq \tau_{1}\end{array}\right\}.
\end{eqnarray}
By strict concavity of $I(X_{1};Y_{1}|X_{2}\vee X_{3})$ in $p_{X_{1}}$, and the compactness of $\mathcal{D}(\ulinecost)$, there exists a unique $p^{*}_{\ulineInputRV\ulineOutputRV}$ with respect to which $I(X_{1};Y_{1}|X_{2}\vee X_{3})=C_{1}$. 

\begin{prop}
 \label{Prop:ResultForOrExampleWithMACAsAProp}
Consider example \ref{Ex:3To1ORICCoupledThroughNonAdditiveMAC} and let $\underline{C}^{*}, C_{1}, \mathcal{D}(\ulinecost),p^{*}_{\ulineInputRV\ulineOutputRV}$ be defined as above. If
\begin{eqnarray}
 \label{Eqn:3To1OrNonAddICUnstructuredCodesSubOptimalInProp}
C_{1}+2(h_{b}(\tau*\delta)-h_{b}(\delta))=I(X_{1};Y_{1}|X_{2}\vee X_{3})+2(h_{b}(\tau*\delta)-h_{b}(\delta)) > I(\ulineInputRV ;Y_{1}),
\end{eqnarray}
where $I(X_{1};Y_{1}|X_{2}\vee X_{3})$, and $I(\ulineInputRV ;Y_{1})$ are evaluated with respect to $p^{*}_{\ulineInputRV\ulineOutputRV}$, then $\underline{C}^{*} \notin \alpha_{u}(\ulinecost)$.
If
$h_{b}(\tau^{2})+(1-\tau^{2})h_{b}(\frac{(1-\tau)^{2}}{1-\tau^{2}})+H(Y_{1}|X_{2}\vee
X_{3})-H(Y_{1}) \leq \min\{ H(X_{2}|Y_{2})H(X_{3}|Y_{3})\}$, where the
entropy terms are evaluated with respect to $p^{*}_{\ulineInputRV\ulineOutputRV}$, then $\underline{C}^* \in \alpha_{f}^{3-1}(\ulinecost)$.
\end{prop}

Please refer to appendix \ref{AppSec:3To1ORNonAddICStrictSubOptimalityOfUnstructuredCodes} for a proof. For example \ref{Ex:3To1ORICCoupledThroughNonAdditiveMAC}, with $\tau_{1}=0.01,\tau=\tau_{2}=\tau_{3}=0.1525,\delta = 0.067$, the conditions stated in proposition \ref{Prop:ResultForOrExampleWithMACAsAProp} hold simultaneously. For this channel, $p^{*}_{X_{1}}(1) = 0.99$,
\[
 C_{1}+2(h_{b}(\tau*\delta)-h_{b}(\delta))- I(\ulineInputRV ;Y_{1}) = 0.0048,
\]
and
\[
\min\{ H(X_{2}|Y_{2})H(X_{3}|Y_{3})\} - [h_{b}(\tau^{2})+(1-\tau^{2})h_{b}(\frac{(1-\tau)^{2}}{1-\tau^{2}})+H(Y_{1}|X_{2}\vee X_{3})-H(Y_{1})] = 0.0031.
\]
A note on our choice of the MAC that relates $(X_{1},X_{2}\vee X_{3})$
and $Y_{1}$ is in order. The reader will recognize that the MAC is `quite close' to the additive scenario $Y_{1}=X_{1}\oplus(X_{2}\vee X_{3})\oplus N_{1}$ studied in example \ref{Ex:A3To1-OR-IC}. In order for coset codes to outperform unstructured codes, we do not need the MAC to be so `close' to the additive MAC. The need for the MAC to be `so close' is a consequence of our desire to provide an \textit{analytical} proof for strict sub-optimality of unstructured codes. Note that since we (i) do not resort to outer bounds, (ii) wish to provide analytical upper bounds to the rates achievable using unstructured codes, and (iii) cannot compute any of the associated rates in a reasonable time, we demand the MAC to be such that coset codes achieve the maximum possible rate for user $1$, with users $2$ and $3$ constrained to achieve their PTP capacities,\footnote{Note that we are demanding the channel to permit user $1$ communicate at a rate as though the receiver knew all of the non-linear interference.} and unstructured codes to be strictly sub-optimal. Finally, the above findings indicate that if structured codes yield gains for a particular channel, then one can reason out the presence of such gains for a slightly perturbed channel simply by appealing to the continuity of rate regions in the channel parameters.

In the achievable rate region presented in Theorem
\ref{Thm:AchievableRateRegionFor3To1ICUsingCosetCodes} for a 
given $3$-IC, there is a union over finite fields. Suppose we want
to maximize $\mu_1R_1+\mu_2 R_2+\mu_3R_3$ for some non-negative vector
$\underline{\mu}$ such that $\|\underline{\mu}\|=1$. The finite field
that maximizes this objective function depends on the channel in a
complicated way. It turns out that for a channel with a fixed
interference pattern,  as we change the cost functions
$\underline{\kappa}$, and the noise distributions, the optimizing
finite field also changes.  This is illustrated in the following example.

\begin{example}
\label{Ex:MultipleGroups}
Consider a quaternary \\3to1IC with
$\mathcal{X}_j=\mathcal{Y}_j=\{0,1,2,3\}: j \in [3]$ with transition
probabilities given by $Y_1=X_1 +_4 X_2 +_4 X_3 +_4 N_1$,
$Y_2=X_2 +_4 N_2$ and $Y_3=X_3 +_4 N_3$. $N_1,N_2$ and $N_3$
are mutually independent, and independent of the inputs, and $+_4$ denote addition modulo-$4$. Let $N_2$ and
$N_3$ have the same pmf. Note that the bivariate function
characterizing the interference in the channel is addition modulo-$4$, which
is not a finite field. 
Our \emph{objective} is to enable each user to attain the
corresponding point-to-point capacity. Note that we have not yet
specified the pmfs $P_{\underline{N}}$ of the noise vector $\underline{N}$, the cost function vector  $\underline{\kappa}$, and the cost constraint $\underline{\tau}$. For every triple $(P_{\underline{N}},\underline{\kappa},\underline{\tau})$, using Theorem \ref{Thm:AchievableRateRegionFor3To1ICUsingCosetCodes} 
(and its extension to Abelian groups given in Section \ref{Sec:AchievableRateRegionsFor3To1ICUsingAbelianGroups}),
one can find whether it is possible to attain our objective, and, if so,
one can find what is the `winning' finite field, or in general abelian group.
We will restrict our attention to the following two finite fields and an abelian
group: $\mathcal{F}_7$, $\mathcal{F}_8$ and $\mathbb{Z}_4$. This
requires appropriate maps from $\mathcal{F}_{7}$ and $\mathcal{F}_{8}$
to $\mathbb{Z}_{4}$. By doing a computer search, we have obtained the
following sample data (see table \ref{table:bestdistribs}). The rates
for the case of $\mathbb{Z}_4$ is obtained by using theorem
\ref{Thm:AchievableRateRegionFor3To1ICUsingGroupCosetCodes} from
Section
\ref{Sec:AchievableRateRegionsFor3To1ICUsingAbelianGroups}. For example, for the distribution in the first row, all
users achieve their respective capacities only with PCCs built on
$\mathcal{F}_7$.  Similarly PCCs built on $\mathcal{F}_8$ and
$\mathbb{Z}_4$ achieve optimality for the distributions of the second
and third rows respectively.  Note that
even though the interference pattern is fixed, the optimizing
algebraic structure depends on the cost function and the noise
distribution.

\begin{table} 
\begin{center}
\begin{tabular}{|c|c|c|p{0.35in}|p{0.35in}|p{0.35in}|c|} \hline \hline
Cost Functions ($\kappa_1,\kappa_2$) & Cost ($\tau_1,\tau_2)$ & Noise pmfs ($P_{N_1},P_{N_2})$ & $R_2 (\mathcal{F}_7)$ &
$R_2 (\mathcal{F}_8)$ & $R_2 (\mathbb{Z}_4)$ & $R_1$ \\
\hline \hline 
$[7.7572,    0.3170,    4.9891,  2.2048]$ & $0.8449$ & $[0.0011,    0.0094,    
0.0010,    0.9886]$ 
& $0.3300$ & $0.1489$ & $0.2556$ & $0.8449$ \\
$[0.2787,    0.3818,    0.3236,    0.6227]$ & $0.3300$ & $[0.5777,    0.1423,    0.1002,
    0.1798]$
& & & &  \\
\hline
$[6.1610,    1.1621,    5.0165,  0.0283]$ & $0.2245$ & $[0.8229,    0.0025,    0.1647,
    0.0099]$
& $0.0006$ & $0.2179$ & $0.0000$ & $0.2245$ \\
$[0.1357,    0.2906,    0.3514,    0.2344]$ & $0.2179$ &  $[0.1255,    0.1043,    0.3293,
    0.4409]$ 
&&&& \\
\hline
$[5.3368,    4.1262,    3.7326, 0.0100]$ & $0.1491$ & $[0.0132,    0.0285,    
0.0327,    0.9256]$
& $0.6241$ & $0.2952$ & $1.2832$ & $0.1491$ \\
$[1.4115, 1.9947, 1.1876, 0.9993]$ & $1.2832$ &  $[0.8752, 0.0290,
0.0034, 0.0924]$ 
&&&& \\
\hline
\hline
\end{tabular} \end{center}
\vspace{0.1in}
\caption{Examples of cost functions and noise distributions with
  $\kappa_2=\kappa_3$, $P_{N_2}=P_{N_3}$,  $\tau_2=\tau_3$ and $R_2=R_3$.} \label{table:bestdistribs}
\end{table}

\end{example}

\subsection{Step II: PCC rate region for a general discrete $3-$IC using codes built over finite fields}
\label{Sec:AchievableRateRegion3ICUsingNestedCosetCodes}
In this section, we employ PCC to manage interference seen by every
receiver. We describe the coding technique and provide
a characterization of the corresponding achievable rate region. In the
interest of brevity, we omit the proof of achievability. All the
non-trivial elements have been detailed in the
proof of theorem
\ref{Thm:AchievableRateRegionFor3To1ICUsingCosetCodes}. 

User $j$ splits its message $M_{j}$ of rate
$R_{j}=L_{j}+T_{ji}+T_{jk}$ into three parts
$(M_{ji}^{U},M_{jk}^{U},M_{j}^{X})$, where $i,j,k$ are distinct
indices in $\left\{ 1,2,3 \right\}$. Let
$\SemiPrivateRVSet_{ji}=\fieldpii, \SemiPrivateRVSet_{jk}=\fieldpik$
be finite fields. Let $\lambda_{ji} \subseteq
\SemiPrivateRVSet_{ji}^{n}$ denote an $(n,s_{ji},t_{ji})$ PCC and
$\lambda_{jk} \subseteq \SemiPrivateRVSet_{jk}^{n}$ denote an
$(n,s_{jk},t_{jk})$ PCC. If we let $S_{ji} \define
\frac{s_{ji}}{n}\log\Prime_{i}, T_{ji} \define
\frac{t_{ji}}{n}\log\Prime_{i}$ and $S_{jk} \define
\frac{s_{jk}}{n}\log\Prime_{k}, T_{jk} \define
\frac{t_{jk}}{n}\log\Prime_{k}$, then recall that ${\lambda_{ji}},
{\lambda_{jk}}$ are coset codes of rates $S_{ji}, S_{jk}$ partitioned
into $\exp\{ nT_{ji} \}, \exp \{nT_{jk}\}$ bins respectively. Observe
that cosets $\lambda_{ji}$ and $\lambda_{ki}$ are built over the same
finite field $\fieldpii$. To contain the range the sum of these
cosets, the larger of $\lambda_{ji}$, $\lambda_{ki}$ contains the
other. A codebook $\mathcal{C}_{j}$ of rate $K_{j}+L_{j}$ is built
over $\mathcal{X}_{j}$.  
Codewords of $\mathcal{C}_{j}$ are partitioned into $\exp\left\{ nL_{j}\right\}$ bins. 
$M_{ji}^{U}$,$M_{jk}^{U}$ and $M_{j}^{X}$ index bins in
$\lambda_{ji}$, $\lambda_{jk}$ and $\mathcal{C}_{j}$
respectively. Encoder looks for a triplet of codewords from the
indexed bins that are jointly typical with respect to a pmf
$p_{U_{ji}U_{jk}X_{j}}$ defined on $\mathcal{U}_{ji}\times
\mathcal{U}_{jk} \times \mathcal{X}_{j}$. The corresponding codeword
chosen from $\mathcal{C}_{j}$ is transmitted on the channel. Decoder
$j$ receives $Y_{j}^{n}$ and looks for all triples
$(u_{ji}^{n},u_{jk}^{n},x_{j}^{n})$ of codewords in $\lambda_{ji}
\times \lambda_{jk} \times \mathcal{C}_{j}$ for which there exists a
$u^{n}_{\oplus} \in (\lambda_{ij} \oplus \lambda_{kj})$ such that
$(u_{\oplus}^{n},u_{ji}^{n},u_{jk}^{n},x_{j}^{n},Y_{j}^{n})$ are
jointly typical with respect to $p_{U_{ij}\oplus
  U_{kj},U_{ji},U_{jk},X_{j},Y_{j}}$. If it finds all such triples in
a unique triple of bins, the corresponding triple of bin indices is
declared as decoded message of user $j$. Otherwise, an error is
declared. 

The distribution induced on the ensemble of codebooks is a simple generalization of that employed in proof of theorem \ref{Thm:AchievableRateRegionFor3To1ICUsingCosetCodes}. In particular, the codewords of $\mathcal{C}_{j}$ are chosen independently according to $\underset{t=1}{\overset{n}{\prod}}p_{X_{j}|Q}(\cdot|q^{t})$, where $q^{n}$ is an appropriately chosen time sharing sequence. The three pairs $(\Lambda_{12},\Lambda_{32}), (\Lambda_{21},\Lambda_{31}), (\Lambda_{13},\Lambda_{23})$ of random PCC are mutually independent.  Within each such pair, (i) the generator matrix of the smaller PCC is obtained by choosing each of its rows uniformly and independently, and (ii) the generator matrix of the larger is obtained by appending the generator matrix of the smaller with an appropriately chosen number of mutually independent and uniformly distributed rows. All the vectors specifying the coset shifts are chosen independently and uniformly. Moreover, partitioning of all codes into their bins is effected uniformly and independently.\footnote{The reader is encouraged to confirm that the distribution induced herein is a simple generalization of that employed in proof of theorem \ref{Thm:AchievableRateRegionFor3To1ICUsingCosetCodes}.} Deriving an upper bound on the average probability of error of this random collection of codebooks coupled with the above coding technique yields the following rate region.
\begin{definition}
 \label{Defn:CollectionOfTestChannelsForCommunicatingOverGeneral3BCUsingNCC}
Let $\mathbb{D}_{f}(\underline{\tau})$ denote the collection of probability mass functions $(p_{\TimeSharingRV\underlineSemiPrivateRV \underline{X} \underlineY})$ defined on $\TimeSharingRVSet\times\underlineSemiPrivateRVSet \times \underline{\mathcal{X}} \times \underlineSetY$, where 
(i)  $\mathcal{Q}$ is an arbitrary finite set,
(ii) $\SemiPrivateRVSet_{ij}=\fieldpij$\footnote{Recall $\fieldpij$ is the finite field of cardinality $\Prime_{j}$.} for each $1\leq  i,j\leq3$, and $\underlineSemiPrivateRVSet \define \SemiPrivateRVSet_{12} \times \SemiPrivateRVSet_{13} \times \SemiPrivateRVSet_{21} \times \SemiPrivateRVSet_{23} \times \SemiPrivateRVSet_{31} \times \SemiPrivateRVSet_{32}$,
(iii) $\underlineSemiPrivateRV\define
(\SemiPrivateRV_{12},\SemiPrivateRV_{13},\SemiPrivateRV_{21},\SemiPrivateRV_{23},\SemiPrivateRV_{31},\SemiPrivateRV_{32})$,
such that (i) the three quadruples $(U_{12},U_{13},X_{1})$, $(U_{23},U_{21},X_{2})$ and $(U_{31},U_{32},X_{3})$ are conditionally mutually independent given $Q$, (ii) $p_{\underlineY| \underlineX \underlineSemiPrivateRV Q}= p_{\underlineY |\underlineX}=W_{\underlineY |\underlineX}$, (iii) $\Expectation\left\{ \kappa_{j}(X_{j}) \right\} \leq \tau_{j}$ for $j=1,2,3$.

For $p_{\TimeSharingRV\underlineSemiPrivateRV \underline{X} \underlineY} \in \mathbb{D}_{f}(\underline{\tau})$, let $\alpha_{f}(p_{\TimeSharingRV\underlineSemiPrivateRV \underline{X} \underlineY})$ be defined as the set of rate triples $(R_{1},R_{2},R_{3}) \in [0,\infty)^{3}$ for which there exists non-negative numbers $S_{ij}:ij \in \left\{12,13,21,23,31,32 \right\}, T_{jk}:jk \in \left\{12,13,21,23,31,32 \right\}, K_{j}:j \in  \left\{ 1,2,3\right\}, L_{j}:j \in \left\{1,2,3 \right\}$ that satisfy $R_{1}=T_{12}+T_{13}+L_{1}, R_{2}=T_{21}+T_{23}+L_{2}, R_{3}=T_{31}+T_{32}+L_{3}$ and
\begin{eqnarray}
 \label{Eqn:3ICManyToManySourceCodingBound1}
S_{A_{j}}-T_{A_{j}}+K_{j} &>& \sum_{a_{j} \in A_{j}}\log |\mathcal{U}_{a_{j}}| + H(X_{j}|Q) - H(U_{A_{j}},X_{j}|Q),\\
\label{Eqn:3ICManyToManySourceCodingBound2}
S_{A_{j}}-T_{A_{j}} &>& \sum_{a_{j} \in A_{j}}\log |\mathcal{U}_{a_{j}}| - H(U_{A_{j}}|Q),
\end{eqnarray}
\begin{equation}
\begin{aligned}\label{Eqn:ManyToManyICChannelCodingBounds}
\lefteqn{S_{A_{j}} < \sum_{a \in A_{j}}\log |\mathcal{U}_{a}| - H(U_{A_{j}}|Q,U_{A_{j}^{c}},U_{ij}\oplus U_{kj},X_{j},Y_{j})}
\\
\lefteqn{S_{A_{j}}+S_{ij} < \sum_{a \in A_{j}}\log |\mathcal{U}_{a}| + \log \Prime_{j} - H(U_{A_{j}},U_{ij}\oplus
U_{kj}|Q,U_{A_{j}^{c}},X_{j},Y_{j})}  \\
\lefteqn{S_{A_{j}}+S_{kj} < \sum_{a \in A_{j}}\log |\mathcal{U}_{a}| + \log \Prime_{j} - H(U_{A_{j}},U_{ij}\oplus
U_{kj}|Q,U_{A_{j}^{c}},X_{j},Y_{j})} \\
\lefteqn{
S_{A_{j}}+K_{j}+L_{j} < \sum_{a \in A_{j}}\log |\mathcal{U}_{a}|+H(X_{j}) - H(U_{A_{j}},X_{j}|Q,U_{A_{j}^{c}},U_{ij}\oplus
U_{kj},Y_{j}) } \\
\lefteqn{S_{A_{j}}+K_{j}+L_{j}+S_{ij} < \sum_{a \in A_{j}}\log |\mathcal{U}_{a}| + \log \Prime_{j}+H(X_{j}) -
H(U_{A_{j}},X_{j},U_{ij}\oplus U_{kj}|Q,U_{A_{j}^{c}},Y_{j})} \\
&S_{A_{j}}+K_{j}+L_{j}+S_{kj} < \sum_{a \in A_{j}}\log |\mathcal{U}_{a}| + \log \Prime_{j}+H(X_{j}) - H(U_{A_{j}},X_{j},U_{ij}\oplus
U_{kj}|Q,U_{A_{j}^{c}},Y_{j}), 
\end{aligned}
\end{equation}
for every $A_{j} \subseteq \left\{ ji,jk\right\}$ with distinct indices $i,j,k$ in $\left\{ 1,2,3 \right\}$, where $S_{A_{j}} \define \sum_{a_{j} \in A_{j}}S_{a_{j}},U_{A_{j}} = (U_{a_{j}}:a_{j} \in A_{j})$. Let
\begin{eqnarray}
 \label{Eqn:AchievableRateRegionFor3BCUsingNestedCosetCodes}
 \alpha_{f}(\underline{\tau}) = \cocl\left(\underset{\substack{p_{\TimeSharingRV\underlineSemiPrivateRV \underline{X} \underlineY} \in\\ \mathbb{D}_{f}(\underline{\tau})}
}{\bigcup}\alpha_{f}(p_{\TimeSharingRV\underlineSemiPrivateRV \underline{X} \underlineY})\right).\nonumber
\end{eqnarray}
\end{definition}
\begin{thm}
 \label{Thm:AchievableRateRegionFor3BCUsingNestedCosetCodes}
For 3-IC $(\underline{\mathcal{X}},\underlineSetY,W_{\underlineY|\underlineX},\kappa)$, $\alpha_{f}(\underline{\tau})$ is achievable, i.e., $\alpha_{f}(\underline{\tau}) \subseteq \mathbb{C}(\underline{\tau})$.
\end{thm}

Although the rate region given in Theorem \ref{Thm:AchievableRateRegionFor3BCUsingNestedCosetCodes} has many auxilliary random variables, we illustrate the key ideas by applying it to a carefully constructed channel and avoiding direct computation. The above coding technique presents an approach to simultaneously manage interference at all of the receivers. It is natural to question whether the use of structured codes to manage interference comes at a cost of respective individual communication. We now provide a simple generalization of example \ref{Ex:A3To1ICForWhichLinearCodesOutperformUnstructuredCodes} that requires managing interference at two receivers. In contrast to \cite{200809Allerton_SriJafVisJafSha}, wherein the benefit of interference alignment can be exploited at all receivers, channels equipped with finite alphabets, in general, present a fundamental trade-off in managing interference and enabling individual respective communication. 

\begin{example}
 \label{Ex:Additive3ICWithInterferenceAt2Rxs}
Consider a binary additive \\3to1IC illustrated in figure \ref{Fig:Additive3ICWithInterferenceAt2Rxs} with $\InputAlphabet_{j}=\OutputAlphabet_{j}=\left\{ 0,1 \right\}:j \in [3]$ with channel transition probabilities $W_{\ulineOutputRV|\ulineInputRV}(\ulineoutput|\ulineinput)=BSC_{\delta_{1}}(y_{1}|x_{1}\oplus x_{2} \oplus x_{3})BSC_{\delta_{2}}(y_{2}|x_{2}\oplus x_{3})BSC_{\delta_{3}}(y_{3}|x_{3})$. Inputs of users $2$ and $3$ are not constrained, i.e., $\costfn_{j}(0)=\costfn_{j}(1)=0$ for $j=2,3$. User $1$'s input is constrained with $\costfn_{1}(x)=x$ for $x \in \left\{ 0,1 \right\}$ to an average cost of $\tau \in (0,\frac{1}{2})$ per symbol. Let $\mathbb{C}(\tau)$ denote the capacity region of this $3-$to$-1$ IC.
\end{example}

In order to illustrate the trade-off, let us consider the case $\delta \define \delta_{2}=\delta_{3}$ is arbitrarily close to, but greater than $\tau * \delta_{1}$. For example, one can choose $\delta_{1}=0.01, \tau=\frac{1}{8}$ and $\delta = 0.1326$. If receiver $1$ desires communication at $h_{b}(\delta_{1}*\tau)-h_{b}(\delta_{1})$, it needs to decode $X_{2}\oplus X_{3}$. To satisfy user $1$'s desire, users $2$ and $3$ have two options. Either employ codes of rates $R_{2}$ and $R_{3}$ such that $R_{2}+R_{3}< 1-h_{b}(\delta_{1}*\tau)$, or employ cosets of the same code with a hope to boost individual rates. In the latter case, user $2$ is hampered by the interference caused to it by user $3$. While we do not provide a detailed analysis, we encourage the reader to contrast this to the Gaussian IC studied in \cite{200809Allerton_SriJafVisJafSha}, wherein the richness of the real field enables each receiver to exploit the benefits of alignment. We conjecture an inherent trade-off in the ability to manage interference over finite valued channels using coset codes, and enable individual respective communication. The reader is referred to \cite{201110TIT_NazGas, 201304arXivTIT_KriZaf, 201408TIT_HonCai} wherein a similar trade-off is discussed.

\begin{figure}
\begin{minipage}{.5\textwidth}
\centering
\includegraphics[height=1.8in,width=1.9in]{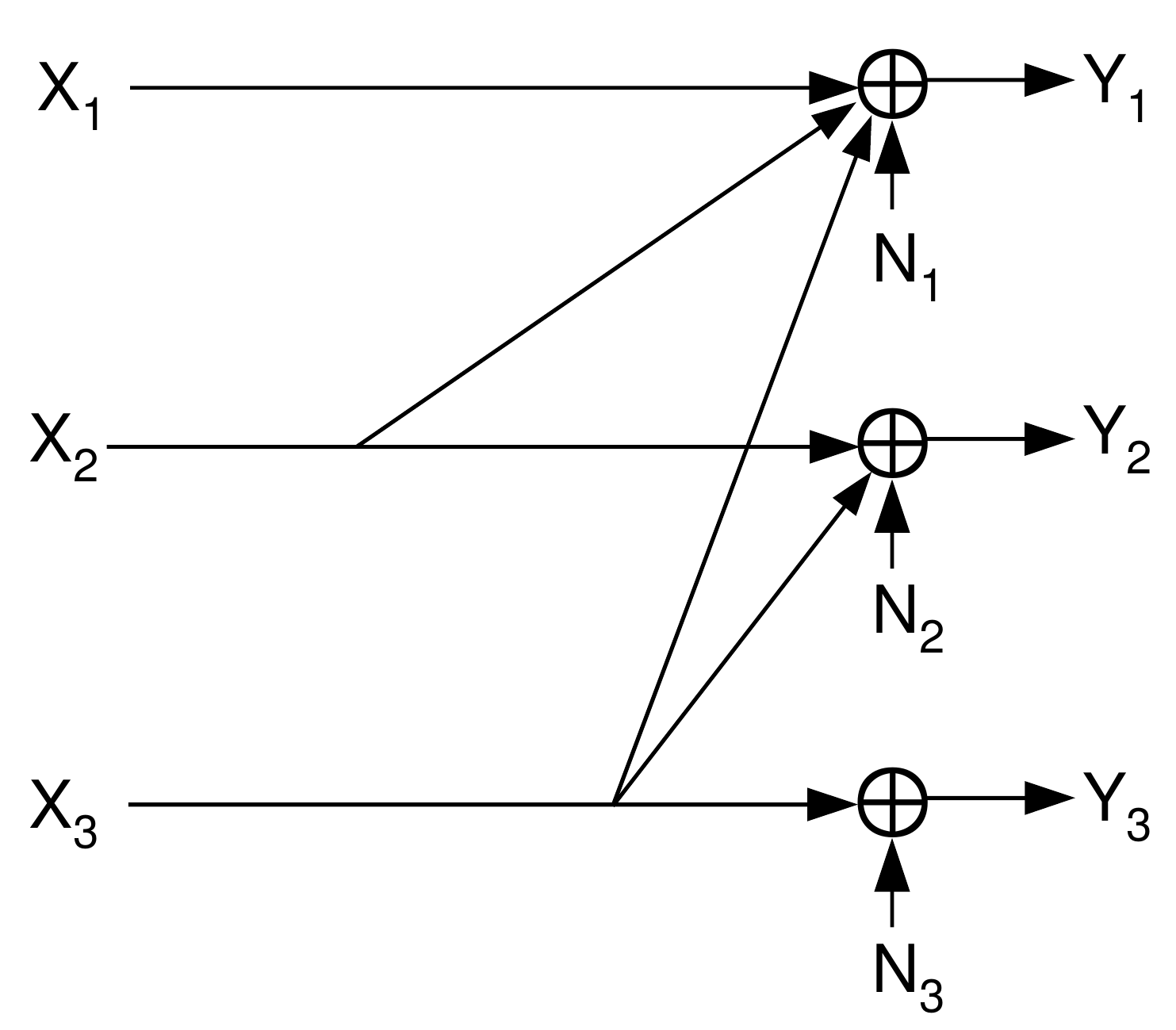}
\caption{A binary additive $3-$to$-1$ IC described in example \ref{Ex:Additive3ICWithInterferenceAt2Rxs}.}
\label{Fig:Additive3ICWithInterferenceAt2Rxs}\end{minipage}\begin{minipage}{.5\textwidth}
\centering
\includegraphics[height=1.5in]{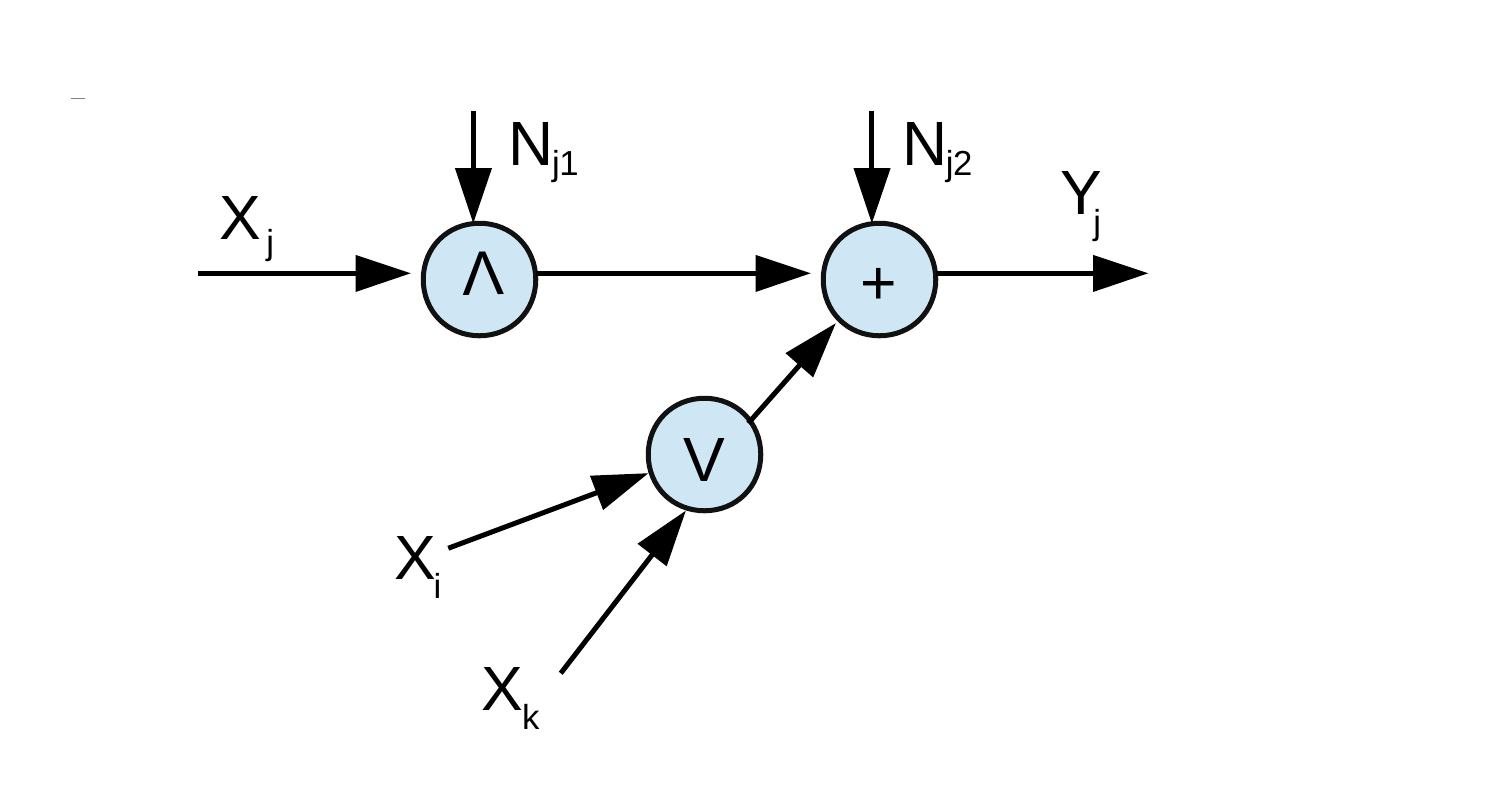}
\caption{A binary non-additive 3-IC in example \ref{Ex:3-to-3}: each user suffers from non-linear interference from other two users.}
\label{Fig:3-to-3}
\end{minipage}
\end{figure}

In the following we consider a 3-IC that is non-additive and 
uses non-uniform input distributions and all three users use structured
codes to facilitate decoding of interference at all receivers.

% \begin{figure}
% \centering
% \includegraphics[height=1.5in]{3-to-3}
% \caption{A binary non-additive 3-IC: each user suffers from non-linear interference from other two users. 
% Each user employs PCC built on $\mathcal{F}_3$ to decode the interference at the receiver.  }
% \label{Fig:3-to-3}
% \end{figure}

\begin{example}
\label{Ex:3-to-3}
Consider a binary 3-IC  with
$\mathcal{X}_j=\mathcal{Y}_j=\{0,1\}: j \in [3]$ with transition
probabilities given by $Y_j=(X_j \land N_{j1}) \oplus (X_{i} \lor X_{k}) \oplus N_{j2}$ 
for $i,j,k \in[3]$, and $i,j$ and $k$ are distinct. This is depicted in figure \ref{Fig:3-to-3}.
$N_{ji}$, $j \in [3]$, $i \in [2]$ are mutually independent and independent of the inputs. 
The cost functions are given by $\kappa_{j}(i)=i$ for $j \in [3]$, $i \in \{0,1\}$. 
$P(N_{j1}=1)=\beta$ and $P(N_{j2}=1)=\delta$ for $j \in [3]$. We let $\Expectation\{ \kappa_{j}(X_{j})\} \leq \tau$. In this channel every user suffers from 
non-linear interference. Moreover all inputs are constrained by a cost function. 
To make the example tractable we wish to operate in the high interference regime, and hence 
we have chosen a $Z$-channel in the signal path from the transmitter
to the respective receiver. 
We consider the projection of the capacity region along the line $R_1=R_2=R_3=R$, and constrain each user to achieve the corresponding PTP capacity. 
We employ PCC built on $\mathcal{F}_3$ as was done before. Using the
rate region given in theorem
\ref{Thm:AchievableRateRegionFor3BCUsingNestedCosetCodes} for this example, we get the following:
\begin{equation}
R \leq \frac{1}{2} I(X_1;Y_1|X_2 \lor X_3) + \frac{1}{2} \min \{I(X_1;Y_1|X_2 \lor X_3), \ \ 
H(X_2)-H(X_2 \oplus_3 X_3|Y_1)  \}
\end{equation}
All the three users can achieve their respective PTP capacities if $I(X_1;Y_1|X_2 \lor X_3) 
\leq H(X_2)-H(X_2 \oplus_3 X_3|Y_1)$.  
It can be verified that the choice $\delta=0.1$, $\tau=0.1284$ and $\beta=0.2210$ satisfies the condition. Hence it is possible for all the users to attain interference alignment and thus achieve their respective capacities using PCC built on $\mathcal{F}_3$.
\end{example}

In the following we consider an example that illustrates the trade-off
between the rates of two users who suffer from interference with the
third user helping one of them. This is referred to as $3-$to$-2$ IC.

\begin{example}
\label{Ex:3-to-2}
Consider a binary 3-IC  with
$\mathcal{X}_j=\mathcal{Y}_j=\{0,1\}: j \in [3]$ with transition
probabilities given by $Y_1=(X_1 \land N_{11}) \oplus (X_{2} \oplus
X_{3}) \oplus N_{12}$, $Y_2=(X_2 \land N_{21}) \oplus (X_{1} \lor
X_{3}) \oplus N_{22}$, and $Y_3=X_3 \oplus N_3$. All noise
components are mutually independent and independent of the
inputs. $\kappa_j(i)=i$ for $j \in [3]$ and $i \in \{0,1\}$. 
$P(N_{12}=1)=P(N_{22}=1)=P(N_3=1)=\delta$, and
$P(N_{11}=1)=P(N_{21}=1)=\beta$. We let $E(\kappa_j(X_j)) \leq \tau$
for $j \in [3]$.  Note that
user 1 and 2 suffer from XOR and logical-OR interference from the other two users,
respectively. The dilemma of user 3 is that it can choose to help (i) user 1 by using  PCCs
built on $\mathcal{F}_2$ and by collaborating with user 2 or (ii) 
user 2 by using  PCCs built on $\mathcal{F}_3$ and by collaborating
with user 1, but not both.  As in the previous example, to operate in the high interference regime
we have chosen the Z-channel between $X_j$ and $Y_j$ for $j=1,2$.
We evaluate the rates of the users at these two ends of the spectrum of
this trade-off. Applying theorem \ref{Thm:AchievableRateRegionFor3BCUsingNestedCosetCodes} on this
example, we get constraints on the rates of the three users. We state
these in the following only for the first operating point for conciseness: 
\begin{align}
\max\{R_2,R_3\} \leq h(\tau)+h(\tau * \tau * \delta * \tau \beta)-
h(\tau * \tau)- h(\tau\beta * d) \\
R_3 \leq h(\tau * \delta) - h(\delta), \ \ R_1 \leq h(\tau \beta *
\delta)-(1-\tau) h(\delta)- \tau h(\beta *\delta) \\
R_2  \leq h(\tau \beta * \delta * (2\tau- \tau^2)) - (1-\tau)h((2\tau- \tau^2) *
\delta) - \tau h( \beta * (2\tau- \tau^2) * \delta)
\end{align}
We provide the following data as a function of $\tau$ (see Table \ref{table:3-to-2}). In the first
operating point, we look at the corner point when $R_1$ is maximized. 
In the second, we look at the corner point when $R_2$ is
maximized. One can see the trade-off between $R_1$ and $R_2$. One can
also contrast between XOR and logical-OR interference. It is much
harder to tackle the latter as can be seen from the rates of user 3. 

\begin{table} 
\begin{center}
\begin{tabular}{|c|c|c|} \hline
$\tau$ & user 2 and 3 help user 1  & user 1 and 3 help user 2  \\
\hline \hline 
0.1 & $\underline{R}=[0.0383, 0.0012,  0.3295]$ &
$\underline{R}=[0.0021, 0.0383, 0.1360]$ \\
\hline 
0.2 & $\underline{R}=[0.0477,  0.0000, 0.4067]$ &
$\underline{R}=[0.0005, 0.0477, 0.0570]$ \\
\hline  
0.3 & $\underline{R}=[0.0520,  0.0018, 0.4364]$ & $\underline{R}=[0.0000,
0.0520, 0.0000]$ \\
\hline 
\end{tabular} \end{center}
\vspace{0.1in}
\caption{$3-$to$-2$ IC: trade-off among the rates of the three users} \label{table:3-to-2}
\end{table}

\end{example}

\subsection{Step III: Enlarging the PCC rate region using unstructured codes}
\label{Sec:UnifiedAchievableRateRegionGulingTogetherStructuredAndUnstructuredCodes}
% \textcolor{red}{Write the first line describing what the main result of the section is.}
% \input{../PhDThesis/3ICSectionAchievableRateRegions3ICWithPCCAndUnstrctrdCodes}
Let us describe a coding technique that unifies both unstructured and partitioned coset codes. We follow the approach of Ahlswede and Han \cite[Section VI]{198305TIT_AhlHan}. Refer to figure \ref{Fig:3ICMapOfRandomVariables} for an illustration of the random variables involved. Each user splits its message into $5$ parts. The $W-$random variable is decoded by all users. In addition, each user decodes a univariate component of the message of the other users. This is represented by the random variable $V$. Furthermore, it decodes a bivariate interference component denoted using $U$. Lastly, each decoder decodes all parts of its intended message.\ifThesis{ The reader may refer to section \ref{Sec:GeneralTechnique} unstructured and coset codes are glued together to derive an achievable rate region for the computation over MAC problem.}\fi
\begin{figure}
\centering
\includegraphics[width=6.45in]{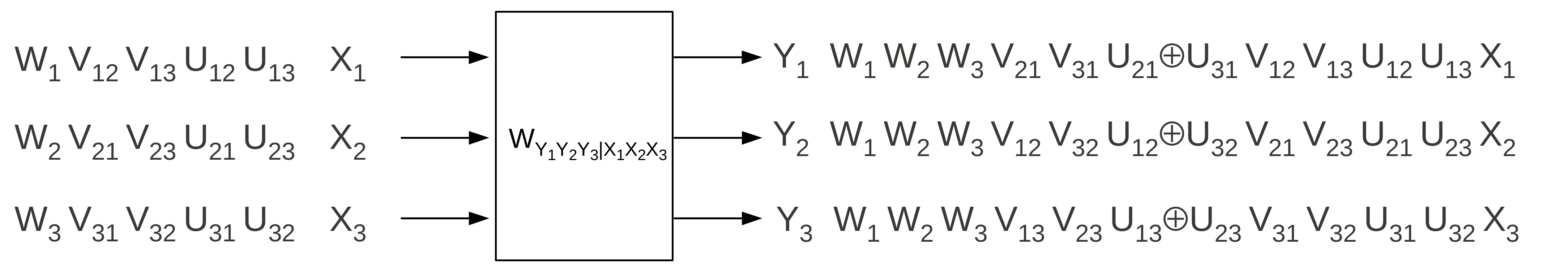}
\caption{Collection of random variables associated with coding technique that incorporates unstructured and partitioned coset codes}
\label{Fig:3ICMapOfRandomVariables}
\end{figure}
\ifTITJournal{
Clearly, a description of the above rate region is involved. In the sequel, we illustrate the key elements via a simplified achievable rate region. In particular, we employ PCC and unstructured codes to manage interference seen by only one receiver, say receiver $1$ and state the corresponding achievable rate region. We begin with a description of the same.

\begin{definition}
 \label{Eqn:TestChannelsStructuredPlusUnstructured}
Consider a $3-$IC $(\underline{\mathcal{X}},\underline{\mathcal{Y}},W_{\underline{Y}|\underline{X}},\underline{\kappa})$. Let $\SetOfDistributions_{uf}(\ulinecost)$ denote the collection of
distributions
$p_{\TimeSharingRV U_{2}V_{2}U_{3}V_{3}\underline{X}\underline{Y}}$ defined over $\TimeSharingRVSet \times \SemiPrivateRVSet_{2} \times \mathcal{V}_{2} \times \SemiPrivateRVSet_{3}  \times\mathcal{V}_{3}\times \underline{\InputAlphabet} \times \underline{\OutputAlphabet}$, where $\SemiPrivateRVSet_{2}=\SemiPrivateRVSet_{3}$ is a finite field and $\mathcal{V}_{2}$ and $\mathcal{V}_{3}$ are finite sets, such that (i) $p_{\ulineOutputRV | \ulineInputRV U_{2}V_{2}U_{3}V_{3}}= W_{\ulineOutputRV | \ulineInputRV}$, (ii) $X_{1}$, $(U_{2},V_{2},X_{2})$ and $(U_{3},V_{3},X_{3})$ are conditionally independent given $Q$, (iii) $\Expectation\{ \kappa_{j}(X_{j}) \} \leq \tau_{j}$ for $j=1,2,3$. For
$p_{\TimeSharingRV U_{2}V_{2}U_{3}V_{3}\underline{X}\underline{Y}} \in
\SetOfDistributions_{uf}(\ulinecost)$, let
$\alpha^{\three-1}_{uf}(p_{\TimeSharingRV U_{2}V_{2}U_{3}V_{3}\underline{X}\underline{Y}})$ be defined as the set of rate triples $(R_{1},R_{2},R_{3})
\in [0,\infty)^{3}$ for which $\mathcal{S}_{uf}(p_{\TimeSharingRV U_{2}V_{2}U_{3}V_{3}\underline{X}\underline{Y}},\underlineR)$ is non-empty, where $\mathcal{S}_{uf}(p_{\TimeSharingRV U_{2}V_{2}U_{3}V_{3}\underline{X}\underline{Y}},\underlineR)$ is defined as the vectors $(S_{j1},T_{j1},S_{j2},T_{j2},L_{j}:j=2,3) \in [0,\infty)^{10}$ that satisfy
\begin{eqnarray}
\label{Eqn:UnstructuredPlusFieldSourceCodingBound}
&S_{j2}-T_{j2} > \log\Prime - H(U_{j}|V_{j},Q), ~~R_{j}=T_{j1}+T_{j2}+L_{j}: j=2,3\\
\label{Eqn:UnstructuredPlusFieldUserjChnlCodingBound1}
\!\!\!\!\!\!\!&\!\!\!\!\!\!\!\!\!L_{j}\!+\!S_{j2} \!<\! \log \Prime \!-\! H(U_{j}|V_{j},Q) \!+\!I(U_{j},X_{j};Y_{j}|V_{j},Q), ~ T_{j1}\!+\!L_{j} \!<\! I(U_{j};V_{j}|Q)\!+\!I(V_{j},X_{j};Y_{j}|U_{j},Q)\!:\!j=2,3,\\
\label{Eqn:UnstructuredPlusFieldUserjChnlCodingBound2}
&L_{j} < I(X_{j};Y_{j}|U_{j},V_{j},Q), ~~ T_{j1}+S_{j2}+L_{j} < \log \Prime - H(U_{j}|V_{j},Q) + I(U_{j},V_{j},X_{j};Y_{j}|Q):j=2,3 \\
\label{Eqn:UnstructuredPlusFieldUser1ChnlCodingBound1}
\!\!\!\!\!\!\!&\!\!\!\!\!\!\!\!\!R_{1} \!<\! I(X_{1};Y_{1},V_{2},V_{3},U_{2}\oplus U_{3}|Q),~  R_{1} \!+\! S_{j2} \!<\! \log\Prime - H(U_{2}\oplus U_{3}|Q) \!+\!I(X_{1},U_{2}\oplus U_{3}; V_{2},V_{3},Y_{1}|Q) \!:\! j\!=\!2,3 \\
\label{Eqn:UnstructuredPlusFieldUser1ChnlCodingBound2}
&R_{1}+T_{j1} < I(X_{1},V_{j};V_{\msout{j}},U_{2}\oplus U_{3},Y_{1}|Q) : j=2,3 ,\nonumber~~
T_{21}+T_{31} + R_{1} < I(V_{2},V_{3},X_{1};U_{2}\oplus U_{3},Y_{1}|Q) \\
\label{Eqn:UnstructuredPlusFieldUser1ChnlCodingBound3}
&R_{1}+T_{j1}+S_{k2} < \log \Prime -H(U_{2}\oplus U_{3}|V_{j},Q)+I(X_{1},V_{j},U_{2}\oplus U_{3};V_{\msout{j}},Y_{1}|Q):j=2,3\mbox{ and }k=2,3\\
\label{Eqn:UnstructuredPlusFieldUser1ChnlCodingBound4}
&T_{21}+T_{31} +S_{j2}+ R_{1} < \log \Prime - H(U_{2}\oplus U_{3}|X_{1},V_{2},V_{3},Q)+I(X_{1},V_{2},V_{3},U_{2}\oplus U_{3};Y_{1}|Q)
\end{eqnarray}
where $\theta = |\mathcal{U}_{2}|= |\mathcal{U}_{3}|$. Let
\begin{equation}
 \label{Eqn:AchievableRateRegionStructuredPlusUnstructured}
\alpha^{\three-1}_{uf}(\ulinecost) = \cocl \left(
\underset{\substack{p_{\TimeSharingRV U_{2}V_{2}U_{3}V_{3}\underline{X}\underline{Y}} \in\\ \SetOfDistributions_{uf}(\ulinecost)}
}{\bigcup}\alpha_{uf}^{\three-1}(p_{\TimeSharingRV U_{2}V_{2}U_{3}V_{3}\underline{X}\underline{Y}}) \right).\nonumber
\end{equation}
\end{definition}
\begin{thm}
\label{Thm:AchievableRateRegionStructuredPlusUnstructured}
For $3-$IC $(\ulineInputAlphabet,\ulineOutputAlphabet,W_{\ulineOutputRV|\ulineInputRV},\ulinecostfn)$,
$\alpha^{\three-1}_{uf}(\ulinecost)$ is achievable, i.e., $\alpha_{uf}^{\three-1}(\ulinecost) \subseteq \mathbb{C}(\ulinecost)$.
\end{thm}
We provide a brief sketch of achievability. For simplicity, user $1$ builds an unstructured independent code of rate $R_{1}$ over $\mathcal{X}_{1}$ by choosing codewords independently and identically according to $p_{X_{1}}^{n}$. For $j=2,3$, user $j$ builds three random codebooks - one each over $\mathcal{V}_{j},\mathcal{U}_{j}, \mathcal{X}_{j}$ respectively. An unstructured and independent codebook of rate $T_{j1}$ is built over $\mathcal{V}_{j}$ by choosing codewords independently and identically according to $p^{n}_{V_{j}}$. A random PCC $(n,\frac{nS_{j2}}{\log\Prime},\frac{nT_{j2}}{\log\Prime},G_{j},B_{j}^{n},I_{j} )$, denoted $\Lambda_{j}$, is built over $\mathcal{U}_{j}$. As before the PCC's of users $2$ and $3$ overlap, i.e., if $j_{1} \leq j_{2}$, then $g_{j_{2}}^{T} = [ g_{j_{1}}^{T}~g_{j_{2}/j_{1}}^{T} ]$. Consider a codeword in $\mathcal{V}_{j}-$codebook and a bin in the PCC. For every such pair, a random unstructured independent codebook is constructed over $\mathcal{X}_{j}$.

User $j$th message is split into three parts - \textit{univariate} part, \textit{bivariate} part and \textit{private} part. The univariate part indexes a codeword, say $V_{j}^{n}(M_{jV})$ in $\mathcal{V}_{j}-$codebook. The bivariate part indexes a bin in the PCC. A codeword, say $U_{j}^{n}(M_{jU})$ is chosen in the indexed bin such that $(V_{j}^{n}(M_{jV}), U_{j}^{n}(M_{jU}))$ is jointly typical according to the probability distribution $p_{QV_{j}U_{j}}$, the marginal of $p_{\TimeSharingRV U_{2}V_{2}U_{3}V_{3}\underline{X}\underline{Y}} \in
\SetOfDistributions_{uf}(\ulinecost)$ in question. The codewords of the codebook built over $\mathcal{X}_{j}$, corresponding to $(M_{jV},M_{jU})$, are independently and identically distributed according to $p_{X_{j}|V_{j}U_{j}}^{n}(\cdot|V_{j}^{n}(M_{jV}),U_{j}^{n}(M_{jU}) )$. The private part $M_{jX}$ indexes a codeword in this codebook. This codeword is input on the channel by user $j$.  User $1$ inputs the codeword from its $\mathcal{X}_{1}-$codebook that is indexed by its message. It can be verified that the inequality in (\ref{Eqn:UnstructuredPlusFieldSourceCodingBound}) ensures users $2$ and $3$ find jointly typical triples of codewords.

Users $2$ and $3$ employ a simple point-to-point decoding technique. However, note that the codebook over $\mathcal{X}_{j}$ is conditionally built. Therefore, an error in decoding the correct $\mathcal{U}_{j}-$ or $\mathcal{V}_{j}-$codeword is interpreted as an error even in decoding the $\mathcal{X}_{j}-$codeword. It can be verified that (\ref{Eqn:UnstructuredPlusFieldUserjChnlCodingBound1}), (\ref{Eqn:UnstructuredPlusFieldUserjChnlCodingBound2}) ensure the probability of decoding error at receiver $j$ decays exponentially with block length $n$.

User $1$ constructs the sum codebook $\Lambda_{2}\oplus \Lambda_{3} \define \{ u_{2}^{n}\oplus u_{3}^{n}: u_{j}^{n} \in \Lambda_{j}:j=2,3 \}$ and decodes into $\mathcal{V}_{2}, \mathcal{V}_{3}, \Lambda_{2}\oplus \Lambda_{3},\mathcal{X}_{1}$ codebooks. In particular it looks for a quadruple of codewords in these codebooks that are jointly typical with the received vector $Y_{1}^{n}$ according to $p_{QV_{2},V_{3},U_{2}\oplus U_{3}|Y_{1}}$. It can be verified that (\ref{Eqn:UnstructuredPlusFieldUser1ChnlCodingBound1}) - (\ref{Eqn:UnstructuredPlusFieldUser1ChnlCodingBound4}) imply the probability of decoding error at receiver $1$ decays exponentially with block length.

\begin{example}We briefly describe an example wherein the above coding technique can yield larger achievable rate regions than ones based exclusively either on PCC or on unstructured based codes. Consider the $3-$IC depicted in figure \ref{Fig:3ICWithMultipleInputDigits}. For each $j=1,2,3$, the input alphabet $\mathcal{X}_{j} = \underset{k=1}{\overset{3}{\times}} \mathcal{X}_{jk}$ and output alphabet is $\mathcal{Y}_{j} = \underset{k=1}{\overset{3}{\times}} \mathcal{Y}_{jk}$ where $\mathcal{X}_{jk}=\mathcal{Y}_{jk}= \{ 0,1\}$. Essentially, each user can input three binary digits on the channel and each receiver observes three binary digits per channel use. Let $X_{jk}:k=1,2,3$ denote the three binary digits input by transmitter $j$ and $Y_{jk}:k=1,2,3$ denote the three digits observed by receiver $j$. Figure \ref{Fig:3ICWithMultipleInputDigits} depicts the input-output relationship. Let us also assume the Bernoulli noise processes $N_{jk}:j=1,2,3, k=1,2,3$ are mutually independent. Users $2$ and $3$ enjoy complete free point-to-point links for each of the digits. They are only constrained by noise that is modeled by the corresponding Bernoulli noise processes. Receiver $1$'s digit $Y_{11}$ experiences bivariate interference. Its $2$nd the $3$rd digits experience univariate interference.\footnote{The IC depicted in figure \ref{Fig:3ICWithMultipleInputDigits} can be used to model a scenario wherein Tx-Rx pair $1$ is assigned frequency bands around carrier frequencies $f_{1},f_{2}, f_{3}$, Tx-Rx pair $2$ is assigned frequency bands around carrier frequencies $f_{1},f_{2}, f_{4}$, Tx-Rx pair $3$ is assigned frequency bands around carrier frequencies $f_{1},f_{3}, f_{5}$ respectively. If the powers transmitted by users $2$ and $3$ are large, then user $1$ does not cause any appreciable interference to users $2$ and $3$. The interference caused by signals of Txs $2$ and $3$ on each other in frequency band around $f_{1}$ has been ignored by this model.} The reader will recognize the need for receiver $1$ to decode univariate and bivariate parts of user $2$ and $3$'s signals. The above coding technique enables the same.
\end{example}
\begin{figure}
\centering
\includegraphics[width=6.45in]{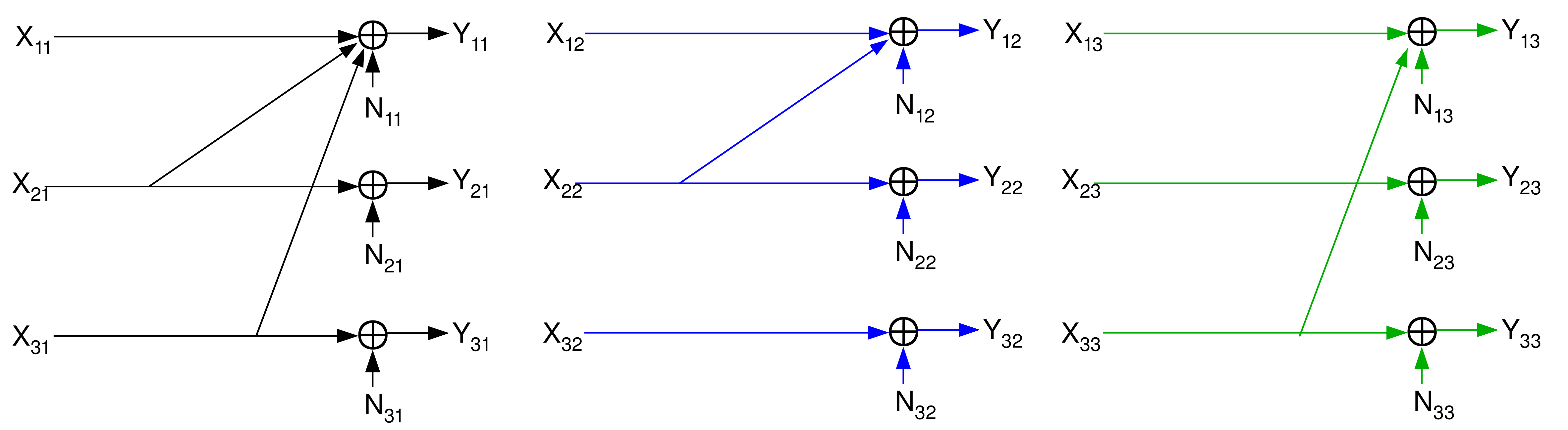}
\caption{A $3-$IC with univariate and bivariate interference components.}
\label{Fig:3ICWithMultipleInputDigits}
\end{figure}

We conclude this section with a discussion, wherein, we employ the notion of common information to argue, more fundamentally, the need to decode bivariate interference components. Let us view the above coding technique from the perspective of
common information in the sense of Gacs, K\"orner and Witsenhausen
\cite{1972MMPCT_GacKor} \cite{197501SJAM_Wit}.  Let $K(A;B)$ denote the common information of two random
variables $A$ and $B$. Let $\tilde{X}_j$ denote the collection of
random variables decoded at decoder $j$. 
The CHK scheme for $2$-IC \cite{198101TIT_HanKob} can be interpreted as inducing non-trivial
common information between $\tilde{X}_1$ and $\tilde{X}_2$, and 
$K(\tilde{X}_1;\tilde{X}_2)=H( W_1,W_2 )$. The question that comes next
is how to extend common information to $3$ random variables? We can
consider the following vector as the common information among three 
random variables $A$, $B$ and $C$: 
\[
[K(A;B;C), K(A;B), K(B;C), K(C;A)],
\]
where $K(A;B;C)$ is defined in a natural way. 
We refer to this as univariate common information as they are
characterized using univariate function of the random variables. The $\UHK$technique
induces non-trivial univariate common information among $\tilde{X}_1$,
$\tilde{X}_2$ and $\tilde{X}_3$, and 
\[
K(\tilde{X}_1;\tilde{X}_2;\tilde{X}_3)=H(W_1,W_2,W_3), \ \ 
K(\tilde{X}_j;\tilde{X}_k)=H(V_{kj},V_{jk}).
\]

The common information captured via univariate functions can be
enhanced with the following components captured via \emph{bivariate}
functions. Define
\begin{equation}
\tilde{K}(A,B;C) := \sup_{h,g_{3}}~~ \inf_{f_{1},f_{2},g_{1},g_{2}} \left\{ H(V_{3}|V_{1},V_{2}): \substack{V_{1}= f_{1}(A)= g_{1}(C) , V_{2}=f_{2}(B)=g_{2}(C), V_{3}=h(A,B)=g_{3}(C)\mbox{ where }f_{1}: \mathcal{A} \rightarrow \mathcal{V},\\  f_{2}:\mathcal{B}\rightarrow \mathcal{V}, g_{i}:\mathcal{C}\rightarrow \mathcal{V}:i=1,2, h:\mathcal{A} \times \mathcal{B} \rightarrow \mathcal{V}\mbox{ are maps into a finite set } \mathcal{V}}\right\}.\nonumber
\end{equation}
We define  common information among three random variables as a
seven-dimensional vector as follows:
\[
[K(A;B;C), K(A;B), K(B;C), K(C;A), \tilde{K}(A,B;C), \tilde{K}(B,C;A), \tilde{K}(C,A;B)].
\]
We refer to the last three components as bivariate common
information. Note that the $\UHK$technique induces trivial bivariate common
information among $\tilde{X}_1, \tilde{X}_2$ and $\tilde{X}_3$. The
PCC technique induces non-trivial bivariate common information among
them, and $\tilde{K}(\tilde{X}_i,\tilde{X}_j;\tilde{X}_k)=H(U_{ik}
\oplus U_{jk})$ for all distinct $i, j, k$.

}\fi

\section{Step IV: Achievable rate region using PCC built over Abelian groups}
\label{Sec:AchievableRateRegionsFor3To1ICUsingAbelianGroups}
In this section, we present PCC scheme using codes built on Abelian groups. The rate region we get can be interpreted as an algebraic extension (from finite fields to Abelian groups) of that given in theorem \ref{Thm:AchievableRateRegionFor3To1ICUsingCosetCodes}. 

\subsection{Definitions}
\label{SubSec:PreliminariesAboutGroups}

For an Abelian group $G$, let $\mathcal{P}(G)$ denote the set of all
distinct primes which divide $|G|$ and for a prime
$p\in\mathcal{P}(G)$ let $S_p(G)$ be the corresponding Sylow subgroup
of $G$. It is known \cite[Theorem 3.3.1]{Hal-TTOG59} that any Abelian
group $G$ can be decomposed in the following manner
\begin{align}\label{eqn:G}
G \cong \bigoplus_{p\in\mathcal{P}(G)} S_p(G)=
\bigoplus_{p\in\mathcal{P}(G)} \bigoplus_{r\in\mathcal{R}_p(G)} \mathbb{Z}_{p^r}^{M_{p,r}}= \bigoplus_{p\in\mathcal{P}(G)} 
\bigoplus_{r\in\mathcal{R}_p(G)} \bigoplus_{m=1}^{M_{p,r}} \mathbb{Z}_{p^r}^{(m)}
= \bigoplus_{(p,r,m)\in\mathcal{G}(G)} \mathbb{Z}_{p^r}^{(m)}, 
\end{align}
where $\mathcal{R}_p(G)\subseteq \mathbb{Z}^+$ and for
$r\in\mathcal{R}_p(G)$, $M_{p,r}$ is a positive integer, 
$\mathcal{G}(G)\subseteq \mathbb{P}\times \mathbb{Z}^+\times\mathbb{Z}^+$ is defined as:
\begin{align*}
\mathcal{G}(G)=\{(p,r,m)\in\mathbb{P}\times \mathbb{Z}^+\times\mathbb{Z}^+|
p\in\mathcal{P}(G), 
r\in\mathcal{R}_p(G), m\in\{1,2,\cdots,M_{p,r}\}  \}
\end{align*}
With a slight abuse of notation, we
represent an element $a$ of $G$ as 
\begin{align*}
a=\bigoplus_{(p,r,m)\in\mathcal{G}(G)} a_{p,r,m}
\end{align*}
We will need to define information theoretic quantities in relation to groups. Define
\begin{align}\label{eqn:Q_G}
\mathcal{Q}(G)=\{(p,r)|p\in\mathcal{P}(G),r\in\mathcal{R}_p(G)\}
\end{align}
Consider vectors $\hat{\theta}$, $w$ and $\theta$, with
components, indexed by $(p,r) \in\mathcal{Q}(G)$, given by 
$\hat{\theta}_{p,r}$, $w_{p,r}$ and $\theta_{p,r}$ respectively. $w$ is a
pmf on $\mathcal{Q}(G)$, $\hat{\theta}_{p,r}$ is
a non-negative integer with $0 \leq \hat{\theta}_{p,r} \leq r$,
and $\theta$ is defined as 
\begin{align*}
\pmb{\theta}(\hat{\theta})_{(p,r)\in \mathcal{Q}(G)} = \min_{s: (p,s)\in\mathcal{Q}(G)}
|r-s|^+ +\hat{\theta}_{p,s}. 
\end{align*}
It turns out that only certain subgroups of $G$ become important in the
achievable rate region when we use Abelian group codes. Define
\begin{align*}
\Theta=\left\{\pmb{\theta}(\hat{\theta})|(\hat{\theta}_{q,s})_{(q,s)\in \mathcal{Q}(G)}:0\le \hat{\theta}_{q,s}\le s\right\}.
\end{align*}
For $\theta\in\Theta$, define 
\begin{align*}
\omega_{\theta}=\frac{\displaystyle{\sum_{(p,r)\in\mathcal{Q}(G)}} 
\theta_{p,r} w_{p,r} \log p}
{\displaystyle{\sum_{(p,r)\in\mathcal{Q}(G)}} r w_{p,r} \log p}, \ \ \
\ \ 
H_{\theta}=\bigoplus_{(p,r,m)\in\mathcal{G}(G)} p^{\theta_{p,r}}
\mathbb{Z}_{p^r}^{(m)} \leq G.
\end{align*}
We give an example in the sequel.
Let $X$ and $Y$ be two  random variables with $X$ taking  values over
$G$  and let $[X]_{\theta}=X+H_{\theta}$ be the random
variable taking values from the cosets of $H_{\theta}$ in $G$ that contains $X$. We define the source coding group mutual information between $X$ and $Y$ as 
\begin{align*}
S_{w}^{G}(X;Y)&=H(X)- \log |G| + 
%\min_{\substack{w_{q,s},
%(q,s)\in\mathcal{Q}(G)\\\sum
%w_{q,s}=1}} 
\max_{\substack{\theta\in\Theta\\\theta\ne \pmb{0}}} \frac{1}{\omega_{\theta}}
\left[ \log|G:H_{\theta}| - H([X]_{\theta}|Y) \right]
\end{align*}
where $\pmb{0}$ is a vector whose components are indexed by
$(p,r)\in\mathcal{Q}(G)$ and whose $(p,r)\textsuperscript{th}$
component is equal to $0$, and $G:H_{\theta}$ is the quotient group.   We
define the channel coding group mutual information between $X$
and $Y$ as
\begin{align}\label{eqn:Rate_Channel}
C_{w}^{G}(X;Y)=H(X)- \log |G|+ 
\min_{\substack{\theta\in\Theta\\\theta\ne \pmb{r}}} \frac{1}{1-\omega_{\theta}}
\left[  \log |H_{\theta}| - H(X|[X]_{\theta},Y) \right] 
\end{align}
where $\pmb{r}$ is a vector whose components are indexed by
$(p,r)\in\mathcal{Q}(G)$ and whose $(p,r)\textsuperscript{th}$
component is equal to $r$.

For example, let $G=\mathds{Z}_2\bigoplus \mathds{Z}_8 \bigoplus
\mathds{Z}_3$. In this case, we have $\mathcal{P}(G)=\{2,3\}$,  
$\mathcal{R}_2(G)=\{1,3\}$, $\mathcal{R}_3(G)=\{1\}$ and $\mathcal{Q}(G)=\{(2,1),(2,3),(3,1)\}$. The vectors 
$w$, $\hat{\theta}$ and $\theta$ are represented by $w=(w_{2,1},w_{2,3},w_{3,1})$, $\hat{\theta}=(\hat{\theta}_{2,1},
\hat{\theta}_{2,3},\hat{\theta}_{3,1})$ and $\theta=(\theta_{2,1},\theta_{2,3},\theta_{3,1})$ and the function $\pmb{\theta}(\cdot)$ is given by
\begin{align*}
\pmb{\theta}(\hat{\theta})=\Big(\min(\hat{\theta}_{2,1},\hat{\theta}_{2,3}),\min(2+\hat{\theta}_{2,1},\hat{\theta}_{2,3}),\hat{\theta}_{3,1}\Big)
\end{align*}
The set $\Theta$ turns out to be equal to
\begin{align*}
\Theta=\Big\{(0,\!0,\!0),(0,\!0,\!1),(0,\!1,\!0),(0,\!1,\!1),(0,\!2,\!0),(0,\!2,\!1),(1,\!1,\!0),(1,\!1,\!1),(1,\!2,\!0),
(1,\!2,\!1),(1,\!3,\!0),(1,\!3,\!1)\Big\}
\end{align*}
and we have $\pmb{0}=(0,0,0)$ and $\pmb{r}=(1,3,1)$. For $\theta=(1,1,0)$, we have 
$\omega_{\theta}=\frac{w_{2,1}+w_{2,3}}{w_{2,1}+3w_{2,3}+w_{3,1}\log
  3}$ and $H_{\theta}=0 \bigoplus 2\mathds{Z}_8 \bigoplus  \mathds{Z}_3$.
so that the random variable $[X]_{\theta}$ takes values from the set
of cosets 
$0 \bigoplus   2\mathds{Z}_8 \bigoplus \mathds{Z}_3, 
0 \bigoplus (1+2\mathds{Z}_8) \bigoplus \mathds{Z}_3, 1 \bigoplus
2\mathds{Z}_8 \bigoplus \mathds{Z}_3, 
1 \bigoplus (1+ 2\mathds{Z}_8)\bigoplus   \mathds{Z}_3\Big\}$. 
Furthermore, for this choice of $\theta$, we have $|H_{\theta}|=12$ and $|G:H_{\theta}|=4$.

When $G$ is cyclic, i.e., $G=\mathbb{Z}_{p^r}$, then $w=1$ and it can be shown
that  
\[
S_w^G(X;Y)= H(X)- \min_{1 \leq \theta \leq r} \frac{r}{\theta}
H([X]_{\theta}|Y), \ \  \ \ 
C_w^G(X;Y)= H(X)- \max_{0 \leq \theta \leq (r-1)} \frac{r}{r-\theta} H(X|[X]_{\theta},Y), 
\]
When $G$ is a primary
field, i.e., $G=\mathbb{Z}_{p}$, then it follows that 
$S_w^G(X;Y)= I(X;Y)= C_w^G(X;Y)$.

\subsection{Managing interference seen by one receiver using PCC built over Abelian groups}
\label{SubSec:DecodingGroupSumOfTransmittedCodewordsUsingPCC}

In this section, we employ PCC built over Abelian groups to manage
interference seen by only receiver $1$. As the reader might have
guessed, receiver $1$ decodes the group sum of codewords chosen by
receivers $2$ and $3$.  In the following, we characterize an achievable rate region
using codes built over groups. 

\begin{definition}
 \label{Eqn:TestChannelsCodingOver3To1ICUsingNestedCosetCodes}
Let $\SetOfDistributions_{g}(\ulinecost)$ denote the collection of pairs consisting of 
a distribution $p_{QU_2U_3\ulineInputRV\ulineOutputRV}$ defined over
$\mathcal{Q} \times \mathcal{U}_2 \times \mathcal{U}_3 
\times \underline{\InputAlphabet} \times \underline{\OutputAlphabet}$, where
$\mathcal{U}_{2}=\mathcal{U}_{3}$ is an Abelian group $G$, and  a  
distribution  $w$ on $\mathcal{Q}(G)$ satisfying the following conditions:
(i) $p_{\underline{Y}|\underline{X}}=W_{\underline{Y}|\underline{X}}$,
(ii) $X_{1},(U_2,X_{2})$ and $(U_3,X_{3})$ are conditionally mutually
independent given $Q$ and (iii)
$\Expectation\{\kappa_{j}(X_{j})\} \leq \tau_{j}:j \in [3]$ and (iv)
$I(X_j;Y_j|Q,U_j)+ C_{w}^G(U_{j};Y_{j}|Q)-S_{w}^G(U_j;0|Q) \geq 0$ for
$j=2,3$. For 
$(p_{QU_2U_3\ulineInputRV\ulineOutputRV},w) \in
\SetOfDistributions_{g}(\ulinecost)$, let
$\alpha^{\three-1}_{g}(p_{QU_{2}U_{3}\ulineInputRV\ulineOutputRV},w)$ be defined as the set of rate triples $(R_{1},R_{2},R_{3})
\in [0,\infty)^{3}$ that satisfy
\begin{align}
R_1 &< I(X_1;Y_1|QZ)-H(Z|Q)+
\min\{ H(Z|Q), H(U_j|Q)+C_w^G(Z;Y_1|Q)-S_{w}^G(U_j;0|Q): j=2,3\} \nonumber\\
R_j &< I(X_j;Y_j|QU_j) + C_{w}^G(U_j;Y_j|Q) \ \ :j=2,3,   \nonumber\\
R_1+R_j &< I(X_1;Y_1|QZ)+ C_{w}^G(Z;Y_1|Q) +H(U_j|Q) -H(Z|Q)
+I(X_j;Y_j|QU_j) \nonumber \\
& \hspace{2in} + \min\{0, C_{w}^G(U_j;Y_j|Q)-S_{w}^G(U_j;0|Q) \} \ :j=2,3,\nonumber
\end{align}
where $Z=U_2 \oplus U_3$, and
\begin{equation}
 \label{Eqn:AchievableRateRegion3to1ICUsingNestedCosetCodes}
\alpha^{\three-1}_{g}(\ulinecost) = \cocl \left(
\underset{\substack{(p_{QU_{2}U_{3}\ulineInputRV\ulineOutputRV},w) \in \SetOfDistributions_{g}(\ulinecost)}
}{\bigcup}\alpha_{g}^{\three-1}(p_{\ulineInputRV\ulineOutputRV}) \right).\nonumber
\end{equation}
\end{definition}
\begin{thm}
\label{Thm:AchievableRateRegionFor3To1ICUsingGroupCosetCodes}
For $3-$IC
$(\ulineInputAlphabet,\ulineOutputAlphabet,W_{\ulineOutputRV|\ulineInputRV},\ulinecostfn)$, 
the set  $\alpha^{\three-1}_{g}(\ulinecost)$ is achievable, 
i.e., $\alpha_{g}^{\three-1}(\ulinecost) \subseteq \mathbb{C}(\ulinecost)$.
\end{thm}

The proof is given in Appendix \ref{AppSec:abeliangroupsachievability}.
We now illustrate the need to build codes over appropriate algebraic
objects to enable interference management. In other words, we provide
an example where codes built over groups outperform unstructured codes
as well as codes built over finite fields.\footnote{While, we do not
provide a proof of the statement that codes built over groups
outperform PCC built over finite fields, this can be recognized
through standard arguments.} 

\begin{example}
\label{Ex:3To1ICGroupAdditiveExample}
Consider a quaternary $3-$to$-1$ IC with input and output alphabets
$\InputAlphabet_{j} = \OutputAlphabet_{j} = \integers_{4}=
\left\{0,1,2,3 \right\}$ being the Abelian group of cardinality
$4$. Let $\GroupSum$ denote the group operation, i.e., addition
mod$-4$ in $\integers_{4}$. The channel transition probabilities are
described through the relation $Y_{1}=X_{1}\GroupSum X_{2} \GroupSum
X_{3} \GroupSum N_{1}$, $Y_{j}=X_{j} \GroupSum N_{j}$ for $j=2,3$ such
that (i) $N_{1},N_{2},N_{3}$ are independent random variables taking
values in $\integers_{4}$
with
$P(N_{j}=0)=1-\delta_{j}$ and $P(N_j=i)=\frac{\delta_{j}}{3}$
for $i,j=1,2,3$.
Inputs $X_{2},X_{3}$ of users $2$ and $3$ are not constrained, i.e.,
$\kappa_{j}(x_{j})=0$ for $j=2,3$ and any $x_{j} \in
\InputAlphabet_{j}$, whereas $\kappa_{1}(x_{1})=1$ if $x_{1} \in
\{1,2,3\}$ and $\kappa_{1}(0)=0$. User $1$'s input is constrained to a average cost of $\tau$ per symbol.
\end{example}
The reader will recognize that the $3-$to$-1$ IC described in example
\ref{Ex:3To1ICGroupAdditiveExample} is analogous to that in example
\ref{Ex:A3To1ICForWhichLinearCodesOutperformUnstructuredCodes} with
the binary field replaced by Abelian group $\integers_{4}$. 
For simplicity, let us henceforth assume
$\delta_{2}=\delta_{3}=\delta$. Since users $2$ and $3$ enjoy
interference free point-to-point links, we let them communicate at
their respective capacities. This is possible  even while using PCC built on
$\mathbb{Z}_4$ because if we choose  $U_j=X_j$ and
put a uniform distribution on $X_j$ for $j=2,3$, we get the group capacity as
\begin{align}
C_w^G(X_j;Y_j) &= \min\{2-h_b(\delta)-\delta \log_2(3),
2+2h_b(2\delta/3)-2h_b(\delta)-2\delta \log_2(3) \}=
2-h_b(\delta)-\delta \log_2(3), \nonumber
\end{align}
where the last equality follows from the concavity of entropy. 
Clearly, user $1$ can achieve a rate not
greater than $C^{*} \define
\underset{p_{X_{1}}:p_{X_{1}}(1)\leq
  \tau}{\sup}~I(X_{1};Y_{1}|X_{2}\GroupSum X_{3})$. 
The following proposition states that $C^*$ is achievable by group codes but
not by unstructured codes.  Our
approach is similar to that of section
\ref{Sec:StrictSub-optimalityOfHanKobayashiRateRegionFor3-To-1IC}. 
The proof is provided in Appendix \ref{AppSec:abeliangroupoptimality}.
\begin{prop}
 \label{thm:3To1ICGroupAddition}
 Consider the $3-$to$-1$ IC described in example
 \ref{Ex:3To1ICGroupAdditiveExample} with
 $\delta_{2}=\delta_{3}=\delta \in (0,\frac{1}{4})$, $\delta_{1} \in
 (0,\frac{1}{4})$ and $\tau < \frac{3}{4}$. If
 $\delta_{1},\tau$ and $\delta$ are such that 
 \begin{equation}
  \label{Eqn:3To1GroupICUpperBoundUnstructuredCodes}
  C^{*} +2(2-h_{b}(\delta)-\delta\log_{2}3) > 2-h_{b}(\delta_{1})-\delta_{1}\log_{2}3,
 \end{equation}
then the rate triple $(C^{*},2-h_{b}(\delta)-\delta\log_{2}3,2-h_{b}(\delta)-\delta\log_{2}3) \notin \alpha_{u}(\tau,0,0)$.
Moreover, if in addition $\beta \triangleq \delta_1+\tau-\frac{4 \delta_1 \tau}{3} 
\leq \delta$, then group codes achieve capacity, i.e., 
$(C^{*},2-h_{b}(\delta)-\delta\log_{2}3,
2-h_{b}(\delta)-\delta\log_{2}3) \in \alpha_{g}(\tau,0,0)=\mathbb{C}(\tau,0,0)$.
%  \begin{equation}
%   \label{Eqn:3To1GroupICrateTripleNotAchievable}
%   (C^{*},2-h_{b}(\delta)-\delta\log_{2}3,2-h_{b}(\delta)-\delta\log_{2}3) \notin \alpha_{u}(\tau,1,1). \nonumber
%  \end{equation}
\end{prop}

It can be shown that there exists a non-empty set of parameters
$(\delta,\delta_1,\tau)$ that satisfy these conditions. An  example is
given by  $\delta=\frac{1}{8}$,
$\delta_1=\tau=\frac{3}{4}-\frac{\sqrt{30}}{8}$. 

\section*{Acknowledgement}
 We thank the anonymous reviewers for their insightful comments that
 have helped us to improve our manuscript considerably.
\appendices
\section{Proof of proposition \ref{Prop:StrictSub-OptimalityOfHanKobayashi}}
\label{AppSec:SubOptAdditiveIC}
We only need to prove the second statement. If $H(X_{j}|\TimeSharingRV, U_{j})=0$ for $j=2,3$, then the upper bound in
(\ref{Eqn:SumRateBoundOn3To1IC}) reduces to $R_{1}+R_{2}+R_{3}\leq I(X_{2}X_{3}X_{1};Y_{1}|Q) \leq
1-h_{b}(\delta_{1})$. From the hypothesis, we have
$h_{b}(\tau*\delta_{1})-h_{b}(\delta_{1})+1-h_{b}(\delta_{2})+1-h_{b}(\delta_{3}) >
1-h_{b}(\delta_{1})$ which violates the above upper bound and hence the theorem statement
is true.

Henceforth, we assume $H(X_{j}|\TimeSharingRV, U_{j})>0$ for $j=2$ or $j=3$. Let us
assume $j,\msout{j}$ are distinct elements in $\left\{ 2,3 \right\}$ and
$H(X_{j}|Q,\SemiPrivateRV_{j})>0$. Since $(U_{2},X_{2})$ and $(U_{3},X_{3})$ are
conditionally independent given $Q$, we have
\begin{eqnarray}
\label{Eqn:EntropyOfInterferingSignalIsPositive}
0 < H(X_{j}|\TimeSharingRV, U_{j}) = H(X_{j}| X_{\msout{j}},\TimeSharingRV,U_{2},U_{3})
= H(X_{2}\oplus X_{3}|X_{\msout{j}},\TimeSharingRV,U_{2},U_{3}) \leq H(X_{2}\oplus
X_{3}|\TimeSharingRV,U_{2}U_{3}).\nonumber
\end{eqnarray}
The univariate components $\SemiPrivateRV_{2},\SemiPrivateRV_{3}$ leave residual
uncertainty in the interfering signal and imply the existence of a $\tilde{\timeshare}^{*} =
(\timeshare^{*},u_{2}^{*},u_{3}^{*}) \in \tilde{\TimeSharingRVSet} \define
\TimeSharingRVSet \times \SemiPrivateRVSet_{2} \times \SemiPrivateRVSet_{3}$ for which
$H(X_{2}\oplus X_{3}|(\TimeSharingRV,U_{2}U_{3})=\tilde{\timeshare}^{*}) > 0$. Under
this condition, we prove that the upper bound (\ref{Eqn:3To1ICUpperBoundOnR1}) on $R_{1}$ is
strictly smaller than $h_{b}(\tau*\delta_{1})-h_{b}(\delta_{1})$. Towards that end, we
prove a simple observation based on strict concavity of binary entropy function.
\begin{lemma}
 \label{Lem:AddingAnIndependentRandomVariableReducesDifferenceInEntropies}
If $Z_{j}:j\in [3]$ are binary random variables such that (i) $H(Z_{1}) \geq H(Z_{2})$, (ii) $Z_{3}$ is independent of $(Z_{1},Z_{2})$, then $H(Z_{1})-H(Z_{2}) \geq |H(Z_{1}\oplus Z_{3}) - H(Z_{2}\oplus Z_{3})|$. Moreover, if $H(Z_{1}) > H(Z_{2})$ and $H(Z_{3})>0$, then the inequality is strict, i.e., $H(Z_{1})-H(Z_{2}) > |H(Z_{1}\oplus Z_{3}) - H(Z_{2}\oplus Z_{3})|$.
\end{lemma}
\begin{proof}
Note that, if either $H(Z_{1}) = H(Z_{2})$ or $H(Z_{3})=0$, then $H(Z_{1})-H(Z_{2}) = H(Z_{1}\oplus Z_{3}) - H(Z_{2}\oplus Z_{3})$. We therefore assume $H(Z_{1}) > H(Z_{2})$ and $H(Z_{3})>0$ and prove the case of strict inequality. For $j \in [3]$, let $\left\{ p_{Z_{j}}(0),p_{Z_{j}}(1) \right\}=\left\{ \delta_{j},1-\delta_{j} \right\}$ with $\delta_{j} \in [0,\frac{1}{2}]$, $\delta_{3}>0$. Define $f:[0,\frac{1}{2}] \rightarrow [0,1]$ as $f(t)=h_{b}(\delta_{1}*t)-h_{b}(\delta_{2}*t)$. It suffices to prove $f(0) > f(\delta_{3})$. By the Taylor series, $f(\delta_{3})=f(0) + \delta_{3}f'(\zeta)$ for some $\zeta \in [0,\delta_{3}]$ and therefore it suffices to prove $f'(t) < 0$ for $t \in (0,\frac{1}{2}]$.

It may be verified that 
\begin{equation}
\label{Eqn:DerivativeDifferenceInEntropies}
 f'(t) = (1-2\delta_{1})\log \frac{1-\bar{\delta}_{1}}{\bar{\delta}_{1}}-(1-2\delta_{2})\log \frac{1-\bar{\delta}_{2}}{\bar{\delta}_{2}}, \mbox{ where }\bar{\delta}_{j}= \delta_{j}+t(1-2\delta_{j}):j \in [2].\nonumber
\end{equation}
Note that (i) $0 \leq (1-2\delta_{1}) < (1-2\delta_{2}) \leq 1$, (ii) $\bar{\delta}_{j} \leq \delta_{j}+\frac{1}{2}(1-2\delta_{j}) \leq \frac{1}{2}$, (iii) since $\delta_{1}> \delta_{2}$ and $t\leq \frac{1}{2}$, $\bar{\delta}_{1}-\bar{\delta}_{2}=(\delta_{1}-\delta_{2})(1-2t) \geq 0$. We therefore have $0 \leq \bar{\delta}_{2} \leq \bar{\delta}_{1} \leq \frac{1}{2}$ and thus $\log \frac{1-\bar{\delta}_{2}}{\bar{\delta}_{2}} \geq \log \frac{1-\bar{\delta}_{1}}{\bar{\delta}_{1}}$. Combining this with the first observation, we conclude $(1-2\delta_{2})\log \frac{1-\bar{\delta}_{2}}{\bar{\delta}_{2}} > (1-2\delta_{1})\log \frac{1-\bar{\delta}_{1}}{\bar{\delta}_{1}}$ which implies $f'(t) < 0$ for $t \in (0,\frac{1}{2}]$.
\end{proof}
We are now equipped to work with the upper bound (\ref{Eqn:3To1ICUpperBoundOnR1}) on
$R_{1}$. Denoting $\tilde{\TimeSharingRV} \define
(\TimeSharingRV,\SemiPrivateRV_{2},\SemiPrivateRV_{3})$ and a generic element
$\tilde{\timeshare} \define (\timeshare,u_{2},u_{3}) \in \tilde{\TimeSharingRVSet}
\define \TimeSharingRVSet \times \SemiPrivateRVSet_{2} \times \SemiPrivateRVSet_{3}$, we
observe that
 \ifTITJournal{\begin{eqnarray}
 \lefteqn{ \label{Eqn:X2X3N1IndependentOfX1GivenQ} I(X_{1};Y_{1}|\tilde{\TimeSharingRV}) = \sum_{\tilde{\timeshare}}p_{\tilde{\TimeSharingRV}}(\tilde{\timeshare})H(X_{1}\oplus
N_{1}\oplus  X_{2}\oplus X_{3} |\tilde{\TimeSharingRV}=\tilde{\timeshare})-\sum_{x_{ 1
},\tilde{\timeshare}}p_{X_{1}\tilde{\TimeSharingRV}}(x_{1,}\tilde{\timeshare})
H(N_{1}\oplus  X_{2}\oplus X_{3}
|\tilde{\TimeSharingRV}=\tilde{\timeshare})}\\
&=&\sum_{\tilde{\timeshare}}p_{\tilde{\TimeSharingRV}}(\tilde{\timeshare})H(X_{1}\oplus
N_{1}\oplus  X_{2}\oplus X_{3} |\tilde{\TimeSharingRV}=\tilde{\timeshare}
)-\sum_{\tilde{\timeshare}}p_{\tilde{\TimeSharingRV}}(\tilde{\timeshare})
H(N_{1}\oplus X_{2}\oplus X_{3}
|\tilde{\TimeSharingRV}=\tilde{\timeshare})\nonumber\\
&<&\sum_{\tilde{\timeshare}}p_{\tilde{\TimeSharingRV}}(\tilde{\timeshare})H(X_{1}\oplus
N_{1}|\tilde{\TimeSharingRV}=\tilde{\timeshare}
)\label{Eqn:InequalityDueToUncertainityInX2PlusX3GivenU2U3}-\sum_{\tilde{\timeshare}} p_ { \tilde{\TimeSharingRV} } (\tilde{\timeshare} )
H(N_{1}|\tilde{\TimeSharingRV}=\tilde{\timeshare})=\sum_{\timeshare}p_{\TimeSharingRV}(\timeshare)H(X_{1}\oplus N_{1}|\TimeSharingRV=\timeshare)-h_{b}(\delta_{1})\\
\label{Eqn:Jensen'sInequalityAndCostConstraint}
&=&\sum_{\timeshare}p_{\TimeSharingRV}(\timeshare)h_{b}(p_{X_{1}|Q}(1|q)*\delta_{1})- h_{b}(\delta_{1})\leq h_{b}(\Expectation_{\TimeSharingRV}\left\{ p_{X_{1}|Q}(1|q)*\delta_{1} \right\})- h_{b}(\delta_{1})\leq h_{b}(\tau*\delta_{1})-h_{b}(\delta_{1}),
 \end{eqnarray}}\fi
where (i) (\ref{Eqn:X2X3N1IndependentOfX1GivenQ}) follows from independence of $(N_{1},X_{2},X_{3})$ and $X_{1}$ conditioned on realization of $\TimeSharingRV$, (ii) (\ref{Eqn:InequalityDueToUncertainityInX2PlusX3GivenU2U3}) follows from the existence of a $\tilde{\timeshare}^{*} \in \tilde{\TimeSharingRVSet}$ for which $H(X_{2}\oplus X_{3}|\tilde{\TimeSharingRV}=\tilde{\timeshare}^{*})>0$ and substituting $p_{X_{1}\oplus N_{1}|\tilde{\TimeSharingRV}}(\cdot|\tilde{\timeshare}^{*})$ for $p_{Z_{1}}$,
$p_{N_{1}|\tilde{\TimeSharingRV}}(\cdot|\tilde{\timeshare}^{*})$ for $p_{Z_{2}}$ and $p_{X_{2}\oplus
X_{3}|\tilde{\TimeSharingRV}}(\cdot|\tilde{\timeshare}^{*})$ for $p_{Z_{3}}$ in lemma \ref{Lem:AddingAnIndependentRandomVariableReducesDifferenceInEntropies}, and noting that $p_{X_{1}\oplus N_{1}|\tilde{Q}}(1|\tilde{q}^{*})> p_{N_{1}|\tilde{Q}}(1|q^{*})$, (iii) the first inequality in (\ref{Eqn:Jensen'sInequalityAndCostConstraint}) follows from Jensen's inequality and the second follows from the cost constraint that any test channel in $\SetOfDistributions_{u}(\tau,0,0)$ must satisfy. 

\section{Proof of achievability}
\label{AppSec:ProofOfAch}

 Let $p_{\TimeSharingRV \SemiPrivateRV_{2}\SemiPrivateRV_{3}\ulineInputRV \ulineOutputRV} \in \SetOfDistributions_{f}(\ulinecost)$, $\underlineR \in \alpha^{\three-1}_{f}(p_{\TimeSharingRV\SemiPrivateRV_{2}\SemiPrivateRV_{3}
\ulineInputRV\ulineOutputRV})$ and $\tilde{\eta} > 0$. Let us assume $\SemiPrivateRVSet_{2}=\SemiPrivateRVSet_{3}=\fieldpi$ is the finite field of size $\Prime$. For each $n \in \naturals $ sufficiently large, we
prove existence of a \threeIC code $(n, \underline{\mathscr{M}},\underlinee,\underlined)$ for
which $\frac{\log \mathscr{M}_{k}}{n} \geq R_{k}-\tilde{\eta}$, $\cost_{k}(e_{k}) \leq
\tau_{k}+\tilde{\eta}$ for $k\in [3]$ and $\overline{\xi}(\ulinee,\ulined) \leq \tilde{\eta}$.

Taking a cue from the above coding technique, we begin with an alternative characterization of $\alpha^{\three-1}_{f}(p_{\TimeSharingRV\SemiPrivateRV_{2}\SemiPrivateRV_{3} \ulineInputRV\ulineOutputRV})$ in terms of the parameters of the code.
\begin{definition}
 \label{Defn:3To1ICAchRegionSimplifiedPreFourierMotzkinBounds}
 Consider $p_{\TimeSharingRV\SemiPrivateRV_{2}\SemiPrivateRV_{3}\ulineInputRV\ulineOutputRV} \in
\SetOfDistributions_{f}(\ulinecost)$ and let $\fieldpi \define \SemiPrivateRVSet_{2}=\SemiPrivateRVSet_{3}$. Let
$\tilde{\alpha}^{\three-1}_{f}(p_{\TimeSharingRV\SemiPrivateRV_{2}\SemiPrivateRV_{3}
\ulineInputRV\ulineOutputRV})$ be defined as the set of rate triples $(R_{1},R_{2},R_{3})
\in [0,\infty)^{3}$ for which $\underset{\delta > 0}{\cup}{\mathcal{S}}(\underlineR,p_{\TimeSharingRV\SemiPrivateRV_{2}\SemiPrivateRV_{3}
\ulineInputRV\ulineOutputRV},\delta)$ is non-empty, where ${\mathcal{S}}(\underlineR,p_{\TimeSharingRV\SemiPrivateRV_{2}\SemiPrivateRV_{3}
\ulineInputRV\ulineOutputRV},\delta )$ is defined as the collection of vectors $(S_{2},T_{2},K_{2},L_{2},S_{3},T_{3},K_{3},L_{3}) \in [0,\infty )^{8}$ that satisfy
\begin{eqnarray}
 \label{Eqn:3-to-1ICBoundsInTermsOfCodeParameters}
&R_{j}=T_{j}+L_{j},~K_{j} > \delta,~~(S_{j}-T_{j}) > \log \Prime- H(U_{j}|Q)+\delta,\nonumber\\ &(S_{j}-T_{j})+{K_{j}} > \log \Prime + {H(X_{j}|Q)-H(U_{j},X_{j}|Q)}+\delta \nonumber\\
&T_{j}>\delta,~L_{j}>\delta, ~K_{j}+L_{j} < I(X_{j};Y_{j},U_{j}|\TimeSharingRV)-\delta,~~S_{j} < \log \Prime -H(U_{j}|X_{j},Y_{j},\TimeSharingRV)-\delta,\nonumber\\
&S_{j}+K_{j}+L_{j} < \log \Prime
+H(X_{j}|\TimeSharingRV)-H(U_{j},X_{j}|Y_{j},\TimeSharingRV)-\delta, ~~
R_{1}  < I(X_{1};Y_{1},U_{2}\oplus U_{3}|\TimeSharingRV)-\delta\nonumber\\
&R_{1}+S_{j} < \log\Prime + H(X_{1}|\TimeSharingRV) -
H(X_{1},U_{2}\oplus U_{3}|Y_{1},\TimeSharingRV)-\delta\nonumber
 \end{eqnarray}
for $j=2,3$.
\end{definition}
\begin{lemma}
 \label{Lem:3To1ICSimplifiedPreFourierMotzkinBoundEquality}
$\tilde{{\alpha}}^{\three-1}_{f}(p_{\TimeSharingRV\SemiPrivateRV_{2}\SemiPrivateRV_{3}
\ulineInputRV\ulineOutputRV})=\alpha^{\three-1}_{f}(p_{\TimeSharingRV\SemiPrivateRV_{2}\SemiPrivateRV_{3}
\ulineInputRV\ulineOutputRV})$.
\end{lemma}
\begin{proof}
The proof follows by substituting $R_{j}=T_{j}+L_{j}$ in the bounds characterizing ${\mathcal{S}}(\underlineR,p_{\TimeSharingRV\SemiPrivateRV_{2}\SemiPrivateRV_{3}
\ulineInputRV\ulineOutputRV})$ and eliminating $S_{j},T_{j},K_{j},L_{j}: j=2,3$ via the technique of Fourier Motzkin. The resulting characterization will be that of $\alpha^{\three-1}_{f}(p_{\TimeSharingRV\SemiPrivateRV_{2}\SemiPrivateRV_{3}
\ulineInputRV\ulineOutputRV})$. The presence of strict inequalities in the bounds characterizing $\alpha^{\three-1}_{f}(p_{\TimeSharingRV\SemiPrivateRV_{2}\SemiPrivateRV_{3}
\ulineInputRV\ulineOutputRV})$ and ${\mathcal{S}}(\underlineR,p_{\TimeSharingRV\SemiPrivateRV_{2}\SemiPrivateRV_{3}
\ulineInputRV\ulineOutputRV},\delta )$ enables one to prove $\underset{\delta > 0}{\cup}{\mathcal{S}}(\underlineR,p_{\TimeSharingRV\SemiPrivateRV_{2}\SemiPrivateRV_{3}
\ulineInputRV\ulineOutputRV},\delta )$ is non-empty for every $\underline{R} \in \alpha^{\three-1}_{f}(p_{\TimeSharingRV\SemiPrivateRV_{2}\SemiPrivateRV_{3}
\ulineInputRV\ulineOutputRV})$.
\end{proof}
Lemma \ref{Lem:3To1ICSimplifiedPreFourierMotzkinBoundEquality} provides us with $\delta > 0$ and parameters $(S_{j},T_{j},K_{j},L_{j},:j=2,3) \in \mathcal{S}(\underlineR,p_{QU_{2}U_{3}\underlineX,\underlineY},\delta)$ of the code whose existence we seek to prove. Define $\eta = \frac{1}{2^{d}}\min\{\delta,\tilde{\eta} \}$, where $d \in \naturals$ will be specified in due course. Let $q^{n} \in T_{\eta}(Q)$ denote the time sharing sequence. User $1$'s code contains $ \exp \{ nR_{1}\}$ codewords $(x_{1}^{n}(m_{1}) \in \InputAlphabet_{1}^{n}:m_{1} \in \MessageSet_{1} )$, where $\MessageSet_{1} \define [ \exp \{ nR_{1}\}]$. For $j \in \{ 2,3\}$, user $j$'th cloud center codebook $\lambda_{j}$ is the PCC $(n,s_{j},t_{j},g_{j},b_{j}^{n},i_{j})$
built over $\SemiPrivateRVSet_{j}^{n}=\fieldpi^{n}$ where $s_{j} \define  \frac{nS_{j}}{\log\Prime}$ and $t_{j} \define \frac{nT_{j}}{\log\Prime}$. We refer the reader to the coding technique described prior to the proof for the definitions of $u_{j}^{n}(a^{s_{j}})$ and $c_{j1}(m_{j1})$. The PCCs \textit{overlap}, and without loss of generality, we assume $s_{2} \leq s_{3}$ and therefore $g_{3}^{T} = [g_{2}^{T}~~g_{3/2}^{T}]$.

We now specify encoding rules. Encoder $1$ feeds codeword $x_{1}^{n}(M_{1})$ indexed by the message as input. For $j=2,3$, encoder $j$ populates
\begin{equation}
\label{Eqn:3to1ICListOfEncoderj}
\mathcal{L}_{j}(M_{j}) \define \{ (u^{n}_{j}(a^{s_{j}}),x_{j}^{n}(M_{jX},b_{jX})) \in T_{2\eta}(U_{j},X_{j}|q^{n}) : (a^{s_{j}},b_{jX}) \in c_{j1}(M_{j1})\times c_{jX}\}.\nonumber
\end{equation}
If $\mathcal{L}_{j}(M_{j})$ is non-empty, one of these pairs is chosen. Otherwise, one pair from $\lambda_{j}\times \mathcal{C}_{j}$ is chosen. Let $(U_{j}^{n}(A^{s_{j}})$, $X_{j}^{n}(M_{jX},B_{jX}))$ denote the chosen pair. $X_{j}^{n}(M_{jX},B_{jX})$ is fed as input on the channel.

Decoder $1$ constructs the sum $\lambda_{2}\oplus \lambda_{3} \define \left\{ u_{2}^{n}
\oplus u_{3}^{n}:u_{j}^{n} \in \lambda_{j} ,j =2,3\right\}$ of the cloud center codebooks. Let $u_{\oplus}^{n}(a^{s_{3}}) \define a^{s_{3}}g_{3}\oplus b_{2}^{n}\oplus b_{3}^{n}$ denote a generic codeword in $\lambda_{2} \oplus \lambda_{3} $. Note that $\lambda_{2} \oplus \lambda_{3} = \left\{ u_{\oplus}^{n}(a^{s_{3}}):a^{s_{3}} \in \SemiPrivateRVSet_{3}^{s_{3}} \right\}$.\footnote{Here we have used the assumption $s_{2} \leq s_{3}$. In general, if $s_{j_{1}} \leq s_{j_{2}}$, we have $\lambda_{2}\oplus \lambda_{3} = \left\{  u_{\oplus}^{n}(a^{s_{j_{2}}}): a^{s_{j_{2}}} \in \SemiPrivateRVSet_{j_{2}}^{s_{j_{2}}} \right\}$, where $u_{\oplus}^{n}(a^{s_{j_{2}}}) \define a^{s_{j_{2}}}g_{j_{2}}\oplus b_{2}^{n}\oplus b_{3}^{n}$ denotes a generic codeword.}
Having received $Y_{1}^{n}$, it looks for all potential message $\hat{m}_{1}$ for which
there exists a $a^{s_{3}} \in \mathcal{U}_{3}^{s_{3}}$ such that
$(q^{n},u_{\oplus}^{n}(a^{s_{3}}),$ $x_{1}^{n}(\hat{m}_{1}),Y_{1}^{n}) \in T_{2\eta_{1}}(Q,U_{2}\oplus U_{3},X_{1},Y_{1})$\footnote{The choice for $\eta_{1}$ is indicated at the end of the proof.}. If it finds exactly
one such message $\hat{m}_{1}$, it declares this as decoded message of user $1$.
Otherwise, it declares an error.

For $j\in \{2,3 \}$, decoder $j$ identifies all $(\hat{\msg}_{j1},\hat{\msg}_{j\InputRV})$ for which there exists $(a^{s_{j}},b_{jX}) \in c_{j1}(\hat{\msg}_{j1}) \times c_{jX}$ such that $(\timeshare^{n},u^{n}_{j}(a^{s_{j}}),x_{j}^{n}(\hat{\msg}_{j\InputRV},b_{j\InputRV}),Y_{j}^{n}) \in
T_{2\eta_{1}}(Q,U_{j},X_{j},Y_{j})$, where $Y_{j}^{n}$ is the received vector. If
there is exactly one such pair
$(\hat{\msg}_{j1},\hat{\msg}_{j\InputRV})$, this is declared as message of user $j$ . Otherwise an error is signaled.

The above encoding and decoding rules map every quintuple of codes
$(\mathcal{C}_{1},\lambda_{2},\lambda_{3},\mathcal{C}_{2},\mathcal{C}_{3})$ into a
corresponding \threeIC code $(n,\underline{\mathcal{M}},\ulinee,\ulined)$ of rate $\frac{\log
|\mathcal{M}_{1}|}{n} = R_{1}, \frac{\log|\mathcal{M}_{j}|}{n} =
\frac{t_{j}}{n}\log\Prime + L_{j} = T_{j}+L_{j}=R_{j}:j \in \{2,3\}$, thus characterizing an
ensemble of \threeIC codes, one for each $n \in \naturals$. We average error
probability over this ensemble of \threeIC codes by letting (i) the codewords
of $\mathcal{C}_{1}\define (X_{1}^{n}(m_{1}):m_{1} \in \MessageSet_{1})$, generator
matrices $G_{2}, G_{3/2}$\footnote{Recall, that we have assumed $s_{2} \leq s_{3}$.}, bias vectors $B_{1}^{n},B_{2}^{n}$, bin indices
$(I_{j}(a^{s_{j}}):a^{s_{j}} \in \SemiPrivateRVSet_{j}^{s_{j}}):j=2,3$ and codewords of $\mathcal{C}_{j}=(X_{j}^{n}(m_{jX},b_{jX}):(m_{jX},b_{jX}) \in \MessageSet_{jX}\times c_{jX}):j=2,3$ be mutually independent, (ii) the codewords of $\mathcal{C}_{j}:j=1,2,3$ are identically distributed according to $\prod_{t=1}^{n}p_{X_{j}|Q}(\cdot|q_{t})$, (iii) generator matrices $G_{j_{1}},
G_{j_{2}/j_{1}}$, bias vectors $B_{1}^{n},B_{2}^{n}$, bin indices $(I_{j}(a^{s_{j}}):a^{s_{j}}
\in \SemiPrivateRVSet_{j}^{s_{j}}):j=2,3$ be uniformly distributed over their respective
range spaces. We denote the random partitioned coset code $(n,s_{j},t_{j},G_{j},B_{j}^{n},I_{j})$ of user $j$ as $\Lambda_{j}$ and let (i) $U_{j}^{n}(a^{s_{j}}) \define a^{s_{j}}G_{j}\oplus B_{j}^{n}$ denote a generic random codeword in $\Lambda_{j}$, (ii) $U_{\oplus}^{n}(a^{s_{3}}) \define a^{s_{3}}G_{3} \oplus B_{2}^{n} \oplus B_{3}^{n}$ denote a generic codeword in $\Lambda_{2} \oplus \Lambda_{3}$, and (iii) $C_{j1}(M_{j1}) = \{ a^{s_{j}} \in \mathcal{U}_{j}^{s_{j}}:I_{j}(a^{s_{j}})=M_{j1} \}$ denote the random collection of indices corresponding to message $M_{j1}$.

We now proceed towards deriving an upper bound on the probability of error. Towards that end, we begin with a characterization of error events. Let
\begin{eqnarray}
\label{Eqn:3To1ICCharacterizationOfErrorEvents}
\lefteqn{~~~~~\epsilon_{11} \define \left\{ (q^{n},X_{1}^{n}(M_{1})) \notin T_{2\eta}(Q,X_{1}) \right\}}\nonumber\\
&&\epsilon_{1j} \define \bigcap_{\substack{(a^{s_{j}},b_{jX})\in\\ C_{j1}(M_{j1}) \times c_{jX}}}\left\{ (q^{n},U_{j}^{n}(a^{s_{j}}),X_{j}^{n}(M_{jX},b_{jX})) \notin T_{2\eta}(Q,U_{j},X_{j}) \right\},\mbox{ for }j=2,3 
  \nonumber\\
\label{Eqn:3To1ICPCCProofEpsilon2DefinitionInMainBody}
&&\epsilon_{2} \define \left\{ (q^{n},U_{2}^{n}(A^{s_{2}}),U_{3}^{n}(A^{s_{3}}),X_{1}^{n}(M_{1}),X_{2}^{n}(M_{2X},B_{2X}),X_{3}^{n}(M_{3X},B_{3X})) \notin T_{\eta_{1}}(Q,U_{2},U_{3},\underline{X}) \right\}\\
 \label{Eqn:3To1ICPCCProofEpsilon3DefinitionInMainBody}
&&\!\!\!\!\!\!\!\!\!\!\!\!\!\!\!\!\epsilon_{3}\define\!\left\{ (q^{n},U_{2}^{n}(A^{s_{2}}),U_{3}^{n}(A^{s_{3}}),X_{1}^{n}(M_{1}),X_{2}^{n}(M_{2X},B_{2X}),X_{3}^{n}(M_{3X},B_{3X}),\underline{Y}^{n}) \notin T_{2\eta_{1}}(Q,X_{1},U_{2},U_{3},\underline{X},\underline{Y}) \right\}\\
&&\epsilon_{41} \define \bigcup_{\hatm_{1} \neq M_{1}}\bigcup_{a^{s_{3}} \in \SemiPrivateRVSet_{3}^{s_{3}}}\left\{  (q^{n},U_{\oplus}^{n}(a^{s_{3}}),X_{1}^{n}(\hatm_{1}),Y_{1}^{n}) \in T_{2\eta_{1}}(Q,U_{2}\oplus U_{3},X_{1},Y_{1})\right\} 
 \nonumber\\
&&\epsilon_{4j} \define \bigcup_{\hatm_{j} \neq M_{j}} \bigcup_{\substack{a^{s_{j}} \in \\C_{j1}(\hat{m}_{j1})}} \bigcup_{b_{jX} \in c_{jX}} \left\{ (q^{n},U_{j}^{n}(a^{s_{j}}),X_{j}^{n}(\hatm_{jX},b_{jX}),Y_{j}^{n}) \in T_{2\eta_{1}}(Q,U_{j},V_{j},Y_{j}) \right\} \mbox{ for }j=2,3.
  \nonumber
\end{eqnarray}
Note that $\epsilon \define \underset{j=1}{\overset{3}{\bigcup}}\left(\epsilon_{1j}\cup \epsilon_{2} \cup \epsilon_{3} \cup
\epsilon_{4j}\right)$ contains the error event. We derive an upper bound on the probability of this event by partitioning it appropriately. The following events will aid us identify such a partition. Define $\epsilon_{l}\define\epsilon_{l_{2}}\cup\epsilon_{l_{3}}$, where \begin{eqnarray}
 \label{Eqn:3To1ICPCCAchievabilityListErrorEvent}
 \epsilon_{l_{j}} \define \left\{  \phi_{j}(q^{n},M_{j}) < \mathscr{L}_{j}(n) \right\}, \mbox{ and } \phi_{j}(q^{n},M_{j}) \define \sum_{\substack{(a^{s_{j}},b_{jX})\in\\ C_{j1}(M_{j1}) \times c_{jX}}}1_{\left\{ (q^{n},U_{j}^{n}(a^{s_{j}}),X_{j}^{n}(M_{jX},b_{jX})) \in T_{2\eta}(Q,U_{j},X_{j}) \right\}}. \nonumber
\end{eqnarray}
$\mathscr{L}_{j}(n)$ is half of the expected number of jointly typical pairs in the indexed pair of bins.\footnote{Since the precise value of $\mathscr{L}_{j}(n)$ is necessary only in the derivation of the upper bound, it is provided in appendix \ref{AppSec:3to1ICEncoderErrorEventCharacterizingVarianceAndExpectationOfPhi}.} For sufficiently large $n$, we prove $\mathscr{L}_{j}(n) > 2$. For such an $n$, $\epsilon_{1j} \subseteq \epsilon_{l_{j}}:j=2,3$. Since, we can choose $n$ sufficiently large, we will henceforth assume $\epsilon_{1j} \subseteq \epsilon_{l_{j}}:j=2,3$. It therefore suffices to derive upper bounds on $P(\epsilon_{11}), P(\epsilon_{l_{j}}):j=2,3, P(\tilde{\epsilon}_{1}^{c}\cap \epsilon_{2}), P((\tilde{\epsilon}_{1}\cup \epsilon_{2})^{c}\cap \epsilon_{3}), P((\tilde{\epsilon}_{1}\cup \epsilon_{2}\cup \epsilon_{3})^{c}\cap \epsilon_{4j}):j=1,2,3$ where $\tilde{\epsilon_{1}}\define \epsilon_{11} \cup \epsilon_{l}=\epsilon_{11} \cup \epsilon_{l_{2}}\cup \epsilon_{l_{3}}$.

\textit{Upper bound on $P(\epsilon_{11})$} :-- By\ifThesis{ lemma \ref{Lem:ConditionalTypicalSetOccursWithHighProbability}}\fi\ifTITJournal{ conditional frequency typicality \cite[Lemma 5]{201301arXivMACDSTx_PadPra}}\fi, for sufficiently large $n$, $P(\epsilon_{11})\leq \frac{\eta}{32}$.

\textit{Upper bound on $P(\epsilon_{l_{j}}) $} :-- Using a second moment method similar to
that employed in \cite[Appendix A]{201301arXivMACDSTx_PadPra}, we derive an upper bound on
$P( \epsilon_{l_{j}})$ in appendix
\ref{AppSec:3to1ICEncoderErrorEventCharacterizingVarianceAndExpectationOfPhi}. In particular, we prove
\begin{eqnarray}
 \label{Eqn:3to1ICUpperBoundOnEncoderErrorEvent}
 P(\epsilon_{1j})&\leq&  12\exp \left\{ -n\left( \delta-32\eta \right)  \right\}
\end{eqnarray}
for sufficiently large $n$. In deriving the above upper bound, we employed, among others, the bounds \begin{eqnarray}&K_{j}> \delta > 0,~~
 (S_{j}-T_{j}) -\left[ \log \Prime- H(U_{j}|Q)\right] > \delta >0\nonumber\\&(S_{j}-T_{j})+{K_{j}} -\left[ \log \Prime + {H(X_{j}|Q)-H(U_{j},X_{j}|Q)}\right] > \delta >0.\nonumber\end{eqnarray}

\textit{Upper bounds on $P(\tilde{\epsilon}_{1}^{c}\cap \epsilon_{2})$, $P((\tilde{\epsilon}_{1} \cup \epsilon_{2})^{c}\cap \epsilon_{3})$} :-- These events are related to the following two events. (i) The codewords chosen by the distributed encoders are \textit{not} jointly typical, and (ii) the channel produces a triple of outputs that is \textit{not} jointly typical with the chosen and input codewords. In deriving upper bounds on $P(\tilde{\epsilon}_{1}^{c}\cap \epsilon_{2})$, $P((\tilde{\epsilon}_{1} \cup \epsilon_{2})^{c}\cap \epsilon_{3})$, we employ (i) conditional mutual independence of the triplet $X_{1}, (U_{j},X_{j}):j=2,3$ given $Q$ and (ii) the Markov chain $(U_{j}:j=2,3)-\underline{X} - \underline{Y}$. For a technique based on unstructured and independent codes, the analysis of this event is quite standard. However, since our coding technique relies on codewords chosen from statistically correlated codebooks, we present the steps in deriving an upper bound in appendix \ref{AppSec:UpperBoundOnProbabilityOfEpsilon2ComplementIntersectEpsilon3}. In particular, we prove that for sufficiently large $n$, 
\begin{eqnarray}
 \label{Eqn:3To1ICPCCProofMarkovEventFinalBoundInMainLine}
 P(\tilde{\epsilon}_{1}^{c}\cap \epsilon_{2})+P((\tilde{\epsilon}_{1} \cup \epsilon_{2})^{c}\cap \epsilon_{3}) \leq  2\exp \{ -n(n^{2}\mu\eta_{1}^{2}-32\eta) \}+\frac{\eta}{32}.
\end{eqnarray}

\textit{Upper bound on $P((\tilde{\epsilon}_{1}\cup\epsilon_{2}\cup \epsilon_{3})^{c}\cap\epsilon_{41})$} :-- In appendix \ref{AppSec:3To1ICDecoder1PopularErrorEvent}, we prove
\begin{eqnarray}
 \label{Eqn:3To1ICDecoder1ErrorEventSubstitutingValueOfDeltaAndEtas}
 P((\tilde{\epsilon}_{1}\cup\epsilon_{2}\cup \epsilon_{3})^{c}\cap\epsilon_{41})  \leq 4\exp \left\{ -n\left[\delta-28\eta_{1}-12\eta \right]\right\}
\end{eqnarray}
for sufficiently large $n$. In deriving (\ref{Eqn:3To1ICDecoder1ErrorEventSubstitutingValueOfDeltaAndEtas}), we employed, among others, the bounds \begin{eqnarray}\log\Prime+H(X_{1}|Q)-H(X_{1},U_{2}\oplus U_{3}|Y_{1},Q)-(R_{1}+\max\{ S_{2},S_{3}\}) > \delta > 0 , I(X_{1};Y_{1},U_{2}\oplus U_{3}|Q)-R_{1}> \delta > 0.\nonumber\end{eqnarray}
\begin{comment}{
\begin{eqnarray}
 \label{Eqn:3to1ICDecoder1ErrorEventFinalStagesRepeated}
 P(\epsilon_{3}^{c} \cap \epsilon_{41}) \leq \frac{\exp \{ n(29\eta_{4}+\eta_{3}) \}}{\exp \{ n(-\eta_{1}-\eta_{3}\log\Prime) \}} \left( \substack{\exp\{ n(\frac{s_{3}}{n} \log\Prime +R_{1} -\log \Prime -H(X_{1}|Q)+H(X_{1},U_{2}\oplus U_{3}|Y_{1},Q)  ) \}\\+\exp \{ -n(I(X_{1};U_{2}\oplus U_{3},Y_{1}|Q)-R_{1}) \}  }\right)\nonumber\\
\label{Eqn:3To1ICDecoder1ErrorEventSubstitutingValueOfDeltaAndEtas}
\leq 2\exp \{ -n(\delta-29\eta_{4}-\frac{\eta(1+\log\Prime)}{2^{t-1}} ) \}
\end{eqnarray}
where (\ref{Eqn:3To1ICDecoder1ErrorEventSubstitutingValueOfDeltaAndEtas}) follows from substituting $\eta_{1}=\eta_{3}=\frac{\eta}{2^{t}}$ and $\delta$.
}\end{comment}
\textit{Upper bound on $P((\tilde{\epsilon}_{1}\cup  \epsilon_{3})^{c}\cap \epsilon_{4j})$} :-- In appendix \ref{AppSec:AnUpperBoundOnProbabilityOfDecoderErrorEvent}, we prove
\begin{equation}
\label{Eqn:3To1ICDecoder2ErrorEventFinalBoundCopiedFromAppendix}
P((\tilde{\epsilon}_{1}\cup \epsilon_{2}\cup\epsilon_{3})^{c}\cap\epsilon_{4j}) \leq 10 \exp \left\{ -n\left( \delta -\left(9\eta+16\eta_{1}\right)  \right) \right\}
\end{equation}
for sufficiently large $n$. In deriving (\ref{Eqn:3To1ICDecoder2ErrorEventFinalBoundCopiedFromAppendix}), we employed, among others, the bounds 
\begin{eqnarray}
&(\log \Prime -H(U_{j}|X_{j},Y_{j},Q))-S_{j} > \delta > 0 ,
~~~~(\log \Prime +H(X_{j}|Q)-H(U_{j},X_{j}|Y_{j},Q))-(S_{j} +K_{j}) > \delta > 0,
\nonumber\\
&I(X_{1};Y_{1},U_{2}\oplus U_{3}|Q)-R_{1}>\delta > 0,~~~~
(I(X_{j};U_{j},Y_{j}|Q))- (K_{j}+L_{j}) >  \delta > 0,
\nonumber\\
&\left( \log\Prime + H(X_{j}|Q)- H(X_{j},U_{j}|Y_{j},Q)\right) -\left( K_{j}+L_{j} +S_{j}\right) > \delta >0.
\nonumber
\end{eqnarray}
We now collect the derived upper bounds. From (\ref{Eqn:3to1ICUpperBoundOnEncoderErrorEvent}), (\ref{Eqn:3To1ICPCCProofMarkovEventFinalBoundInMainLine}), (\ref{Eqn:3To1ICDecoder1ErrorEventSubstitutingValueOfDeltaAndEtas}) and (\ref{Eqn:3To1ICDecoder2ErrorEventFinalBoundCopiedFromAppendix}), we have
\begin{eqnarray}
 \label{Eqn:3To1ICUpperBoundOnProbabilityOfErrorBySummingAllBounds}
 P(\underset{j=1}{\overset{3}{\cup}}\left(\epsilon_{1j} \cup \epsilon_{3j} \cup
\epsilon_{4j}\right)) \!\!\!\!&\leq& \!\!\!\frac{\eta}{32}+
24\exp \left\{ -n\left( \delta-32\eta \right)  \right\} + 2\exp \{ -n(n^{2}\mu\eta_{1}^{2}-32\eta) \}+\frac{\eta}{32}\nonumber\\\!\!\!\!\!&&\!\!\!\!\!\!\!\!\!\!\!\!\!\!\!\!+4\exp \left\{ -n\left[\delta-28\eta_{1}-12\eta \right]\right\} +20 \exp \left\{ -n\left( \delta -\left(9\eta+16\eta_{1}\right)  \right) \right\}
\nonumber
\end{eqnarray}
The reader may recall that we need $\eta= \frac{1}{2^{d}}\min \{ \tilde{\eta},\delta \}$ and that $\eta_{1} \geq 4\eta$ for the above bounds to hold. The reader may verify that, by choosing $d$ sufficiently large, one can choose $\eta$ and $\eta_{1}\geq 4\eta$ such that the upper bound above decays exponentially. This completes the derivation of an upper bound on the probability of error.

We only need to argue that the chosen input codewords satisfy the cost constraint. For sufficiently large $n$, we have proved that the chosen input codewords are jointly typical with respect to $p_{QU_{2}U_{3}\underline{X}\underline{Y}}$, a distribution that satisfies $\Expectation\left\{ \kappa_{j}(X_{j})\right\} \leq \tau_{j}$. Using standard typicality arguments and finiteness of $\max \left\{ \kappa_{k}(x_{k}):x_{k} \in \InputAlphabet_{k}:k \in [3] \right\}$, it is straight forward to show that the average cost of the codeword input by encoder $j$ is close to $\tau_{j}$ per symbol.

\section{Upper bound on $P(\epsilon_{l_{j}})$}
\label{AppSec:3to1ICEncoderErrorEventCharacterizingVarianceAndExpectationOfPhi}
Recall
\begin{eqnarray}
 \label{Eqn:3to1ICEventEpsilon1jDefnOfPhi}
 \phi_{j}(q^{n},M_{j}) \define \sum_{a^{s_{j}}\in \SemiPrivateRVSet^{s_{j}}}\sum_{b_{jX} \in c_{jX}}1_{\{I_{j}(a^{s_{j}})=M_{j1},(q^{n},U_{j}^{n}(a^{s_{j}}),X_{j}^{n}(M_{jX},b_{jX})  ) \in T_{2\eta}(Q,U_{j},X_{j})\}},~~
 \mathscr{L}_{j}(n) \define \frac{1}{2}\Expectation\left\{ \phi_{j}(q^{n},M_{j}) \right\}
 \nonumber
\end{eqnarray}
and $\epsilon_{l_{j}}=\left\{ \phi_{j}(q^{n},M_{j})< \mathscr{L}_{j}(n) \right\}$. Employing Cheybyshev's inequality, we have
\begin{eqnarray}
 \label{Eqn:3to1ICUsingCheybyshevInequalityForEncoderErrorEvent}
 P(\epsilon_{l_{j}}) = P(\phi_{j}(q^{n},M_{j})< \mathscr{L}_{j}(n)) \leq P(|\phi_{j}(q^{n},M_{j})-\Expectation \{ \phi_{j}(q^{n},M_{j}) \}|\geq \frac{1}{2}\Expectation \{ \phi_{j}(q^{n},M_{j}) \}) \leq \frac{4\mbox{Var}\{ \phi_{j}(q^{n},M_{j})\}}{\left(\Expectation\{ \phi_{j}(q^{n},M_{j}) \}\right)^{2}}.
 \nonumber
\end{eqnarray}
Note that $\mbox{Var}\left\{\phi_{j}(q^{n},M_{j})\right\} = \mathscr{T}_{0}+\mathscr{T}_{1}+\mathscr{T}_{2}+\mathscr{T}_{3}-\mathscr{T}_{0}^{2}$, where
\begin{eqnarray}
 \label{Eqn:3to1ICUpperBoundOnEncoderErrorEventPhiSquared}
\mathscr{T}_{0}&=& \sum_{a^{s_{j}}\in \SemiPrivateRVSet^{s_{j}}}\sum_{b_{jX} \in c_{jX}}\sum_{\substack{(u_{j}^{n},x_{j}^{n}) \in\\T_{2\eta}(U_{j},X_{j}|q^{n}) }}P\left(\substack{I_{j}(a^{s_{j}})=M_{j1},U_{j}^{n}(a^{s_{j}})=u_{j}^{n}\\X_{j}^{n}(M_{jX},b_{jX})=x_{j}^{n}}\right)=\Expectation\{ \phi_{j}(q^{n},M_{j}) \},
 \\
\mathscr{T}_{1}&=& \sum_{a^{s_{j}}\in \SemiPrivateRVSet^{s_{j}}}\sum_{\substack{b_{jX},\tilde{b}_{jX} \in c_{jX}\\b_{jX}\neq\tilde{b}_{jX}}}\sum_{\substack{(u_{j}^{n},x_{j}^{n}),(u_{j}^{n},\tilde{x}_{j}^{n}) \in\\T_{2\eta}(U_{j},X_{j}|q^{n}) }}P\left(\substack{I_{j}(a^{s_{j}})=M_{j1},X_{j}^{n}(M_{jX},b_{jX})=x_{j}^{n},\\U_{j}^{n}(a^{s_{j}})=u_{j}^{n},X_{j}^{n}(M_{jX},\tilde{b}_{jX})=\tilde{x}_{j}^{n}}\right),
\nonumber \\
\mathscr{T}_{2}&=&\sum_{\substack{a^{s_{j}},\tilde{a}^{s_{j}}\in \SemiPrivateRVSet^{s_{j}}\\a^{s_{j}}\neq \tilde{a}^{s_{j}}}}\sum_{b_{jX} \in c_{jX}}\sum_{\substack{(u_{j}^{n},x_{j}^{n}),(\tilde{u}_{j}^{n},x_{j}^{n}) \in\\T_{2\eta}(U_{j},X_{j}|q^{n}) }}P\left(\substack{I_{j}(a^{s_{j}})=M_{j1},I_{j}(\tilde{a}^{s_{j}})=M_{j1},U_{j}^{n}(a^{s_{j}})=u_{j}^{n},\\X_{j}^{n}(M_{jX},b_{jX})=x_{j}^{n},U_{j}^{n}(\tilde{a}^{s_{j}})=\tilde{u}_{j}^{n}}\right),
 \nonumber
\end{eqnarray}
\begin{eqnarray}
\mathscr{T}_{3}&=&\sum_{\substack{a^{s_{j}},\tilde{a}^{s_{j}}\in \SemiPrivateRVSet^{s_{j}}\\a^{s_{j}}\neq \tilde{a}^{s_{j}}}}\sum_{\substack{b_{jX},\tilde{b}_{jX} \in c_{jX}\\b_{jX}\neq\tilde{b}_{jX}}}\sum_{\substack{(u_{j}^{n},x_{j}^{n}),(\tilde{u}_{j}^{n},\tilde{x}_{j}^{n}) \in\\T_{2\eta}(U_{j},X_{j}|q^{n}) }}P\left(\substack{I_{j}(a^{s_{j}})=M_{j1},X_{j}^{n}(M_{jX},b_{jX})=x_{j}^{n},U_{j}^{n}(a^{s_{j}})=u_{j}^{n},\\I_{j}(\tilde{a}^{s_{j}})=M_{j1},X_{j}^{n}(M_{jX},\tilde{b}_{jX})=\tilde{x}_{j}^{n},U_{j}^{n}(\tilde{a}^{s_{j}})=\tilde{u}_{j}^{n}}\right).
 \nonumber
\end{eqnarray}
The codewords of PCC $\Lambda_{j}$ are pairwise independent \cite[Theorem 6.2.1]{Gal-ITRC68}, and therefore
\begin{eqnarray}
 P\left(\substack{I_{j}(a^{s_{j}})=M_{j1},X_{j}^{n}(M_{jX},b_{jX})=x_{j}^{n},U_{j}^{n}(a^{s_{j}})=u_{j}^{n},\\I_{j}(\tilde{a}^{s_{j}})=M_{j1},X_{j}^{n}(M_{jX},\tilde{b}_{jX})=\tilde{x}_{j}^{n},U_{j}^{n}(\tilde{a}^{s_{j}})=\tilde{u}_{j}^{n}}\right) = P\left(\substack{I_{j}(a^{s_{j}})=M_{j1},U_{j}^{n}(a^{s_{j}})=u_{j}^{n}\\X_{j}^{n}(M_{jX},b_{jX})=x_{j}^{n}}\right)P\left(\substack{I_{j}(\tilde{a}^{s_{j}})=M_{j1},U_{j}^{n}(\tilde{a}^{s_{j}})=\tilde{u}_{j}^{n},\\X_{j}^{n}(M_{jX},\tilde{b}_{jX})=\tilde{x}_{j}^{n}}\right).
 \nonumber
\end{eqnarray}
It can be verified that $\mathscr{T}_{3} \leq \mathscr{T}_{0}^{2}$, and therefore, $P(\epsilon_{1j}) \leq 4\frac{\mathscr{T}_{0}+\mathscr{T}_{1}+\mathscr{T}_{2}}{\mathscr{T}_{0}^{2}}$. For sufficiently large $n$, we employ upper bounds on conditional probability and the number of conditional typical sequences to conclude
\begin{eqnarray}
\label{Eqn:3to1ICUpperBoundOnEncoderErrorUpperBoundOnT0}
\mathscr{T}_{0} &\geq& \frac{\exp \left\{  -nH(X_{j}|Q)-4n\eta \right\}|c_{jX}||T_{2\eta}(U_{j},X_{j}|q^{n})|}{\Prime^{t_{j}+n-s_{j}}}
\\
\mathscr{T}_{1} &\leq& \frac{\exp \left\{  -2nH(X_{j}|Q)+8n\eta+nH(X_{j}|U_{j},Q)+8n\eta \right\}|c_{jX}|(|c_{jX}|-1)|T_{2\eta}(U_{j},X_{j}|q^{n})|}{\Prime^{t_{j}+n-s_{j}}}
\nonumber\\
\mathscr{T}_{2} &\leq& \frac{\exp \left\{ -nH(X_{j}|Q)+4n\eta+nH(U_{j}|X_{j}Q)+8n\eta \right\}|c_{jX}||T_{2\eta}(U_{j},X_{j}|q^{n})|}{\Prime^{2(t_{j}+n-s_{j})}}.
\nonumber
\nonumber
\end{eqnarray}
For sufficiently large $n$, $\exp\{-4n\eta \}\leq \exp\{-nH(U_{j},X_{j}|Q) \}|T_{2\eta}(U_{j},X_{j}|q^{n})|\leq \exp\{4n\eta\}$. Substituting $S_{j}=\frac{s_{j}\log\Prime}{n}, T_{j}=\frac{t_{j}\log\Prime}{n}$ and $|c_{jX}|=\exp \{nK_{j} \}$, it maybe verified that, for sufficiently large $n$,
\begin{eqnarray}
  P(\epsilon_{1j}) &\leq& 4\exp \left\{-n\left[ S_{j}-T_{j}+K_{j}- \left( \log\Prime+H(X_{j}|Q) -H(U_{j},X_{j}|Q) \right)-8\eta \right]  \right\} +
 \nonumber\\
 &&4\exp \left\{ -n\left[S_{j}-T_{j}-( \log\Prime-{H(U_{j}|Q)} )-{28\eta} \right] \right\} + 4\exp \left\{ -n\left[ K_{j}-32\eta \right] \right\}.\nonumber
\end{eqnarray}
Using the bounds on $S_{j},T_{j}$ and $K_{j}$ as given in definition \ref{Defn:3To1ICAchRegionSimplifiedPreFourierMotzkinBounds} in terms of $\delta$, we have
 \begin{eqnarray}
\label{Eqn:3to1ICSubstitutingUpperBoundsForT0T1T2} 
P(\epsilon_{1j})&\leq&  12\exp \left\{ -n\left( \delta-32\eta \right)  \right\}
\end{eqnarray}
for sufficiently large $n$. Before we conclude this appendix, let us confirm $\mathscr{L}_{j}(n)$ grows exponentially with $n$. This would imply $\epsilon_{1j} \subseteq \epsilon_{l_{j}}$ and therefore $\epsilon_{1j} \cap \epsilon_{l_{j}}^{c} = \phi$, the empty set. From (\ref{Eqn:3to1ICUpperBoundOnEncoderErrorEventPhiSquared}), (\ref{Eqn:3to1ICUpperBoundOnEncoderErrorUpperBoundOnT0}), we haven for sufficiently large $n$,
\begin{eqnarray}
 \mathscr{L}_{j}(n) &=& \frac{1}{2}\Expectation\left\{ \phi_{j}(q^{n},M_{j}) \right\} = \frac{\mathscr{T}_{0}}{2} \geq  \frac{\exp \left\{  -nH(X_{j}|Q)-4n\eta \right\}|c_{jX}||T_{2\eta}(U_{j},X_{j}|q^{n})|}{2\Prime^{t_{j}+n-s_{j}}} \nonumber\\
 \label{Eqn:3To1ICProofOfPCCListSizeIsLargerThan2UsingParametersOfCode}
 &\geq & \frac{1}{2}\exp \left\{n\left[ S_{j}-T_{j}+K_{j}- \left( \log\Prime+H(X_{j}|Q)-H(U_{j},X_{j}|Q)\right)-8\eta \right]  \right\} \geq  \frac{1}{2}\exp \left\{n\left[\delta-8\eta\right]  \right\},
\end{eqnarray}
where, as before, we have employed $S_{j}=\frac{s_{j}\log\Prime}{n}, T_{j}=\frac{t_{j}\log\Prime}{n}$ and $|c_{jX}|=\exp \{nK_{j} \}$, the lower bounds on $|T_{2\eta}(U_{j},X_{j}|q^{n})|$ and the definition of $\delta$.
\section{Upper bounds on $P(\tilde{\epsilon}_{1}^{c}\cap \epsilon_{2})$, $P((\tilde{\epsilon}_{1} \cup \epsilon_{2})^{c}\cap \epsilon_{3})$}
\label{AppSec:UpperBoundOnProbabilityOfEpsilon2ComplementIntersectEpsilon3}
In the first step, we derive an upper bound on $P(\tilde{\epsilon}_{1}^{c}\cap \epsilon_{2})$, where $\tilde{\epsilon}_{1} = \epsilon_{1} \cup \epsilon_{l}$, and
\begin{equation}
 \label{Eqn:3to1ICErrorEventCharacterizingMarkovLemma}
 \epsilon_{2} = \left\{ (q^{n},U_{2}^{n}(A^{s_{2}}),U_{3}^{n}(A^{s_{3}}),X_{1}^{n}(M_{1}),X_{2}^{n}(M_{2X},B_{2X}),X_{3}^{n}(M_{3X},B_{3X})) \notin T_{\eta_{1}}(Q,U_{2},U_{3},\underline{X}) \right\} .
\end{equation}
was defined in (\ref{Eqn:3To1ICPCCProofEpsilon2DefinitionInMainBody}). In the second step, we employ the result of conditional frequency typicality \cite[Lemma 4 and 5]{201301arXivMACDSTx_PadPra} to provide an upper bound on $P((\epsilon_{1}\cup\epsilon_{l_{2}}\cup\epsilon_{l_{3}}\cup \epsilon_{2})^{c}\cap (\epsilon_{31}\cup\epsilon_{32}\cup\epsilon_{33}))$.

As an astute reader might have guessed, the proof of first step will employ conditional independence of the triple $X_{1}, (U_{2},X_{2}),(U_{3},X_{3})$ given $Q$. The proof is non-trivial because of statistical dependence of the codebooks. We begin with the definition
\begin{equation}
 \label{Eqn:3to1ICSetOfVectorsThatAreIndividuallyTypicalButNotJointly}
 \Theta(q^{n}) \define \left\{
 \begin{array}{l}
 (u_{2}^{n},u_{3}^{n},\underline{x}^{n}) \in {\SemiPrivateRVSet}_{2}^{n} \times {\SemiPrivateRVSet}_{3}^{n} \times \underline{\InputAlphabet}^{n}: (q^{n},{u}_{j}^{n},x_{j}^{n}) \in T_{2\eta}(Q,{\SemiPrivateRV}_{j},\InputRV_{j}):j=2,3 \\(q^{n},x_{1}^{n}) \in T_{2\eta}(Q,X_{1}), (q^{n},{u}_{2}^{n},{u}_{3}^{n},\underline{x}^{n}) \notin T_{\eta_{1}}(Q,{U}_{2},{U}_{3},\underline{X})
 \end{array}
  \right\}.
 \nonumber
\end{equation}
Observe that
\begin{eqnarray}
 \label{Eqn:3to1ICStepsInProvingMarkovLemma}
 \lefteqn{P(\tilde{\epsilon}_{1}^{c}\cap \epsilon_{2}) =\sum_{\substack{({u}_{2}^{n},{u}_{3}^{n},\underline{x}^{n}) \\\in \Theta(q^{n})}}
 P\left(\substack{ I_{j}(A^{s_{j}})=M_{j1},{U}_{j}^{n}(A^{s_{j}})=u_{j}^{n},X_{j}^{n}(M_{jX},B_{jX})=x_{j}^{n}\\\phi_{j}(q^{n},M_{j})\geq \frac{1}{2}\Expectation\left\{ \phi_{j}(q^{n},M_{j}) \right\}:j=2,3,X_{1}^{n}(M_{1})=x_{1}^{n} }  \right)}
 \nonumber\\
&=& \!\!\!\!\!\!\!\!\sum_{\substack{({u}_{2}^{n},{u}_{3}^{n},\underline{x}^{n}) \\\in \Theta(q^{n})}}P\left(\bigcup_{a^{s_{2}}\in \SemiPrivateRVSet_{2}^{s_{2}}} \bigcup_{a^{s_{3}}\in \SemiPrivateRVSet_{3}^{s_{3}}} \bigcup_{\substack{ b_{2X}\in\\ c_{2X}}}\bigcup_{\substack{b_{3X}\in \\c_{3X}}}  \left\{ \substack{I_{j}(a^{s_{j}})=M_{j1},{U}_{j}^{n}(a^{s_{j}})=u_{j}^{n},X_{j}^{n}(M_{jX},b_{jX})=x_{j}^{n} , A^{s_{j}}=a^{s_{j}}\\\phi_{j}(q^{n},M_{j})\geq \frac{1}{2}\Expectation\left\{ \phi_{j}(q^{n},M_{j}) \right\},B_{jX}=b_{jX}:j=2,3,X_{1}^{n}(M_{1})=x_{1}^{n}  }  \right\}\right) \nonumber\\
 &\leq&\!\!\!\!\!\!\!\!\sum_{\substack{({u}_{2}^{n},{u}_{3}^{n},\underline{x}^{n}) \\\in \Theta(q^{n})}}\sum_{\substack{a^{s_{2}} \in \\\mathcal{U}_{2}^{s_{2}}}}\sum_{\substack{a^{s_{3}} \in \\\mathcal{U}_{3}^{s_{3}}}}\sum_{\substack{b_{2X}\in\\ c_{2X}}}\sum_{\substack{b_{3X}\in \\c_{3X}}}
 P\left(\substack{I_{j}(a^{s_{j}})=M_{j1},{U}_{j}^{n}(a^{s_{j}})=u_{j}^{n}\\X_{j}^{n}(M_{jX},b_{jX})=x_{j}^{n},2\phi_{j}(q^{n},M_{j})\geq \\\Expectation\left\{ \phi_{j}(q^{n},M_{j}) \right\} :j=2,3,X_{1}^{n}(M_{1})=x_{1}^{n}   } \right)P\left(\substack{A^{s_{j}}=a^{s_{j}}\\B_{jX}=b_{jX}\\:j=2,3}\middle|\substack{I_{j}(a^{s_{j}})=M_{j1},{U}_{j}^{n}(a^{s_{j}})=u_{j}^{n}\\X_{j}^{n}(M_{jX},b_{jX})=x_{j}^{n},2\phi_{j}(q^{n},M_{j})\geq \\\Expectation\left\{ \phi_{j}(q^{n},M_{j}) \right\} :j=2,3,X_{1}^{n}(M_{1})=x_{1}^{n}      } \right)\nonumber\end{eqnarray}\begin{eqnarray}
 \label{Eqn:3to1ICMarkovLemmaAfterApplyingTheUnionBound}
  &\leq&\!\!\!\!\!\!\!\!\sum_{\substack{({u}_{2}^{n},{u}_{3}^{n},\underline{x}^{n}) \\\in \Theta(q^{n})}}\sum_{\substack{a^{s_{2}} \in \\\mathcal{U}_{2}^{s_{2}}}}\sum_{\substack{a^{s_{3}} \in \\\mathcal{U}_{3}^{s_{3}}}}\sum_{\substack{b_{2X}\in\\ c_{2X}}}\sum_{\substack{b_{3X}\in \\c_{3X}}}
 P\left(\substack{I_{j}(a^{s_{j}})=M_{j1},{U}_{j}^{n}(a^{s_{j}})=u_{j}^{n}\\X_{j}^{n}(M_{jX},b_{jX})=x_{j}^{n}:j=2,3\\X_{1}^{n}(M_{1})=x_{1}^{n}   } \right)\prod_{j=2}^{3}P\left(\substack{A^{s_{j}}=a^{s_{j}}\\B_{jX}=b_{jX}}\middle|\substack{I_{j}(a^{s_{j}})=M_{j1}\\\phi_{j}(q^{n},M_{j})\geq \frac{1}{2}\Expectation\left\{ \phi_{j}(q^{n},M_{j}) \right\} } \right).
\end{eqnarray}
Let us now evaluate a generic term in the above sum (\ref{Eqn:3to1ICMarkovLemmaAfterApplyingTheUnionBound}). Since the codebooks $\mathcal{C}_{1},\mathcal{C}_{2},\mathcal{C}_{3},\Lambda_{2},\Lambda_{3}$ are mutually independent, the probability of the event in question factors as
\begin{equation}
 \label{Eqn:3to1ICMarkovLemmaGenericTermAfterUnionBoundFactorsAs}
P\left(\substack{ U_{j}^{n}(a^{s_{j}})=u_{j}^{n}, X_{j}^{n}(M_{jX},b_{jX})=x_{j}^{n},\\I_{j}(a^{s_{j}})=M_{j1}:j=2,3,X_{1}^{n}(M_{1})=x_{1}^{n}  } \right) = P(X_{1}^{n}(M_{1})=x_{1}^{n}) P\left(\substack{U_{j}^{n}(a^{s_{j}})=u_{j}^{n},\\I_{j}(a^{s_{j}})=M_{j1}}:j=2,3\right)\prod_{j=2}^{3}P(X_{j}^{n}(M_{jX},b_{jX})=x_{j}^{n})
\nonumber
\end{equation}
Furthermore, (i) mutual independence of $I_{j}(a^{s_{j}}):a^{s_{j}} \in \SemiPrivateRVSet_{j}^{s_{j}}:j=2,3, G_{3},B_{2}^{n},B_{3}^{n}$, (ii) uniform distribution of the indices $I_{j}(a^{s_{j}}):a^{s_{j}} \in \SemiPrivateRVSet_{j}^{s_{j}}:j=2,3$ and (iii) distribution of codewords in $\mathcal{C}_{j}:j=1,2,3$ imply 
\begin{equation}
 \label{Eqn:3to1ICMarkovLemmaSimplyfingTheGenericTermInTheSummation}
P\left(\substack{ U_{j}^{n}(a^{s_{j}})=u_{j}^{n}, X_{j}^{n}(M_{jX},b_{jX})=x_{j}^{n},\\I_{j}(a^{s_{j}})=M_{j1}:j=2,3,X_{1}^{n}(M_{1})=x_{1}^{n}  } \right) = P(U_{j}^{n}(a^{s_{j}})=u_{j}^{n}:j=2,3)\frac{\prod_{j=1}^{3}\prod_{t=1}^{n}p_{X_{j}|Q}(x_{jt}|q_{t}) }{\Prime^{t_{2}+t_{3}}}
\end{equation}
The following simple lemma enables us to characterize $P(U_{j}^{n}(a^{s_{j}})=u_{j}^{n}:j=2,3)$.
\begin{lemma}
 \label{Lem:3to1ICAnyPairOfCodewordsInTheTwoCloudCenterCodewordsAreIndependent}
Let $s_{2},s_{3},n \in \naturals$ be such that $s_{2} \leq s_{3}$. Let $G_{3}^{T} \define [ G_{2}^{T}~~ G_{3/2}^{T}]  \in \fieldpi^{s_{3} \times n}$ be a random matrix such that $G_{2} \in \fieldpi^{s_{2} \times n}$ and $B_{2}^{n},B_{3}^{n} \in \fieldpi^{n}$ be random vectors such that $G_{3},B_{2}^{n},B_{3}^{n}$ be mutually independent and uniformly distributed over their respective range spaces. For $j=2,3$ and any $a^{s_{j}} \in \fieldpi^{s_{j}}$, let $U(a^{s_{j}})\define a^{s_{j}}G_{j}\oplus B_{j}^{n}$ be a random vector in the corresponding coset. Then $P(U_{j}^{n}(a^{s_{j}})=u_{j}^{n}:j=2,3)=\frac{1}{\Prime^{2n}}$.
\end{lemma}
The proof follows from a simple counting argument and is omitted.
\begin{comment}{
\begin{proof}
The proof follows from a simple counting argument. It maybe verified that for every $g_{3} \in \fieldpi^{s_{3} \times n}$, there exists a unique pair of vectors $b_{2}^{n},b_{3}^{n} \in \fieldpi^{n}$ such that $a^{s_{j}}g_{j}\oplus b_{j}^{n}=u_{j}^{n}$ for $j=2,3$. Therefore 
\begin{equation} 
\label{Eqn:3to1ICSimpleCountingArgument}
|\left\{  (g_{3},b_{2}^{n},b_{3}^{n}) \in \fieldpi^{s_{3} \times n} \times \fieldpi^{n} \times \fieldpi^{n}: a^{s_{j}}g_{j}\oplus b_{j}^{n}=u_{j}^{n}\mbox{ for }j=2,3  \right\}|=\Prime^{ns_{3}}.
\nonumber
\end{equation}
Now employing the mutually independence and uniformly distribution of $G_{3},B_{2}^{n},B_{3}^{n}$, we have the probability of the event in question to be
\begin{equation} 
\label{Eqn:3to1ICSimpleCountingArgument}
\frac{|\left\{  (g_{3},b_{2}^{n},b_{3}^{n}) \in \fieldpi^{s_{3} \times n} \times \fieldpi^{n} \times \fieldpi^{n}: a^{s_{j}}g_{j}\oplus b_{j}^{n}=u_{j}^{n}\mbox{ for }j=2,3  \right\}|}{|\left\{  (g_{3},b_{2}^{n},b_{3}^{n}) \in \fieldpi^{s_{3} \times n} \times \fieldpi^{n} \times \fieldpi^{n}  \right\}|}=\frac{\Prime^{ns_{3}}}{\Prime^{ns_{3}+2n}}=\frac{1}{\Prime^{2n}}.
\nonumber
\end{equation}
\end{proof}
}\end{comment}
We therefore have
\begin{eqnarray}
 \label{Eqn:3To1ICPCCProofNumeratorOfFirstStepOfMarkovLemma}
 P\left(\substack{ U_{j}^{n}(a^{s_{j}})=u_{j}^{n}, X_{j}^{n}(M_{jX},b_{jX})=x_{j}^{n},\\I_{j}(a^{s_{j}})=M_{j1}:j=2,3,X_{1}^{n}(M_{1})=x_{1}^{n}  } \right) = \frac{\prod_{j=1}^{3}\prod_{t=1}^{n}p_{X_{j}|Q}(x_{jt}|q_{t}) }{\Prime^{2n+t_{2}+t_{3}}} \leq \frac{\prod_{t=1}^{n}p_{X_{1}|Q}(x_{1t}|q_{t})\exp \left\{ -nH(X_{2}|Q)\right\}}{\exp\left\{-8n\eta+nH(X_{3}|Q) \right\}\Prime^{2n+t_{2}+t_{3}}}
\end{eqnarray}
Encoders $2$ and $3$ choose one among the jointly typical pairs uniformly at random. Hence,
\begin{eqnarray}
 \label{Eqn:3To1ICPCCProofKickingUpTheDenomInMarkovLemmaUsingList}
\prod_{j=2}^{3}P\left(\substack{A^{s_{j}}=a^{s_{j}}\\B_{jX}=b_{jX}}\middle|\substack{I_{j}(a^{s_{j}})=M_{j1}\\\phi_{j}(q^{n},M_{j})\geq \frac{1}{2}\Expectation\left\{ \phi_{j}(q^{n},M_{j}) \right\} } \right) \leq \frac{4}{\Expectation\left\{ \phi_{2}(q^{n},M_{2}) \right\}\Expectation\left\{ \phi_{3}(q^{n},M_{3}) \right\}}.
\end{eqnarray}
It maybe verified from (\ref{Eqn:3to1ICUpperBoundOnEncoderErrorEventPhiSquared}) that
\begin{eqnarray}
 \label{Eqn:3To1ICPCCProofLowerBoundOnExpecttaionOfPhi}
 2\mathscr{L}_{j}(n)=\Expectation\left\{ \phi_{j}(q^{n},M_{j}) \right\} \geq \Prime^{s_{j}-t_{j}-n}|c_{jX}|\exp\left\{ -n(H(X_{j}|Q)+4\eta)  \right\}|T_{2\eta}(U_{j},X_{j}|q^{n})|.
\end{eqnarray}
Substituting (\ref{Eqn:3To1ICPCCProofLowerBoundOnExpecttaionOfPhi}), (\ref{Eqn:3To1ICPCCProofKickingUpTheDenomInMarkovLemmaUsingList}) and (\ref{Eqn:3To1ICPCCProofNumeratorOfFirstStepOfMarkovLemma}) in (\ref{Eqn:3to1ICMarkovLemmaAfterApplyingTheUnionBound}), we have
\begin{eqnarray}
 P(\tilde{\epsilon}_{1}^{c}\cap \epsilon_{2}) &\leq& \sum_{\substack{({u}_{2}^{n},{u}_{3}^{n},\underline{x}^{n}) \\\in \Theta(q^{n})}} \frac{\exp\{n16\eta\}\prod_{t=1}^{n}p_{X_{1}|Q}(x_{1t}|q_{t})}{|T_{2\eta}(U_{2},X_{2}|q^{n})||T_{2\eta}(U_{3},X_{3}|q^{n})|} \nonumber\\
 \label{Eqn:3to1ICMarkovLemmaUsingTheTightnessOfTheSumOfBinningRates}
&\leq& \sum_{\substack{({u}_{2}^{n},{u}_{3}^{n},\underline{x}^{n}) \\\in \Theta(q^{n})}} \prod_{t=1}^{n}p_{X_{1}|Q}(x_{1t}|q_{t}) \frac{\exp \left\{ 24n\eta-n H(U_{3},X_{3}|Q) \right\}}{\exp \left\{ n H(U_{2},X_{2}|Q)\right\}}
\end{eqnarray}
where the last inequality follows from lower bound on size of the conditional typical set\ifThesis{ (lemma \ref{Lem:BoundsOnSizeOfConditionalTypicalSet})}\fi. We now employ the lower bound for conditional probability of jointly typical vectors. In particular, \begin{comment}{there exists $N_{8}(\eta_{2}) \in \naturals$, such that for all $n \geq N_{8}(\eta_{2})$, $j=2,3$, we have}\end{comment}
\begin{equation}
 \label{Eqn:3to1ICMarkovLemmaReplacingLowerBoundForConditionalProbability}
\exp\left\{ -nH(U_{j},X_{j}|Q)-4n\eta \right\} \leq  \prod_{t=1}^{n}p_{U_{j},X_{j}|Q}(u_{jt},x_{jt}|q_{t})  \leq \exp\left\{ -nH(U_{j},X_{j}|Q)+4n\eta \right\}
\end{equation}
for any $(u_{2}^{n},u_{3}^{n},\underline{x}^{n}) \in \Theta (q^{n})$. Substituting lower bound (\ref{Eqn:3to1ICMarkovLemmaReplacingLowerBoundForConditionalProbability}) in (\ref{Eqn:3to1ICMarkovLemmaUsingTheTightnessOfTheSumOfBinningRates}), for $n$ sufficiently large, we have
\begin{eqnarray}
P(\tilde{\epsilon}_{1}^{c}\cap \epsilon_{2}) &\leq& \left[\sum_{\substack{({u}_{2}^{n},{u}_{3}^{n},\underline{x}^{n}) \\\in \Theta(q^{n})}} \prod_{t=1}^{n}p_{X_{1}|Q}(x_{1t}|q_{t}) \prod_{j=2}^{3}\prod_{t=1}^{n}p_{U_{j}X_{j}|Q}(u_{jt}x_{jt}|q_{t}) \right]\exp\left\{ 32n\eta\right\}
 \nonumber\\
\label{Eqn:3to1ICMarkovLemmaEmployingTheConditionalIndependenceGivenQ}
  &\leq& \left[\sum_{\substack{({u}_{2}^{n},{u}_{3}^{n},\underline{x}^{n}) \\\in \Theta(q^{n})}} \prod_{t=1}^{n}p_{X_{1}U_{2}X_{2}U_{3}X_{3}|Q}(x_{1t},u_{2t},x_{2t},u_{3t},x_{3t}|q_{t})  \right]\exp\left\{ 32n\eta\right\},
\end{eqnarray}
where (\ref{Eqn:3to1ICMarkovLemmaEmployingTheConditionalIndependenceGivenQ}) follows from conditional mutual independence of the triple $X_{1},(U_{2},X_{2})$ and $(U_{3},X_{3})$ given $Q$. We now employ the exponential upper bound due to Hoeffding \cite{1965MMAMS_Hoe}, Sanov \cite{1957MMMS_San}. Under the condition $\eta_{1} \geq 4\eta$, a `conditional version' of Sanov's lemma \cite{1957MMMS_San} guarantees
\begin{equation}
 \label{Eqn:3to1ICUsingTheExponentialUpperBoundDueToHoeffdingSanovinMarkovLemma}
\sum_{\substack{({u}_{2}^{n},{u}_{3}^{n},\underline{x}^{n}) \\\in \Theta(q^{n})}} \prod_{t=1}^{n}p_{X_{1}U_{2}X_{2}U_{3}X_{3}|Q}(x_{1t},u_{2t},x_{2t},u_{3t},x_{3t}|q_{t}) \leq 2\exp \{ -n^{3}\mu\eta_{1}^{2} \}
\end{equation}
for sufficiently large $n$. Thus we conclude
\begin{equation}
 \label{Eqn:3to1ICUpperBoundOnMarkovLemmaEvent}
P(\tilde{\epsilon}_{1}^{c}\cap \epsilon_{2}) \leq 2\exp \{ -n(n^{2}\mu\eta_{1}^{2}-32\eta) \}
\end{equation}
for such an $n$.

This gets us to the second step where we seek an upper bound on $P((\tilde{\epsilon}_{1} \cup \epsilon_{2})^{c}\cap \epsilon_{3})$, where 
\begin{equation}
 \label{Eqn:3to1ICMarkovLemmaSecondStepCharacterizingASuperSet}
\epsilon_{3}=\left\{ (q^{n},U_{2}^{n}(A^{s_{2}}),U_{3}^{n}(A^{s_{3}}),X_{1}^{n}(M_{1}),X_{2}^{n}(M_{2X},B_{2X}),X_{3}^{n}(M_{3X},B_{3X}),\underline{Y}^{n}) \notin T_{2\eta_{1}}(Q,X_{1},U_{2},U_{3},\underline{X},\underline{Y}) \right\}
\end{equation}
was defined in (\ref{Eqn:3To1ICPCCProofEpsilon3DefinitionInMainBody}). Deriving an upper bound on $P((\tilde{\epsilon}_{1} \cup \epsilon_{2})^{c}\cap \epsilon_{3})$ employs conditional frequency typicality \cite[Lemma 4 and 5]{201301arXivMACDSTx_PadPra} and the Markov chain $(Q,U_{2},U_{3})-\underline{X}-\underline{Y}$. In the sequel, we prove $P(\epsilon_{2}^{c} \cap \epsilon_{3}) \leq \frac{\eta}{32}$ for sufficiently large $n$.

If
\begin{eqnarray}
 \label{Eqn:3to1ICSecondStepOfMarkovLemmaInputVectorsAreTypicalButNotOutput}
\overline{\Theta}(q^{n}) \define \left\{ 
\begin{array}{rl}
(u_{2}^{n},u_{3}^{n},\underline{x}^{n},\underline{y}^{n}) \in \SemiPrivateRVSet_{2}^{n} \times \SemiPrivateRVSet_{3}^{n} \times \underline{\InputAlphabet}^{n} \times \underline{\OutputAlphabet}^{n}:& (u_{2}^{n},u_{3}^{n},\underline{x}^{n}) \in T_{\eta_{1}}(U_{2},U_{3},\underline{X}|q^{n}),\\ &(u_{2}^{n},u_{3}^{n},\underline{x}^{n},\underline{y}^{n}) \notin T_{2\eta_{1}}(U_{2},U_{3},\underline{X},\underline{Y}|q^{n})
\end{array}
\right\},\nonumber
\end{eqnarray}
then
\begin{eqnarray}
 \label{Eqn:3to1ICMarkovLemmaSecondStepCharacterizingEpsilon2ComplementIntersectEpsilon3}
\lefteqn{P(\epsilon_{2}^{c}\cap\epsilon_{3}) = \sum_{\substack{(u_{2}^{n},u_{3}^{n},\underline{x}^{n},\underline{y}^{n}) \\ \in \overline{\Theta}(q^{n})}} P \left( U_{j}^{n}(A^{s_{j}})=u_{j}^{n},X_{j}^{n}(M_{jX},B_{jX})=x_{j}^{n}:j=2,3,X_{1}^{n}(M_{1})=x_{1}^{n},\underline{Y}^{n}=\underline{y}^{n} \right)}
\nonumber\\
&=&\sum_{\substack{(u_{2}^{n},u_{3}^{n},\underline{x}^{n},\underline{y}^{n}) \\ \in \overline{\Theta}(q^{n})}} P \left( \substack{U_{j}^{n}(A^{s_{j}})=u_{j}^{n},X_{1}^{n}(M_{1})=x_{1}^{n},\\X_{j}^{n}(M_{jX},B_{jX})=x_{j}^{n}:j=2,3,} \right)\prod_{t=1}^{n}W_{\underline{Y}|\underline{X}}(\underline{y}_{t}|\underline{x}_{t}) 
\nonumber\\
\label{Eqn:3to1ICMarkovLemmaSecondStepUtilizingMarkovChain}
&=&\!\!\!\! \!\!\!\!\sum_{\substack{(u_{2}^{n},u_{3}^{n},\underline{x}^{n},\underline{y}^{n}) \\ \in \overline{\Theta}(q^{n})}} \!\!\!\!\!\!\!\!P \left( \substack{U_{j}^{n}(A^{s_{j}})=u_{j}^{n},X_{1}^{n}(M_{1})=x_{1}^{n},\\X_{j}^{n}(M_{jX},B_{jX})=x_{j}^{n}:j=2,3,} \right)\prod_{t=1}^{n}p_{\underline{Y}|\underline{X}U_{2}U_{3}}(\underline{y}_{t}|\underline{x}_{t},u_{2t},u_{3t})\\
\label{Eqn:MarkovLemmaLastStep}
&\leq &\sum_{\substack{(u_{2}^{n},u_{3}^{n},\underline{x}^{n}) \in\\ T_{\eta_{1}}(U_{2},U_{3},\underline{X}|q^{n})}}  P\left( \substack{U_{j}^{n}(A^{s_{j}})=u_{j}^{n},X_{1}^{n}(M_{1})=x_{1}^{n},\\X_{j}^{n}(M_{jX},B_{jX})=x_{j}^{n}:j=2,3,} \right)\sum_{\substack{y^{n}:y^{n} \notin \\T_{2\eta_{1}}(Y|u_{2}^{n},u_{3}^{n},\underline{x}^{n})}}\prod_{t=1}^{n}p_{\underline{Y}|\underline{X}U_{2}U_{3}}(\underline{y}_{t}|\underline{x}_{t},u_{2t},u_{3t}) \leq \frac{\eta}{32},
\end{eqnarray}
 for sufficiently large $n$, where (\ref{Eqn:3to1ICMarkovLemmaSecondStepUtilizingMarkovChain}) follows from the Markov chain $(Q,U_{2},U_{3})-\underline{X}-\underline{Y}$ and the last inequality in (\ref{Eqn:MarkovLemmaLastStep}) follows from conditional typicality.

\section{An upper bound on $P((\tilde{\epsilon}_{1}\cup\epsilon_{2}\cup \epsilon_{3})^{c}\cap\epsilon_{41})$}
\label{AppSec:3To1ICDecoder1PopularErrorEvent}
In this appendix, our objective is to derive an upper bound on $P((\tilde{\epsilon}_{1}\cup \epsilon_{2}\cup \epsilon_{3})^{c}\cap \epsilon_{41})$. Recall that $\tilde{\epsilon}_{1} = \epsilon_{1}\cup \epsilon_{l}$, 
\begin{eqnarray}
\label{Eqn:3to1ICDecoder1ErrorEventCharacterization}
(\epsilon_{1}\cup\epsilon_{2}\cup\epsilon_{3})^{c} \cap \epsilon_{41} = \bigcup_{a^{s_{3}} \in \SemiPrivateRVSet_{3}^{s_{3}}} \bigcup_{\hat{m}_{1} \neq M_{1}} \left\{ \left(\substack{ U_{j}^{n}(A^{s_{j}}):j=2,3,X_{1}^{n}(M_{1}),\\X_{j}^{n}(M_{jX},B_{jX}),:j=2,3,Y_{1}^{n}}\right) \substack{\in \hat{T}(q^{n}),
\left(\substack{U_{\oplus}^{n}(a^{s_{3}} ),Y_{1}^{n}\\X_{1}^{n}(\hat{m}_{1})} \right) \in T_{4\eta_{1}}(U_{2}\oplus U_{3},Y_{1},X_{1}|q^{n})}
  \right\}.
\nonumber
\end{eqnarray}
where
\begin{eqnarray}
 \label{Eqn:3to1ICTypicalSetContainingEntireBunchOfLegitimateCodewords}
  \hat{T}(q^{n})  \define \left\{ \substack{(u_{2}^{n},u_{3}^{n},\underline{x}^{n},y_{1}^{n}) \in\\ \SemiPrivateRVSet_{2}^{n}\times \SemiPrivateRVSet_{3}^{n}\times \underline{\InputAlphabet}^{n} \times \OutputAlphabet_{1}^{n}}  : \substack{(u_{2}^{n},u_{3}^{n},\underline{x}^{n},y_{1}^{n}) \in T_{2\eta_{1}}(U_{2},U_{3},\underline{X},Y_{1}|q^{n}),(u_{2}^{n},u_{3}^{n},\underline{x}^{n}) \in T_{\eta_{1}}(U_{2},U_{3},\underline{X}|q^{n}) \\(u_{j}^{n},x_{j}^{n}) \in T_{2\eta}(U_{j},X_{j}|q^{n}):j=2,3,x_{1}^{n} \in T_{2\eta}(X_{1}|q^{n})}\right\}.
 \nonumber
\end{eqnarray}
Employing the union bound, we have
\begin{eqnarray}
\label{Eqn:3to1ICDecoder1ErrorEventExplicitCharacterization}
P((\tilde{\epsilon}_{1}\cup \epsilon_{2}\cup \epsilon_{3})^{c}\cap \epsilon_{41}) \leq \sum_{\substack{\hat{a}^{s_{3}} \in \\ \SemiPrivateRVSet_{3}^{s_{3}} }} \sum_{ \substack{ m_{1},\hat{m}_{1}\\ \hat{m}_{1} \neq m_{1} } } \sum_{\substack{ (u_{2}^{n},u_{3}^{n},\underline{x}^{n},y_{1}^{n}) \in \\ \hat{T}(q^{n} ) } } \sum_{\substack{(\hat{u}^{n},\hat{x}_{1}^{n}) \in \\ T_{4\eta_{1}}(U_{2}\oplus U_{3},X_{1}|y_{1}^{n},q^{n}) } }
\!\!\!\!\!\!\!\!\!\!\!\!P \left( \left\{\substack{ X_{j}^{n}(M_{jX},B_{jX})=x_{j}^{n}, U_{j}^{n}(A^{s_{j}})=u_{j}^{n}\\I_{j}(A^{s_{j}})=M_{j1},X_{1}^{n}(M_{1})=x_{1}^{n},U_{\oplus}(\hat{a}^{s_{3}})=\hat{u}^{n}\\X_{1}^{n}(\hat{m}_{1})=\hat{x}_{1}^{n} ,Y_{1}^{n}=y_{1}^{n},M_{1}=m_{1}:j=2,3 } \right\}\cap\epsilon_{l}^{c}\right).
\end{eqnarray}
We evaluate a generic term in the above sum. Defining $\mathscr{S}(\hat{a}^{s_{3} }  ) \define \left\{  (a^{s_{2}},a^{s_{3}}) \in \mathcal{U}_{2}^{s_{2}} \times \mathcal{U}_{3}^{s_{3}} : a^{s_{2}}0^{s_{+}}\oplus a^{s_{3}}\neq \hat{a}^{s_{3}}\right\}$, where $s_{+}\define s_{3}-s_{2}$, $\mathscr{S}^{c}(\hat{a}^{s_{3} }  ) \define \left(\mathcal{U}_{2}^{s_{2}} \times \mathcal{U}_{3}^{s_{3}}\right) \setminus \mathscr{S}(\hat{a}^{s_{3} }  )$, and
\begin{eqnarray}
 \label{Eqn:3To1ICPCCProofDecoder1ErrorEventComplicatedEvent}
 E \define \left\{ \substack{ X_{j}^{n}(m_{jX},b_{jX})=x_{j}^{n}, U_{j}^{n}(a^{s_{j}})=u_{j}^{n},M_{j}=m_{j}\\I_{j}(a^{s_{j}})=m_{j1}X_{1}^{n}(m_{1})=x_{1}^{n},U_{\oplus}(\hat{a}^{s_{3}})=\hat{u}^{n},\\X_{1}^{n}(\hat{m}_{1})=\hat{x}_{1}^{n},M_{1}=m_{1}: j=2,3, }\right\}\nonumber
\end{eqnarray}
we have
\begin{eqnarray}
\label{Eqn:3To1ICDecoder1ErrorEventGenericTermInTheSum}
 P \left(\left\{\substack{ X_{j}^{n}(M_{jX},B_{jX})=x_{j}^{n}, U_{j}^{n}(A^{s_{j}})=u_{j}^{n}\\I_{j}(A^{s_{j}})=M_{j1},X_{1}^{n}(M_{1})=x_{1}^{n},U_{\oplus}(\hat{a}^{s_{3}})=\hat{u}^{n}\\X_{1}^{n}(\hat{m}_{1})=\hat{x}_{1}^{n} ,Y_{1}^{n}=y_{1}^{n},M_{1}=m_{1}:j=2,3 } \right\}\cap\epsilon_{l}^{c} \right) = \sum_{\substack{m_{2},m_{3} }}\sum_{\substack{b_{2X},b_{3X}}} \sum_{\substack{(a^{s_{2}},a^{s_{3}}   ) \\\in \mathscr{S}(\hat{a}^{s_{3}})}}
 P \left(E\cap\epsilon_{l}^{c}\cap\left\{\substack{Y_{1}^{n}=y_{1}^{n},A^{s_{j}}=a^{s_{j}}\\B_{jX}=b_{jX}:j=2,3} \right\}\right)
 \nonumber\\
 \label{Eqn:3To1ICDecoder1ErrorEventGenericTermInTheSumSecondTerm}
 +\sum_{\substack{m_{2},m_{3} }}\sum_{\substack{b_{2X},b_{3X}}} \sum_{\substack{(a^{s_{2}},a^{s_{3}}   ) \\\in \mathscr{S}^{c}(\hat{a}^{s_{3}})}} P \left(E\cap\epsilon_{l}^{c}\cap\left\{\substack{Y_{1}^{n}=y_{1}^{n},A^{s_{j}}=a^{s_{j}}\\B_{jX}=b_{jX}:j=2,3} \right\}\right)
 \end{eqnarray}
Note that
\begin{eqnarray}
 \label{Eqn:3to1ICDecoder1ErrorEventChannelLooksOnlyAtTransmittedCodewords}
 &P \left( Y_{1}^{n}=y_{1}^{n} \middle| E\cap\epsilon_{l}^{c}\cap\left\{\substack{A^{s_{j}}=a^{s_{j}}\\B_{jX}=b_{jX}:j=2,3} \right\} \right)
 = W^{n}_{Y_{1}|\underline{X}}(y_{1}^{n}|\underline{x}^{n}),\\
 \label{Eqn:3To1ICDecoder1ErrorEventJustExpandingTheGenericTermIAmBugged}
 &P\left(  E\cap\epsilon_{l}^{c}\cap\left\{\substack{A^{s_{j}}=a^{s_{j}}\\B_{jX}=b_{jX}:j=2,3} \right\}\right) = P(E)P\left( \substack{A^{s_{j}}=a^{s_{j}}\\B_{jX}=b_{jX}:j=2,3}\middle| E\cap\epsilon_{l}^{c} \right)=P(E)\frac{1}{\mathscr{L}_{2}(n)\mathscr{L}_{3}(n)}
 \end{eqnarray}
Moreover, for $ (u_{2}^{n},u_{3}^{n},x_{1}^{n},x_{2}^{n},x_{3}^{n},y_{1}^{n}) \in \hat{T}(q^{n})$, $(\hat{u}^{n},\hat{x}_{1}^{n}) \in T_{4\eta_{1}}(U_{2}\oplus U_{3},X_{1}|y_{1}^{n},q^{n})$, we have
 \begin{eqnarray}
  \label{Eqn:3To1ICDecoder1ErrorEventJustExpandingTheGenericTerm}
  P \left(E \right) \leq \left\{
  \begin{array}{lr}
\frac{P(M_{j}=m_{j}:j=2,3,M_{1}=m_{1})}{\Prime^{3n+t_{2}+t_{3}}\exp \left\{ n(H(X_{1}|Q)+\sum_{j=1}^{3}H(X_{j}|Q) -20\eta_{1}) \right\}}&\mbox{if }(a^{s_{2}},a^{s_{3}}) \in \mathscr{S}(\hat{a}^{s_{3}}),\\
\frac{P(M_{jX}=m_{jX}:j=2,3,M_{1}=m_{1})W^{n}_{Y_{1}|\underline{X}}(y_{1}^{n}|\underline{x}^{n})1_{\{\hat{u}^{n}=u_{2}^{n}\oplus u_{3}^{n}\}}}{\Prime^{2n+t_{2}+t_{3}}\exp \left\{ n(H(X_{1}|Q)+\sum_{j=1}^{3}H(X_{j}|Q) -20\eta_{1}) \right\}}&\mbox{if }(a^{s_{2}},a^{s_{3}}) \in \mathscr{S}^{c}(\hat{a}^{s_{3}})
  \end{array} \right.
 \end{eqnarray}
In deriving the above upper bounds, we have used the upper bound on conditional probability of jointly typical sequences\ifThesis{ proved in lemma \ref{Lem:BoundsOnProbabilityOfTypicalSequence}(iii)}\fi. We have also employed independence of (triple in the former and pair in the latter) codewords in the coset code. Substituting (\ref{Eqn:3to1ICDecoder1ErrorEventChannelLooksOnlyAtTransmittedCodewords}), (\ref{Eqn:3To1ICDecoder1ErrorEventJustExpandingTheGenericTermIAmBugged}) and (\ref{Eqn:3To1ICDecoder1ErrorEventJustExpandingTheGenericTerm}), in (\ref{Eqn:3To1ICDecoder1ErrorEventGenericTermInTheSumSecondTerm}), we have
 \begin{eqnarray}
  \label{Eqn:3to1ICDecoder1ErrorEventSubstitutingUpperBoundOnGenericTerm}
  P \left(\left\{\substack{ X_{j}^{n}(M_{jX},B_{jX})=x_{j}^{n}, U_{j}^{n}(A^{s_{j}})=u_{j}^{n}\\I_{j}(A^{s_{j}})=M_{j1},X_{1}^{n}(M_{1})=x_{1}^{n},U_{\oplus}(\hat{a}^{s_{3}})=\hat{u}^{n}\\X_{1}^{n}(\hat{m}_{1})=\hat{x}_{1}^{n} ,Y_{1}^{n}=y_{1}^{n},M_{1}=m_{1}:j=2,3 } \right\}\cap\epsilon_{l}^{c} \right) \leq \frac{\substack{\Prime^{s_{2}-t_{2}}P(M_{1}=m_{1})W^{n}_{Y_{1}|\underline{X}}(y_{1}^{n}|\underline{x}^{n})|c_{2X}||c_{3X}|}}{\substack{\Prime^{2n+t_{3}}\exp \left\{ n(H(X_{1}|Q)+\sum_{j=1}^{3}H(X_{j}|Q) -20\eta_{1}) \right\}}} \!\!\frac{\left[\frac{\Prime^{s_{3}}}{\Prime^{n}}+ 1_{\{\hat{u}^{n}=u_{2}^{n}\oplus u_{3}^{n}\}}\right]}{\mathscr{L}_{2}(n)\mathscr{L}_{3}(n)}\!.
 \end{eqnarray}
Our next step is to substitute (\ref{Eqn:3to1ICDecoder1ErrorEventSubstitutingUpperBoundOnGenericTerm}) in (\ref{Eqn:3to1ICDecoder1ErrorEventExplicitCharacterization}). Let us restate (\ref{Eqn:3to1ICDecoder1ErrorEventExplicitCharacterization}) below as (\ref{Eqn:3to1ICDecoder1ErrorEventExplicitCharacterizationRestated}) for ease of reference.
\begin{eqnarray}
\label{Eqn:3to1ICDecoder1ErrorEventExplicitCharacterizationRestated}
P((\tilde{\epsilon}_{1}\cup \epsilon_{2}\cup \epsilon_{3})^{c}\cap \epsilon_{41}) \leq \sum_{\substack{\hat{a}^{s_{3}} \in \\ \SemiPrivateRVSet_{3}^{s_{3}} }} \sum_{ \substack{ m_{1},\hat{m}_{1}\\ \hat{m}_{1} \neq m_{1} } } \sum_{\substack{ (u_{2}^{n},u_{3}^{n},\underline{x}^{n},y_{1}^{n}) \in \\ \hat{T}(q^{n} ) } } \sum_{\substack{(\hat{u}^{n},\hat{x}_{1}^{n}) \in \\ T_{4\eta_{1}}(U_{2}\oplus U_{3},X_{1}|y_{1}^{n},q^{n}) } }
\!\!\!\!\!\!\!\!\!\!\!\!P \left( \left\{\substack{ X_{j}^{n}(M_{jX},B_{jX})=x_{j}^{n}, U_{j}^{n}(A^{s_{j}})=u_{j}^{n}\\I_{j}(A^{s_{j}})=M_{j1},X_{1}^{n}(M_{1})=x_{1}^{n},U_{\oplus}(\hat{a}^{s_{3}})=\hat{u}^{n}\\X_{1}^{n}(\hat{m}_{1})=\hat{x}_{1}^{n} ,Y_{1}^{n}=y_{1}^{n},M_{1}=m_{1}:j=2,3 } \right\}\cap\epsilon_{l}^{c}\right).
\end{eqnarray}
We do some spade work before we substitute (\ref{Eqn:3to1ICDecoder1ErrorEventSubstitutingUpperBoundOnGenericTerm}) in (\ref{Eqn:3to1ICDecoder1ErrorEventExplicitCharacterizationRestated}). (\ref{Eqn:3to1ICDecoder1ErrorEventSubstitutingUpperBoundOnGenericTerm}) is a sum of two terms. The first term is not dependent on the arguments of the innermost summation in (\ref{Eqn:3to1ICDecoder1ErrorEventExplicitCharacterizationRestated}). By conditional frequency typicality lemma \cite[Lemma 5]{201301arXivMACDSTx_PadPra}, for sufficiently large $n$ we have $|T_{4\eta_{1}}(U_{2}\oplus U_{3},X_{1}|y_{1}^{n},q^{n})| \leq \exp \left\{ n(H(U_{2}\oplus U_{3},X_{1}|Y_{1},Q))+8\eta_{1}  \right\}$. Substituting this upper bound, the summation in (\ref{Eqn:3to1ICDecoder1ErrorEventExplicitCharacterizationRestated}) corresponding to the first term in (\ref{Eqn:3to1ICDecoder1ErrorEventSubstitutingUpperBoundOnGenericTerm}) is upper bounded by
\begin{equation}
 \label{Eqn:3to1ICDecoder1ErrorEventExplicitCharacterizationWithFirstTermSubstituted}
 \mathscr{T}_{1}\define \sum_{\hat{a}^{s_{3}} }   \sum_{ \substack{ m_{1},\hat{m}_{1}\\ \hat{m}_{1} \neq m_{1} } } \sum_{\substack{ (u_{2}^{n},u_{3}^{n},\underline{x}^{n},y_{1}^{n}) \in \\ \hat{T}(q^{n} ) } } \frac{W^{n}_{Y_{1}|\underline{X}}(y_{1}^{n}|\underline{x}^{n})}{\mathscr{L}_{2}(n)\mathscr{L}_{3}(n)}  \frac{\Prime^{s_{2}+s_{3}}|c_{2X}||c_{3X}|P(M_{1}=m_{1})\exp\{ n(H(U_{2}\oplus U_{3},X_{1}|Y_{1},Q))\}}{\Prime^{3n+t_{2}+t_{3}}\exp \left\{ n(H(X_{1}|Q)+\sum_{j=1}^{3}H(X_{j}|Q) -28\eta_{1}) \right\}}.\nonumber
\end{equation}
The indicator in the second term of (\ref{Eqn:3to1ICDecoder1ErrorEventSubstitutingUpperBoundOnGenericTerm}) restricts the outermost summation in (\ref{Eqn:3to1ICDecoder1ErrorEventExplicitCharacterizationRestated}) to $\hat{x}_{1}^{n} \in T_{4\eta_{1}}(X_{1}|u_{2}^{n}\oplus u_{3}^{n},y_{1}^{n},q^{n})$. As earlier, note that the second term is independent of $\hat{x}_{1}^{n}$. Once again, employing the conditional frequency typicality lemma \cite[Lemma 5]{201301arXivMACDSTx_PadPra}, for sufficiently large $n$, $|T_{4\eta_{1}}(X_{1}|u_{2}^{n}\oplus u_{3}^{n},y_{1}^{n},q^{n})| \leq \exp\left\{ n(H(X_{1}|U_{2}\oplus U_{3},Y_{1},Q)+8\eta_{1})  \right\}$. Substituting this upper bound, the summation in (\ref{Eqn:3to1ICDecoder1ErrorEventExplicitCharacterizationRestated}) corresponding to the second term in (\ref{Eqn:3to1ICDecoder1ErrorEventSubstitutingUpperBoundOnGenericTerm}) is upper bounded by
\begin{equation}
 \label{Eqn:Eqn:3to1ICDecoder1ErrorEventExplicitCharacterizationWithSecondTermSubstituted}
 \mathscr{T}_{2}\define\sum_{\hat{a}^{s_{3}} }   \sum_{ \substack{ m_{1},\hat{m}_{1}\\ \hat{m}_{1} \neq m_{1} } } \sum_{\substack{ (u_{2}^{n},u_{3}^{n},\underline{x}^{n},y_{1}^{n}) \in \\ \hat{T}(q^{n} ) } } \frac{W^{n}_{Y_{1}|\underline{X}}(y_{1}^{n}|\underline{x}^{n})}{\mathscr{L}_{2}(n)\mathscr{L}_{3}(n)}  \frac{\Prime^{s_{2}}|c_{2X}||c_{3X}|P(M_{1}=m_{1})\exp\{ n(H(X_{1}|U_{2}\oplus U_{3},Y_{1},Q))\}}{\Prime^{2n+t_{2}+t_{3}}\exp \left\{ n(H(X_{1}|Q)+\sum_{j=1}^{3}H(X_{j}|Q) -28\eta_{1}) \right\}}.\nonumber
\end{equation}
It can be verified that
\begin{eqnarray}
 \label{Eqn:3To1ICPCCProofDecoder1ErrorEventDecipheringASum}
 \sum_{\substack{ (u_{2}^{n},u_{3}^{n},\underline{x}^{n},y_{1}^{n}) \in \\ \hat{T}(q^{n} ) } } W^{n}_{Y_{1}|\underline{X}}(y_{1}^{n}|\underline{x}^{n}) \leq \min\{ |T_{2\eta}(U_{2},X_{2}|q^{n})||T_{2\eta}(U_{3},X_{3}|q^{n})||T_{2\eta}(X_{1}|q^{n})|, | T_{\eta_{1}}(U_{2},U_{3},\underline{X}|q^{n} )|\}.
\end{eqnarray}
Using (\ref{Eqn:3To1ICPCCProofDecoder1ErrorEventDecipheringASum}) and lower bounds $\mathscr{L}_{j}(n):j=2,3$ from (\ref{Eqn:3To1ICPCCProofLowerBoundOnExpecttaionOfPhi}), we have
\begin{eqnarray}
 \label{Eqn:3To1ICDecoder1ErrorEventUpperBoundOnTerm1}
 \mathscr{T}_{1} \leq 2\frac{\Prime^{s_{3}}\exp \{-n(2H(X_{1}|Q)-8\eta-R_{1})  \}|T_{2\eta}(X_{1}|q^{n})|}{\Prime^{n}\exp \{ -n(H(U_{2}\oplus U_{3},X_{1}|Y_{1},Q)+28\eta_{1}) \}}\leq 2\frac{\Prime^{s_{3}}\exp \{-n(H(X_{1}|Q)-12\eta-R_{1})  \}}{\Prime^{n}\exp \{ -n(H(U_{2}\oplus U_{3},X_{1}|Y_{1},Q)+28\eta_{1}) \}},\nonumber
\end{eqnarray}
where the last inequality above follows from upper bound on $|T_{2\eta}(X_{1}|q^{n})|$\ifThesis{ (Lemma \ref{Lem:BoundsOnSizeOfConditionalTypicalSet})}\fi. An identical sequence of steps yields
\begin{eqnarray}
 \label{Eqn:3To1ICDecoder1ErrorEventUpperBoundOnTerm2}
 \mathscr{T}_{2} \leq 2\frac{ \exp \{ -n(H(X_{1}|Q) -28\eta_{1}-R_{1}) \}}{\exp \{ -n(H(X_{1}|U_{2}\oplus U_{3},Y_{1},Q)+12\eta) \}}.\nonumber
\end{eqnarray}
for sufficiently large $n$. Substituting $\frac{s_{3}\log\Prime}{n} = S_{3}$, we have
\begin{eqnarray}
 \label{Eqn:3to1ICDecoder1ErrorEventFinalStages}
 P((\tilde{\epsilon}_{1}\cup \epsilon_{2}\cup \epsilon_{3})^{c}\cap \epsilon_{41}) \leq 2\exp \{ n(28\eta_{1}+12\eta+S_{3}  +R_{1} -\log \Prime -H(X_{1}|Q)+H(X_{1},U_{2}\oplus U_{3}|Y_{1},Q)  ) \} \nonumber\\+ 2\exp \{ n(28\eta_{1}+12\eta+R_{1}-I(X_{1};U_{2}\oplus U_{3},Y_{1}|Q)) \}  .
 \nonumber
\end{eqnarray}
Employing the definition of $\delta$, we have
\begin{eqnarray}
 \label{Eqn:3to1ICDecoder1ErrorEventResultFromAppendix}
 P((\tilde{\epsilon}_{1}\cup \epsilon_{2}\cup \epsilon_{3})^{c}\cap \epsilon_{41}) \leq 4\exp \left\{ -n\left[\delta-28\eta_{1}-12\eta \right]\right\}.
\end{eqnarray}
for sufficiently large $n$.

\section{An upper bound on $P((\tilde{\epsilon}_{1}\cup \epsilon_{2}\cup\epsilon_{3})^{c}\cap\epsilon_{4j})$}
\label{AppSec:AnUpperBoundOnProbabilityOfDecoderErrorEvent}
% \input{../PhDThesis/3To1ICErrorEventAtDecoder2And3Analysed}
% \begin{comment}{

While it seems that analysis of this event is similar to the error event over a
point-to-point channel, and is therefore straight forward, the structure of the code lends
this considerable complexity. A few remarks are in order. Firstly, the distribution
induced on the codebooks does not lend the bins $C_{j1}(m_{j1}):m_{j1} \in
\mathcal{M}_{j1}$ to be statistically independent. Secondly, since the cloud center and
satellite codebooks are binned, the error event needs to be carefully partitioned and
analyzed separately. 

In this appendix, we seek an upper bound on $P((\tilde{\epsilon}_{1}\cup  \epsilon_{3})^{c}\cap \epsilon_{4j})$ for $j=2,3$. Let $({\epsilon}_{1}\cup  \epsilon_{3})^{c}\cap \epsilon_{4j} = \epsilon_{4j}^{1}\cup\epsilon_{4j}^{2}\cup\epsilon_{4j}^{3}$, where
\begin{eqnarray}
 \label{Eqn:3to1ICEpsilon4jBrokenDown}
\epsilon_{4j}^{1} &\define &\bigcup_{\hatm_{j1} \neq M_{j1}}
\bigcup_{\substack{\hat{a}^{s_{j}} \in \SemiPrivateRVSet_{j}^{s_{j}} }}\bigcup_{\hat{b}_{jX} \in
c_{jX}} \left\{ \substack{ (q^{n},U_{j}(\hat{a}^{s_{j}}),X_{j}(M_{jX},\hat{b}_{jX}),Y_{j}^{n}) \in
T_{4\eta_{1}}(Q,U_{j},V_{j},Y_{j}) ,~(q^{n},U_{j}(A^{s_{j}}),X_{j}^{n}(M_{jX},B_{jX}))\in \\T_{2\eta}(Q,U_{j},X_{j}),~I_{j}(\hat{a}^{s_{j}})=\hat{m}_{j1} , ~ (q^{n}, U_{j}^{n}(A^{s_{j}}),
X_{j}^{n}(M_{jX},B_{jX}), Y_{j}^{n}) \in T_{2\eta_{1}}(Q,U_{j},X_{j},Y_{j})  } \right\},
\nonumber
\end{eqnarray}
\begin{eqnarray}
\epsilon_{4j}^{2}& \define&\bigcup_{\hatm_{jX} \neq
M_{jX}}
\bigcup_{\substack{a^{s_{j}} \in \SemiPrivateRVSet_{j}^{s_{j}} }}\bigcup_{b_{jX} \in
c_{jX}} \left\{ \substack{ (q^{n},U_{j}(a^{s_{j}}),X_{j}(\hatm_{jX},b_{jX}),Y_{j}^{n}) \in
T_{4\eta_{1}}(Q,U_{j},V_{j},Y_{j}),~(q^{n},U_{j}(A^{s_{j}}),X_{j}^{n}(M_{jX},B_{jX}))\in \\T_{2\eta}(Q,U_{j},X_{j}),I_{j}(a^{s_{j}})=M_{j1},~ (q^{n}, U_{j}^{n}(A^{s_{j}}),
X_{j}^{n}(M_{jX},B_{jX}), Y_{j}^{n}) \in T_{2\eta_{1}}(Q,U_{j},X_{j},Y_{j})  } \right\},
\nonumber\\
\epsilon_{4j}^{3} &\define&\bigcup_{\substack{\hatm_{j1} \neq\\
M_{j1}}}\bigcup_{\substack{\hatm_{jX} \neq\\ M_{jX}}}\bigcup_{\substack{a^{s_{j}} \in
\SemiPrivateRVSet_{j}^{s_{j}} }}\bigcup_{b_{jX} \in c_{jX}} \!\!\!\!\!\left\{\substack{
(q^{n},U_{j}(a^{s_{j}}),X_{j}(\hatm_{jX},b_{jX}),Y_{j}^{n}) \in
T_{4\eta_{1}}(Q,U_{j},V_{j},Y_{j}),~(q^{n},U_{j}(A^{s_{j}}),X_{j}^{n}(M_{jX},B_{jX}))\in \\T_{2\eta}(Q,U_{j},X_{j}),~I_{j}(a^{s_{j}})=\hat{m}_{j1} ,~ (q^{n}, U_{j}^{n}(A^{s_{j}}),
X_{j}^{n}(M_{jX},B_{jX}), Y_{j}^{n}) \in T_{2\eta_{1}}(Q,U_{j},X_{j},Y_{j}) }\right\}.
\nonumber
\end{eqnarray}
The event of interest is $\epsilon_{l}^{c}\cap(\epsilon_{4j}^{1}\cup\epsilon_{4j}^{2}\cup\epsilon_{4j}^{3})$. Since $\epsilon_{l_{j}}^{c}\cap (\epsilon_{4j}^{1}\cup\epsilon_{4j}^{2}\cup\epsilon_{4j}^{3})$ contains the above error event, it suffices to derive upper bounds on $P(\epsilon_{l_{j}}^{c}\cap \epsilon_{4j}^{1} ), P(\epsilon_{l_{j}}^{c}\cap \epsilon_{4j}^{2} ),P(\epsilon_{l_{j}}^{c}\cap \epsilon_{4j}^{3} )$. We begin by studying $P(\epsilon_{l_{j}}^{c}\cap \epsilon_{4j}^{1})$. Defining,
\begin{eqnarray}
\label{Eqn:3To1ICDecoderErrorEventFirstTermEpsilon4j1DefiningComplicatedSet}
\tilde{T}(q^{n}) \define \left\{ (u_{j}^{n},x_{j}^{n},y_{j}^{n}) \in T_{2\eta_{1}}(U_{j},X_{j},Y_{j}|q^{n}): (u_{j}^{n},x_{j}^{n}) \in  T_{2\eta}(U_{j},X_{j}|q^{n}) \right\},\mbox{ we have}
\\
 \label{Eqn:3to1ICDecoderErrorEventFirstTermEpsilon4j1}
 P(\epsilon_{l_{j}}^{c}\cap \epsilon_{4j}^{1} ) = P\left(
 \bigcup_{\substack{m_{j1},\hatm_{j1} \in \MessageSet_{j1} \\m_{j1}\neq\hatm_{j1} }} \bigcup_{\substack{\hat{a}^{s_{j}} \\ \in
~\SemiPrivateRVSet_{j}^{s_{j}}}} \bigcup_{\substack{\hat{b}_{jX} \\ \in ~c_{jX}}}
\bigcup_{\substack{ (u_{j}^{n},x_{j}^{n},y_{j}^{n}) \in \\\tilde{T}(q^{n})}}
\bigcup_{\substack{
(\hat{u}_{j}^{n},\hat{x}_{j}^{n}) \in \\ T_{4\eta_{1}}(U_{j},X_{j}|y_{j}^{n},q^{n})}}
\!\!\!\!\!\!\!\!\!\left\{  \substack{U_{j}(A^{s_{j}}
)=u_{j}^{n},U_{j}(\hat{a}^{s_{j}})=\hat{u}_{j}^{n}
,M_{j1}=m_{j1}\\ I_{j}(A^{s_{j}})=m_{j1},Y_{j}^{n}=y_{j}^{n},I_{j}(\hat{a}^{s_{j}})=\hat{m}_{j1},\\X_{j}^{n}(M_{jX},B_{jX})=x_{j}^{n},X_{j}^{
n }(M_{jX},\hat{b}_{jX})=\hat{x}_{j}^{n}} \right\}\cap\epsilon_{l_{j}}^{c}\right) 
 \nonumber\\
\label{Eqn:3to1ICDecoderErrorEventFirstTermEpsilon4j1AfterApplyingUnionBound}
\leq
\sum_{\substack{m_{j1},\hatm_{j1} \in \MessageSet_{j1} \\m_{j1}\neq\hatm_{j1} } } \sum_{\substack{\hat{a}^{s_{j}} \\ \in
~\SemiPrivateRVSet_{j}^{s_{j}}}} \sum_{\substack{\hat{b}_{jX} \\ \in ~c_{jX}}}
\sum_{\substack{ (u_{j}^{n},x_{j}^{n},y_{j}^{n}) \in \\\tilde{T}(q^{n})}}
\sum_{\substack{(\hat{u}_{j}^{n},\hat{x}_{j}^{n}) \in \\ T_{4\eta_{1}}(U_{j},X_{j}|y_{j}^{n},q^{n})}}P\left( \left\{  \substack{U_{j}(A^{s_{j}}
)=u_{j}^{n},U_{j}(\hat{a}^{s_{j}})=\hat{u}_{j}^{n}
,M_{j1}=m_{j1}\\ I_{j}(A^{s_{j}})=m_{j1},Y_{j}^{n}=y_{j}^{n},I_{j}(\hat{a}^{s_{j}})=\hat{m}_{j1},\\X_{j}^{n}(M_{jX},B_{jX})=x_{j}^{n},X_{j}^{
n }(M_{jX},\hat{b}_{jX})=\hat{x}_{j}^{n}}\right\}\cap\epsilon_{l_{j}}^{c}\right).
\end{eqnarray}
We now consider two factors of generic term in the above summation. Since $X_{1}^{n}(M_{1}),X_{\msout{j}}^{n}(M_{\msout{j}X},B_{\msout{j}X})$ is independent of the collection $U_{j}(A^{s_{j}}
),U_{j}(\hat{a}^{s_{j}}),M_{j1}, I_{j}(A^{s_{j}}),I_{j}(\hat{a}^{s_{j}}),X_{j}^{n}(M_{jX},B_{jX}),X_{j}^{
n }(M_{jX},\hat{b}_{jX})$ for any $(\hat{a}^{s_{j}},\hat{b}_{jX})$, and $Y_{1}^{n} - (X_{1}^{n}(M_{1}),X_{j}^{n}(M_{jX},B_{jX}):j=2,3 )- (U_{j}(A^{s_{j}}
),U_{j}(\hat{a}^{s_{j}}),M_{j1}, I_{j}(A^{s_{j}}),I_{j}(\hat{a}^{s_{j}}),X_{j}^{
n }(M_{jX},\hat{b}_{jX}))$ is a Markov chain, we have
\begin{eqnarray}
 \label{Eqn:3To1ICTheReceivedCodewordIsConditionallyIndependentOfAnillegitimateCodewordinCloudCenterCodebook}
P\left( Y_{j}^{n}=y_{j}^{n} \middle|\substack{U_{j}(A^{s_{j}}
)=u_{j}^{n},U_{j}(\hat{a}^{s_{j}})=\hat{u}_{j}^{n}
,M_{j1}=m_{j1}\\ \phi_{j}(q^{n},M_{j})\geq \mathscr{L}_{j}(n),I_{j}(A^{s_{j}})=m_{j1},I_{j}(\hat{a}^{s_{j}})=\hat{m}_{j1},\\X_{j}^{n}(M_{jX},B_{jX})=x_{j}^{n},X_{j}^{
n }(M_{jX},\hat{b}_{jX})=\hat{x}_{j}^{n}}\right) =P\left( Y_{j}^{n}=y_{j}^{n}
|X_{j}^{n}(M_{jX},B_{jX})=x_{j}^{n} \right) =: \hat{\theta}\left(y_{j}^{n}|x_{j}^{n}\right).
\nonumber
\nonumber\end{eqnarray}
By the law of total probability, we have
\begin{eqnarray}
P\left( \substack{U_{j}(A^{s_{j}}
)=u_{j}^{n},U_{j}(\hat{a}^{s_{j}})=\hat{u}_{j}^{n}
,M_{j1}=m_{j1}\\ \phi_{j}(q^{n},M_{j})\geq \mathscr{L}_{j}(n),I_{j}(A^{s_{j}})=m_{j1},I_{j}(\hat{a}^{s_{j}})=\hat{m}_{j1},\\X_{j}^{n}(M_{jX},B_{jX})=x_{j}^{n},X_{j}^{
n }(M_{jX},\hat{b}_{jX})=\hat{x}_{j}^{n}} \right) =\!\!\! \!\!\!\sum_{m_{jX} \in \MessageSet_{jX}}\sum_{a^{s_{j}} \in \SemiPrivateRVSet_{j}^{s_{j}}}\!\!\!\!P\left(\left\{  \substack{U_{j}(a^{s_{j}}
)=u_{j}^{n},U_{j}(\hat{a}^{s_{j}})=\hat{u}_{j}^{n}
,M_{j}=m_{j},B_{jX}=\hat{b}_{jX}\\ A^{s_{j}}=a^{s_{j}},I_{j}(a^{s_{j}})=m_{j1},I_{j}(\hat{a}^{s_{j}})=\hat{m}_{j1},\\X_{j}^{n}(m_{jX},\hat{b}_{jX})=x_{j}^{n},X_{j}^{
n }(m_{jX},\hat{b}_{jX})=\hat{x}_{j}^{n}} \right\}\cap\epsilon_{l_{j}}^{c}\right)  +
\nonumber\\
\label{Eqn:3to1ICDecoderErrorEventFirstTermEpsilon4j1TheSecondTermWhereTheTransmittedVectorIsDifferent}
+\sum_{m_{jX} \in \MessageSet_{jX}}\sum_{a^{s_{j}} \in \SemiPrivateRVSet_{j}^{s_{j}}} \sum_{  \substack{b_{jX} \in c_{jX}\\b_{jX}\neq \hat{b}_{jX}}} P\left( \left\{\substack{U_{j}(a^{s_{j}}
)=u_{j}^{n},U_{j}(\hat{a}^{s_{j}})=\hat{u}_{j}^{n}
,M_{j}=m_{j},B_{jX}=b_{jX}\\ A^{s_{j}}=a^{s_{j}},I_{j}(a^{s_{j}})=m_{j1},I_{j}(\hat{a}^{s_{j}})=\hat{m}_{j1},\\X_{j}^{n}(m_{jX},b_{jX})=x_{j}^{n},X_{j}^{
n }(m_{jX},\hat{b}_{jX})=\hat{x}_{j}^{n}}\right\}\cap\epsilon_{l_{j}}^{c} \right). \nonumber
\end{eqnarray}
Now recognize that a generic term of the sum in (\ref{Eqn:3to1ICDecoderErrorEventFirstTermEpsilon4j1AfterApplyingUnionBound}) is a product of the left hand sides of the above two identities.
Before we substitute the right hand sides of the above two identities in (\ref{Eqn:3to1ICDecoderErrorEventFirstTermEpsilon4j1AfterApplyingUnionBound}), we simplify the terms involved in the second identity (involving the two sums). Denoting
\begin{eqnarray}
 \label{Eqn:3To1ICPCCProofDecoders2And3ComplicatedEvent}
 &E^{1} \define \left\{\substack{U_{j}(a^{s_{j}} )=u_{j}^{n},U_{j}(\hat{a}^{s_{j}})=\hat{u}_{j}^{n} ,M_{j}=m_{j}\\ I_{j}(a^{s_{j}})=m_{j1},I_{j}(\hat{a}^{s_{j}})=\hat{m}_{j1},\\X_{j}^{n}(m_{jX},b_{jX})=x_{j}^{n},X_{j}^{ n }(m_{jX},\hat{b}_{jX})=\hat{x}_{j}^{n}}\right\},\mbox{ we have},\nonumber\\
 \label{Eqn:3to1ICDecoderErrorEventFirstTermEpsilon4j1SimplyfyingTheTermWhereTransmittedVectorIsDifferent}
&P\left(\left\{ \substack{U_{j}(a^{s_{j}}
)=u_{j}^{n},U_{j}(\hat{a}^{s_{j}})=\hat{u}_{j}^{n}
,M_{j}=m_{j},B_{jX}=b_{jX}\\ A^{s_{j}}=a^{s_{j}},I_{j}(a^{s_{j}})=m_{j1},I_{j}(\hat{a}^{s_{j}})=\hat{m}_{j1},\\X_{j}^{n}(m_{jX},b_{jX})=x_{j}^{n},X_{j}^{
n }(m_{jX},\hat{b}_{jX})=\hat{x}_{j}^{n}} \right\}\cap\epsilon_{l_{j}}^{c}\right) \leq P\left( E^{1} \right)P\left( \substack{A^{s_{j}}=a^{s_{j}}\\\\B_{jX}=b_{jX}}\middle| E^{1} \cap\epsilon_{l_{j}}^{c} \right) \mbox{ where},\nonumber\\
&P(E^{1})= P\left(\substack{M_{j}=m_{j}, I_{j}(a^{s_{j}})=m_{j1},
I_{j}(\hat{a}^{s_{j}})=\hat{m}_{j1},\\ X_{j}^{n}(m_{jX},b_{jX})=x_{j}^{n},X_{j}^{
n }(m_{jX},\hat{b}_{jX})=\hat{x}^{n} }\right)P\left(\substack{U_{j}^{n}(\hat{a}^{s_{j}})=\hat{u}_{j}^{n}\\U_{j}(a^{s_{j}}
)=u_{j}^{n}}\right),~~ P\left( \substack{A^{s_{j}}=a^{s_{j}}\\\\B_{jX}=b_{jX}}\middle| E^{1} \cap\epsilon_{l_{j}}^{c} \right) = \frac{1}{\mathscr{L}_{j}(n)} = \frac{2}{\Expectation\left\{ \phi_{j}(q^{n},M_{j}) \right\}}
\end{eqnarray}
Let us work with $P(E^{1})$. If $\hat{m}_{j1} \neq m_{j1}$ and $\hat{a}^{s_{j}}\neq a^{s_{j}}$, then
\begin{eqnarray}
 P\left(\substack{M_{j}=m_{j}, I_{j}(a^{s_{j}})=m_{j1},
I_{j}(\hat{a}^{s_{j}})=\hat{m}_{j1},\\ X_{j}^{n}(m_{jX},b_{jX})=x_{j}^{n},X_{j}^{
n }(m_{jX},\hat{b}_{jX})=\hat{x}^{n} }\right)P\left(\substack{U_{j}^{n}(\hat{a}^{s_{j}})=\hat{u}_{j}^{n}\\U_{j}(a^{s_{j}}
)=u_{j}^{n}}\right) \leq \left\{ 
\begin{array}{lr}
\frac{P(M_{j}=m_{j})\exp\left\{ -n(2H(X_{j}|Q) ) \right\}}{\Prime^{2n+2t_{j}}\exp\{ -n4\eta-n8\eta_{1} \}} & \mbox{ if }\hat{b}_{jX} \neq b_{jX} \\
\frac{P(M_{j}=m_{j})\exp\left\{ -n(H(X_{j}|Q) ) \right\}}{\Prime^{2n+2t_{j}}\exp\{ -n4\eta \}} &\mbox{ otherwise.}
\end{array}
  \right.
\end{eqnarray}
Substituting the above observations in (\ref{Eqn:3to1ICDecoderErrorEventFirstTermEpsilon4j1AfterApplyingUnionBound}), we have
\begin{eqnarray}
\label{Eqn:3to1ICFirstPartOfTheDecoderErrorEventSimplified}
 P(\epsilon_{l_{j}}^{c}\cap \epsilon_{4j}^{1} ) \leq \!\!\!\! \sum_{m_{j} \in \MessageSet_{j}} \sum_{\hat{m}_{j1} \neq m_{j1}} \sum_{\substack{a^{s_{j}},\hat{a}^{s_{j} } \\ a^{s_{j}}\neq\hat{a}^{s_{j} } }   } \sum_{\substack{ b_{jX},\hat{b}_{jX}\\ \hat{b}_{jX} \neq b_{jX}  }} \sum_{\substack{ (u_{j}^{n},x_{j}^{n},y_{j}^{n}) \in \\\tilde{T}(q^{n})}}\!\!\!\!\!\hat{\theta}\left(y_{j}^{n}|x_{j}^{n}\right)\!\!\!\!\!\!\!\!\!\!\!\!\!\!
\sum_{\substack{(\hat{u}_{j}^{n},\hat{x}_{j}^{n}) \in \\ T_{4\eta_{1}}(U_{j},X_{j}|y_{j}^{n},q^{n})}}\!\!\!\!\!\!\!\!\!\!\!\!\!\!\!   \frac{P(M_{j}=m_{j})\exp\left\{ -2nH(X_{j}|Q)  \right\}}{\Prime^{2n+2t_{j}}\exp\{ -n4\eta-n8\eta_{1}\}\mathscr{L}_{j}(n)}  +
\nonumber\\
+ \sum_{m_{j} \in \MessageSet_{j}} \sum_{\hat{m}_{j1} \neq m_{j1}} \sum_{\substack{a^{s_{j}},\hat{a}^{s_{j} } \\ a^{s_{j}}\neq\hat{a}^{s_{j} } }   } \sum_{\substack{ b_{jX} \in c_{jX}  }} \sum_{\substack{ (u_{j}^{n},x_{j}^{n},y_{j}^{n}) \in \\\tilde{T}(q^{n})}}\!\!\!\!\!\!\hat{\theta}\left(y_{j}^{n}|x_{j}^{n}\right)\!\!\!\!\!\!\!\!\!\!\!\!\!\!
\sum_{\substack{\hat{u}_{j}^{n} \in \\ T_{4\eta_{1}}(U_{j}|x_{j}^{n},y_{j}^{n},q^{n}) } }\!\!\!\!\!\!\!\frac{P(M_{j}=m_{j})\exp\left\{ -nH(X_{j}|Q) \right\}}{\Prime^{2n+2t_{j}}\exp\{- n4\eta \}\mathscr{L}_{j}(n)}.
\nonumber
\end{eqnarray}
Using the upper bounds on the size of the conditional frequency
typical sets $T_{4\eta_{1}}(U_{j},X_{j}|y_{j}^{n},q^{n})$ and
$T_{4\eta_{1}}(U_{j}|x_{j}^{n},y_{j}^{n},q^{n})$,
 for sufficiently large $n$ (\cite[Lemma 5]{201301arXivMACDSTx_PadPra}),  we have
\begin{eqnarray}
P(\epsilon_{l_{j}}^{c}\cap \epsilon_{4j}^{1}) 
\leq \sum_{m_{j} \in \MessageSet_{j}} \sum_{\hat{m}_{j1} \neq m_{j1}} \sum_{\substack{a^{s_{j}},\hat{a}^{s_{j} } \\ a^{s_{j}}\neq\hat{a}^{s_{j} } }   } \sum_{\substack{ b_{jX},\hat{b}_{jX}\\ \hat{b}_{jX} \neq b_{jX}  }} 
\sum_{\substack{(u_{j}^{n},x_{j}^{n})\in \\ T_{2\eta}(U_{j},X_{j}|q^{n})}}   \frac{P(M_{j}=m_{j})\exp\left\{  -2nH(X_{j}|Q)+n16\eta_{1} \right\}}{\Prime^{2n+2t_{j}}\exp\{-n4\eta -nH(U_{j},X_{j}|Y_{j},Q) \}\mathscr{L}_{j}(n)}  +
\nonumber\\
+ \sum_{m_{j} \in \MessageSet_{j}} \sum_{\hat{m}_{j1} \neq m_{j1}} \sum_{\substack{a^{s_{j}},\hat{a}^{s_{j} } \\ a^{s_{j}}\neq\hat{a}^{s_{j} } }   } \sum_{\substack{ b_{jX} \in c_{jX}  }} \sum_{\substack{(u_{j}^{n},x_{j}^{n})\in \\ T_{2\eta}(U_{j},X_{j}|q^{n})}} \frac{P(M_{j}=m_{j})\exp\left\{ -nH(X_{j}|Q)+8n\eta_{1} \right\}}{\Prime^{2n+2t_{j}}\exp\{ -n4\eta-nH(U_{j}|X_{j},Y_{j},Q) \}\mathscr{L}_{j}(n)}.
\nonumber
\end{eqnarray}
Substituting the lower bound for $\mathscr{L}_{j}(n)$ from (\ref{Eqn:3To1ICPCCProofLowerBoundOnExpecttaionOfPhi}) and noting that the terms in the summation do not depend on the arguments of the sum, for $n\geq N_{11}(\eta_{1})$, it can be verified that
\begin{eqnarray}
P(\epsilon_{l_{j}}^{c}\cap \epsilon_{4j}^{1} )\leq 2\frac{\Prime^{s_{j}}\exp\left\{ -nH(X_{j}|Q)+8n\eta_{1}+4n\eta \right\}}{\Prime^{n}\exp\{ -nH(U_{j}|X_{j},Y_{j},Q) \}} \left(\frac{\exp\{ -nH(X_{j}|Q)+8n\eta_{1}  \}}{\exp\{-nH(X_{j}|Y_{j},Q)-nK_{j} \}}+1\right).
\nonumber
\end{eqnarray}
Finally, substituting $\frac{s_{j}\log\Prime}{n}=S_{j}$, $\delta$, we
have, for sufficiently large $n$,
\begin{eqnarray}
\lefteqn{P(\epsilon_{l_{j}}^{c}\cap \epsilon_{4j}^{1})\leq 2\exp \{ -n\left[(\log \Prime -H(U_{j}|X_{j},Y_{j},Q))-S_{j}- \left(8\eta_{1}+4\eta\right) \right] \} +}\nonumber\\
  &+& 2\exp \{ -n\left[(\log \Prime +H(X_{j}|Q)-H(U_{j},X_{j}|Y_{j},Q))-(S_{j} +K_{j}) - \left(16\eta_{1}+4\eta\right) \right] \} \nonumber\\
 \label{Eqn:3to1ICDecoderErrorEventFirstTermEpsilon4j1AfterTediousStepsAndSubstitutingUpperBoundOnsjByn}
 &\leq& 4\exp \{ -n\left[ \delta -\left(16\eta_{1}+8\eta\right)  \right] \}.
\end{eqnarray}
We follow a similar sequence of steps to derive an upper bound on
$P(\epsilon_{4j}^{2})$. Defining $\tilde{T}(q^{n})$ as in
(\ref{Eqn:3To1ICDecoderErrorEventFirstTermEpsilon4j1DefiningComplicatedSet}),
we have 
\begin{eqnarray}
% \label{Eqn:3to1ICDecoderErrorSecondEventEpsilon4j2CharacterizedMoreExplicitly}
 \label{Eqn:3to1ICDecoderErrorSecondEventEpsilon4j2CharacterizedMoreExplicitlyAfterUnionBound}
P(\epsilon_{l_{j}}^{c}\cap\epsilon_{4j}^{2}) 
 \leq \sum_{\substack{m_{jX},\hat{m}_{jX} \in \MessageSet_{jX} \\ \hat{m}_{jX} \neq m_{jX} } } \sum_{\substack{ \hat{a}^{s_{j}} \\\in \SemiPrivateRV_{j}^{s_{j}} }} \sum_{\substack{ \hat{b}_{jX} \\ \in c_{jX}   }} 
\sum_{\substack{ (u_{j}^{n},x_{j}^{n},y_{j}^{n}) \in \\\tilde{T}(q^{n})}}
\sum_{\substack{
(\hat{u}_{j}^{n},\hat{x}_{j}^{n}) \in \\ T_{4\eta_{1}}(U_{j},X_{j}|y_{j}^{n},q^{n})}}\!\!\!\!\!\!\!\!\!  P\left(  
 \left\{\substack{X_{j}^{
n }(\hat{m}_{jX},\hat{b}_{jX})=\hat{x}_{j}^{n},U_{j}(\hat{a}^{s_{j}})=\hat{u}_{j}^{n}
,Y_{j}^{n}=y_{j}^{n}\\ I_{j}(A^{s_{j}})=I_{j}(\hat{a}^{s_{j}})=M_{j1},M_{jX}=m_{jX},\\X_{j}^{n}(M_{jX},B_{jX})=x_{j}^{n},U_{j}(A^{s_{j}}
)=u_{j}^{n}}\right\}\cap\epsilon_{l_{j}}^{c} \right)
\end{eqnarray}
We now consider two factors of a generic term in the above sum. Since $X_{1}^{n}(M_{1}),X_{\msout{j}}^{n}(M_{\msout{j}X},B_{\msout{j}X})$ is independent of the collection $X_{j}^{
n }(\hat{m}_{jX},\hat{b}_{jX}),U_{j}(\hat{a}^{s_{j}}), I_{j}(A^{s_{j}}),I_{j}(\hat{a}^{s_{j}}),M_{jX},X_{j}^{n}(M_{jX},B_{jX}),U_{j}(A^{s_{j}}
)$ for any $(\hat{a}^{s_{j}},\hat{b}_{jX})$ as long as $\hat{m}_{jX} \neq M_{jX}$, and $Y_{1}^{n} - (X_{1}^{n}(M_{1}),X_{j}^{n}(M_{jX},B_{jX}):j=2,3 )- (X_{j}^{
n }(\hat{m}_{jX},\hat{b}_{jX}),U_{j}(\hat{a}^{s_{j}}), I_{j}(A^{s_{j}}),I_{j}(\hat{a}^{s_{j}}),$ $M_{jX},X_{j}^{n}(M_{jX},B_{jX}),$ $U_{j}(A^{s_{j}}
))$ is a Markov chain, we have
\begin{eqnarray}
 \label{Eqn:3to1ICDecoder2ErrorEventEpsilon4j2GenericTerm}
 P\left(Y_{j}^{n}=y_{j}^{n}\middle|\left\{\substack{X_{j}^{
n }(\hat{m}_{jX},\hat{b}_{jX})=\hat{x}_{j}^{n},U_{j}(\hat{a}^{s_{j}})=\hat{u}_{j}^{n}
\\ I_{j}(A^{s_{j}})=I_{j}(\hat{a}^{s_{j}})=M_{j1},M_{jX}=m_{jX},\\X_{j}^{n}(M_{jX},B_{jX})=x_{j}^{n},U_{j}(A^{s_{j}}
)=u_{j}^{n}}\right\}\cap\epsilon_{l_{j}}^{c}\right) = P \left(Y_{j}^{n}=y_{j}^{n}|X_{j}^{n}(M_{jX},B_{jX})=x_{j}^{n}\right) =: \hat{\theta} \left(y_{j}^{n}|  x_{j}^{n} \right). 
\nonumber
\end{eqnarray}
By the law of total probability, we have
\begin{eqnarray}
 \label{Eqn:3to1ICDecoder2ErrorEventEpsilon4j2GenericTermProbabilityOfConditioningEvent}
 P\left(\left\{ \substack{X_{j}^{
n }(\hat{m}_{jX},\hat{b}_{jX})=\hat{x}_{j}^{n},U_{j}(\hat{a}^{s_{j}})=\hat{u}_{j}^{n}
\\ I_{j}(A^{s_{j}})=I_{j}(\hat{a}^{s_{j}})=M_{j1},M_{jX}=m_{jX},\\X_{j}^{n}(M_{jX},B_{jX})=x_{j}^{n},U_{j}(A^{s_{j}}
)=u_{j}^{n}} \right\}\cap\epsilon_{l_{j}}^{c}\right) = \sum_{\substack{ m_{j1} \in \MessageSet_{j1}} } \sum_{b_{jX} \in c_{jX}} P\left( \left\{\substack{X_{j}^{
n }(\hat{m}_{jX},\hat{b}_{jX})=\hat{x}_{j}^{n},U_{j}(\hat{a}^{s_{j}})=\hat{u}_{j}^{n}, A^{s_{j}}=\hat{a}^{s_{j}}
\\ I_{j}(\hat{a}^{s_{j}})=M_{j1},M_{j}=m_{j},B_{jX}=b_{jX}\\X_{j}^{n}(m_{jX},b_{jX})=x_{j}^{n},U_{j}(\hat{a}^{s_{j}}
)=u_{j}^{n}} \right\}\cap\epsilon_{l_{j}}^{c}\right) 
\nonumber
\end{eqnarray}
\begin{eqnarray}
+ \sum_{\substack{ m_{j1} \in \MessageSet_{j1}} } \sum_{b_{jX} \in c_{jX}} \sum_{ \substack{a^{s_{j}} \in \SemiPrivateRVSet_{j}^{s_{j}}\\a^{s_{j}}\neq \hat{a}^{s_{j} } }} P\left( \left\{\substack{X_{j}^{
n }(\hat{m}_{jX},\hat{b}_{jX})=\hat{x}_{j}^{n},U_{j}(\hat{a}^{s_{j}})=\hat{u}_{j}^{n}, A^{s_{j}}=a^{s_{j}}
\\ I_{j}(a^{s_{j}})=I_{j}(\hat{a}^{s_{j}})=M_{j1},M_{j}=m_{j},B_{jX}=b_{jX}\\X_{j}^{n}(m_{jX},b_{jX})=x_{j}^{n},U_{j}(a^{s_{j}}
)=u_{j}^{n}} \right\}\cap\epsilon_{l_{j}}^{c}\right).
\nonumber
\end{eqnarray}
Now recognize that a generic term of the sum in (\ref{Eqn:3to1ICDecoderErrorSecondEventEpsilon4j2CharacterizedMoreExplicitlyAfterUnionBound}) is a product of the left hand sides of the above two identities.
Before we substitute the right hand sides of the above two identities in (\ref{Eqn:3to1ICDecoderErrorSecondEventEpsilon4j2CharacterizedMoreExplicitlyAfterUnionBound}), we simplify the terms involved in the second identity (involving the two sums). Denoting
\begin{eqnarray}
 \label{Eqn:3To1ICPCCProofDecoder2And3ErrorEventComplicatedEvent}
 E^{2} \define \left\{ X_{j}^{ n }(\hat{m}_{jX},\hat{b}_{jX})=\hat{x}_{j}^{n},U_{j}(\hat{a}^{s_{j}})=\hat{u}_{j}^{n}, 
I_{j}(a^{s_{j}})=I_{j}(\hat{a}^{s_{j}})=m_{j1},M_{j}=m_{j},
X_{j}^{n}(m_{jX},b_{jX})=x_{j}^{n},U_{j}(a^{s_{j}})=u_{j}^{n} \right\}, 
\end{eqnarray}
and evaluating $P(E_2)$ (similiar to $P(E_1)$), and 
substituting this  in (\ref{Eqn:3to1ICDecoderErrorSecondEventEpsilon4j2CharacterizedMoreExplicitlyAfterUnionBound}), we get
\begin{eqnarray}
 \label{Eqn:3to1ICDecoder2ErrorEventEpsilon4j2BrokenDown}
 P(\epsilon_{l_{j}}^{c}\cap\epsilon_{4j}^{2}) \leq \sum_{m_{j} \in \MessageSet_{j}} \sum_{\hat{m}_{jX} \neq m_{jX}} \sum_{\substack{b_{jX},\hat{b}_{jX}\\\in c_{jX}}} \sum_{ \substack{ a^{s_{j}} \in \\ \SemiPrivateRVSet_{j}^{s_{j}} }}  \sum_{\substack{ (u_{j}^{n},x_{j}^{n},y_{j}^{n}) \in \\\tilde{T}(q^{n})}}\!\!\!\!\hat{\theta} \left(y_{j}^{n}|  x_{j}^{n} \right)\!\!\!\!\!\!\!\!\!\!\!\!\!\!
\sum_{\substack{\hat{x}_{j}^{n} \in \\ T_{4\eta_{1}}(X_{j}|u_{j}^{n},y_{j}^{n},q^{n})}}\!\!\!\!\!\!\!\!\!\!\!\!\!\!\!   \frac{P(M_{j}=m_{j})\exp\left\{ -2nH(X_{j}|Q)  \right\}}{\Prime^{n+t_{j}}\exp\{ -n4\eta-n8\eta_{1}\}\mathscr{L}_{j}(n)}  +
\nonumber\\
+\sum_{m_{j} \in \MessageSet_{j}} \sum_{\hat{m}_{jX} \neq m_{jX}} \sum_{\substack{b_{jX},\hat{b}_{jX}\\\in c_{jX}}} \sum_{ \substack{ a^{s_{j}},\hat{a}^{s_{j}} \in  \SemiPrivateRVSet_{j}^{s_{j}}  \\ a^{s_{j}} \neq \hat{a}^{s_{j} }  } }  \sum_{\substack{ (u_{j}^{n},x_{j}^{n},y_{j}^{n}) \in \\\tilde{T}(q^{n})}}\!\!\!\!\!\!\!\!\hat{\theta} \left(y_{j}^{n}|  x_{j}^{n} \right)\!\!\!\!\!\!\!\!\!\!\!\!\!\!
\sum_{\substack{ ( \hat{u}_{j}^{n}, \hat{x}_{j}^{n}) \in \\ T_{4\eta_{1}}(U_{j},X_{j}|y_{j}^{n},q^{n})}}\!\!\!\!\!\!\!\!\!\!\!\!\!\!\!   \frac{P(M_{j}=m_{j})\exp\left\{ -2nH(X_{j}|Q)  \right\}}{\Prime^{2n+2t_{j}}\exp\{ -n4\eta-n8\eta_{1}\}\mathscr{L}_{j}(n)}.
\nonumber
\end{eqnarray}
We now employ the upper bounds on $|T_{4\eta_{1}}(X_{j}|u_{j}^{n},y_{j}^{n},q^{n})|$ and $|T_{4\eta_{1}}(U_{j},X_{j}|y_{j}^{n},q^{n})|$. For sufficiently large $n$, $|T_{4\eta_{1}}(X_{j}|u_{j}^{n},y_{j}^{n},q^{n})| \leq \exp \left\{n(H(X_{j}|U_{j},Y_{j},Q)+8\eta_{1})  \right\}$ and $|T_{4\eta_{1}}(U_{j},X_{j}|y_{j}^{n},q^{n})| \leq \exp \left\{ n(H(U_{j},X_{j}|Y_{j},Q)+8\eta_{1})  \right\}$ for all $(u_{j}^{n},y_{j}^{n},q^{n}) \in T_{2\eta_{1}}(U_{j},Y_{j},Q)$. For such an $n$, we have
\begin{eqnarray}
P(\epsilon_{l_{j}}^{c}\cap\epsilon_{4j}^{2})&\leq& \sum_{m_{j} \in \MessageSet_{j}} \sum_{\hat{m}_{jX} \neq m_{jX}} \sum_{\substack{b_{jX},\hat{b}_{jX}\\\in c_{jX}}} \sum_{ \substack{ a^{s_{j}} \in \\ \SemiPrivateRVSet_{j}^{s_{j}} }}  \sum_{\substack{ (u_{j}^{n},x_{j}^{n},y_{j}^{n}) \in \\\tilde{T}(q^{n})}}\!\!\!\!\!\!\!\!\hat{\theta} \left(y_{j}^{n}|  x_{j}^{n} \right)\!
  \frac{\Prime^{-n-t_{j}}P(M_{j}=m_{j})\exp\left\{ -2nH(X_{j}|Q)  \right\}}{\exp\{ -n4\eta-n16\eta_{1} -nH(X_{j}|U_{j},Y_{j},Q)\}\mathscr{L}_{j}(n)}  +
\nonumber\\
&&\!\!\!\!\!\!\!\!\!+\sum_{m_{j} \in \MessageSet_{j}} \sum_{\hat{m}_{jX} \neq m_{jX}} \sum_{\substack{b_{jX},\hat{b}_{jX}\\\in c_{jX}}} \sum_{ \substack{ a^{s_{j}},\hat{a}^{s_{j}} \in  \SemiPrivateRVSet_{j}^{s_{j}}  \\ a^{s_{j}} \neq \hat{a}^{s_{j} }  } }  \sum_{\substack{ (u_{j}^{n},x_{j}^{n},y_{j}^{n}) \in \\\tilde{T}(q^{n})}}\!\!\!\!\!\!\!\!\!\hat{\theta} \left(y_{j}^{n}|  x_{j}^{n} \right)  \frac{\Prime^{-2n-2t_{j}}P(M_{j}=m_{j})\exp\left\{ -2nH(X_{j}|Q)  \right\}}{\exp\{ -n4\eta-n16\eta_{1} -nH(X_{j},U_{j}|Y_{j},Q)\}\mathscr{L}_{j}(n)}
\nonumber\\
&&\!\!\!\!\!\!\!\!\!\leq \sum_{m_{j} \in \MessageSet_{j}} \sum_{\hat{m}_{jX} \neq m_{jX}} \sum_{\substack{b_{jX},\hat{b}_{jX}\\\in c_{jX}}} \sum_{ \substack{ a^{s_{j}} \in \\ \SemiPrivateRVSet_{j}^{s_{j}} }}  \sum_{\substack{(u_{j}^{n},x_{j}^{n})\in \\ T_{2\eta}(U_{j},X_{j}|q^{n})}}
  \frac{\Prime^{-n-t_{j}}P(M_{j}=m_{j})\exp\left\{ -2nH(X_{j}|Q)  \right\}}{\exp\{ -n4\eta-n16\eta_{1} -nH(X_{j}|U_{j},Y_{j},Q)\}\mathscr{L}_{j}(n)}  +
\nonumber\\
&&\!\!\!\!\!\!\!\!\!+\sum_{m_{j} \in \MessageSet_{j}} \sum_{\hat{m}_{jX} \neq m_{jX}} \sum_{\substack{b_{jX},\hat{b}_{jX}\\\in c_{jX}}} \sum_{ \substack{ a^{s_{j}},\hat{a}^{s_{j}} \in  \SemiPrivateRVSet_{j}^{s_{j}}  \\ a^{s_{j}} \neq \hat{a}^{s_{j} }  } }  \sum_{\substack{(u_{j}^{n},x_{j}^{n})\in \\ T_{2\eta}(U_{j},X_{j}|q^{n})}} \frac{\Prime^{-2n-2t_{j}}P(M_{j}=m_{j})\exp\left\{ -2nH(X_{j}|Q)  \right\}}{\exp\{ -n4\eta-n16\eta_{1} -nH(X_{j},U_{j}|Y_{j},Q)\}\mathscr{L}_{j}(n)}.
\nonumber
\end{eqnarray}
Substituting the lower bound for $\mathscr{L}_{j}(n)$ from (\ref{Eqn:3To1ICPCCProofLowerBoundOnExpecttaionOfPhi}), we have
\begin{eqnarray}
P(\epsilon_{l_{j}}^{c}\cap\epsilon_{4j}^{2})&\leq& 2\sum_{m_{j} \in \MessageSet_{j}} \sum_{\hat{m}_{jX} \neq m_{jX}} \sum_{\substack{b_{jX},\hat{b}_{jX}\\\in c_{jX}}} \sum_{ \substack{ a^{s_{j}} \in \\ \SemiPrivateRVSet_{j}^{s_{j}} }}  
  \frac{P(M_{j}=m_{j})\exp\left\{ -nH(X_{j}|Q)  +n16\eta_{1}\right\}}{\Prime^{s_{j}}\exp\{ -n8\eta -nH(X_{j}|U_{j},Y_{j},Q)\}|c_{jX}|}  +
\nonumber\\
&&+2\sum_{m_{j} \in \MessageSet_{j}} \sum_{\hat{m}_{jX} \neq m_{jX}} \sum_{\substack{b_{jX},\hat{b}_{jX}\\\in c_{jX}}} \sum_{ \substack{ a^{s_{j}},\hat{a}^{s_{j}} \in  \SemiPrivateRVSet_{j}^{s_{j}}  \\ a^{s_{j}} \neq \hat{a}^{s_{j} }  } }  \ \frac{P(M_{j}=m_{j})\Prime^{-s_{j}}\exp\left\{ -nH(X_{j}|Q)+n16\eta_{1} \right\}}{\Prime^{n+t_{j}}\exp\{ -n8\eta -nH(X_{j},U_{j}|Y_{j},Q)\}|c_{jX}|}
\nonumber
\end{eqnarray}
\begin{eqnarray}
&&\leq 2\sum_{m_{j} \in \MessageSet_{j}} \sum_{\hat{m}_{jX} \neq m_{jX}} 
  \frac{P(M_{j}=m_{j})\exp\left\{ -nH(X_{j}|Q)  +n16\eta_{1}\right\}}{\exp\{ -n8\eta -nH(X_{j}|U_{j},Y_{j},Q)-nK_{j}\}}  +
\nonumber\\
&&+2\sum_{m_{j} \in \MessageSet_{j}} \sum_{\hat{m}_{jX} \neq m_{jX}}   \frac{P(M_{j}=m_{j})\Prime^{s_{j}}\exp\left\{ -nH(X_{j}|Q)+n16\eta_{1} \right\}}{\Prime^{n+t_{j}}\exp\{ -n8\eta -nH(X_{j},U_{j}|Y_{j},Q)-nK_{j}\}}
\nonumber\\
&&\leq 2\frac{\exp\left\{ -nH(X_{j}|Q)+nL_{j}+n\eta_{1}+n16\eta_{1}  \right\}}{\exp\{ -n8\eta -nH(X_{j}|U_{j},Y_{j},Q)-nK_{j}\}}\left[ 1+ \frac{\exp\{ nH(U_{j}|Y_{j},Q) \}}{\Prime^{n+t_{j}-s_{j} } } \right]
\nonumber
\end{eqnarray}
We have for sufficiently large $n$
\begin{eqnarray}
  P(\epsilon_{l_{j}}^{c}\cap\epsilon_{4j}^{2}) &\leq& 2\exp\left\{ -n(I(X_{j};U_{j},Y_{j}|Q)- K_{j}-L_{j} - \left[ 9\eta_{1}+16\eta_{1} \right] )    \right\}
 \nonumber\\\label{Eqn:3to1ICDecoder2ErrorEventEpsilon4j2UpperBound}
 &&+2\exp \left\{ -n \left[\left(\substack{\log\Prime + H(X_{j}|Q)-\\
         H(X_{j},U_{j}|Y_{j},Q)}\right)-\left(
       \substack{K_{j}+L_{j}+\\(S_{j}-T_{j})\log\Prime}\right)
     -\left[(9+16\eta_{1}\right] \right]  \right\} \leq 4\exp \left\{
   -n\left( \delta -\left(9\eta+16\eta_{1}\right)  \right)
 \right\}. \nonumber 
\end{eqnarray}
We are left to study $P(\epsilon_{4j}^{3})$. Defining $\tilde{T}(q^{n})$ as in (\ref{Eqn:3To1ICDecoderErrorEventFirstTermEpsilon4j1DefiningComplicatedSet}), and
\begin{eqnarray}
 \label{Eqn:3To1ICPCCProofDecoder2And3ErrorEventThirdTermComplicatedEvent}
 E^{3}\define \left\{\substack{X_{j}^{
n }(\hat{m}_{jX},\hat{b}_{jX})=\hat{x}_{j}^{n},U_{j}(\hat{a}^{s_{j}})=\hat{u}_{j}^{n}
\\ I_{j}(a^{s_{j}})=m_{j1}, I_{j}(\hat{a}^{s_{j}})=\hat{m}_{j1}\\X_{j}^{n}(m_{jX},b_{jX})=x_{j}^{n}, U_{j}(a^{s_{j}}
)=u_{j}^{n}  ,M_{j}=m_{j}}\right\}
\end{eqnarray}
the union bound yields
\begin{eqnarray}
 \label{Eqn:3to1ICDecoder2ErrorEventEpsilon4j3Description}
 P(\epsilon_{l_{j}}^{c}\cap\epsilon_{4j}^{3}) 
\leq \!\!\!\!\!\!\sum_{ \substack{ m_{j1},\hat{m}_{j1} \\m_{j1}\neq \hat{m}_{j1} } } \sum_{ \substack{   m_{jX},\hat{m}_{jX} \\m_{jX}\neq \hat{m}_{jX}  }  } \sum_{ \substack{ a^{s_{j}},\hat{a}^{s_{j}} \\ \hat{a}^{s_{j}}\neq a^{s_{j}} } } \sum_{ b_{jX},\hat{b}_{jX}   } \sum_{\substack{ (u_{j}^{n},x_{j}^{n},y_{j}^{n}) \in \\\tilde{T}(q^{n}) } }
\sum_{\substack{
(\hat{u}_{j}^{n},\hat{x}_{j}^{n}) \in \\ T_{4\eta_{1}}(U_{j},X_{j}|y_{j}^{n},q^{n}) } } \!\!\!\!\!\!\!\!\!\!\!\!\!\!P\left( \left\{\substack{A^{s_{j}}=a^{s_{j}}\\Y_{j}^{n}=y_{j}^{n},B_{jX}=b_{jX}}\right\}\cap E^{3}\cap\epsilon_{l_{j}}^{c} \right)
\end{eqnarray}
As earlier, we consider a generic term in the above sum and simplify the same. Observe that
\begin{eqnarray}
 \label{Eqn:3to1ICDecoder2ErrorEventEpsilon4j3Terms}
 P\left(  Y_{j}^{n}=y_{j}^{n} \middle| \left\{\substack{A^{s_{j}}=a^{s_{j}}\\B_{jX}=b_{jX}}\right\}\cap E^{3}\cap\epsilon_{l_{j}}^{c} \right) &=& P\left(  Y_{j}^{n}=y_{j}^{n} | X_{j}^{n}(M_{jX},B_{jX})=x_{j}^{n} \right) =: \hat{\theta}\left(  y_{j}^{n} | x_{j}^{n} \right),
\nonumber\\
 \label{Eqn:3to1ICDecoder2ErrorEventEpsilon4j3TermsConditioningEvent}
 P\left(
   \left\{\substack{A^{s_{j}}=a^{s_{j}}\\B_{jX}=b_{jX}}\right\}\cap
   E^{3}\cap\epsilon_{l_{j}}^{c}  \right) &\leq & 
\frac{P(M_{j}=m_{j}) \exp\{ -2nH(X_{j}|Q) \}   }{   \Prime^{2n+2t_{j}} \exp\{ -4n\eta-8n\eta_{1} \}  }\frac{1}{\mathscr{L}_{j}(n)}.
\nonumber
\end{eqnarray}
Substituting the above observations in (\ref{Eqn:3to1ICDecoder2ErrorEventEpsilon4j3Description}), we have
\begin{eqnarray}
 \label{Eqn:3to1ICDecoder2ErrorEventEpsilon4j3AfterUnionBoundAndSubstitutions}
 P(\epsilon_{l_{j}}^{c}\cap\epsilon_{4j}^{3}) 
\leq \!\!\!\!\!\!\sum_{ \substack{ m_{j1},\hat{m}_{j1} \\m_{j1}\neq \hat{m}_{j1} } } \sum_{ \substack{   m_{jX},\hat{m}_{jX} \\m_{jX}\neq \hat{m}_{jX}  }  } \sum_{ \substack{ a^{s_{j}},\hat{a}^{s_{j}} \\ \hat{a}^{s_{j}}\neq a^{s_{j}} } } \sum_{ b_{jX},\hat{b}_{jX}   } \sum_{\substack{ (u_{j}^{n},x_{j}^{n},y_{j}^{n}) \in \\\tilde{T}(q^{n}) } }\!\!\!\!\!\!\!\hat{\theta}\left(  y_{j}^{n} | x_{j}^{n} \right)\!\!\!\!\!\!\!\!\!\!\!\!\!
\sum_{\substack{
(\hat{u}_{j}^{n},\hat{x}_{j}^{n}) \in \\ T_{4\eta_{1}}(U_{j},X_{j}|y_{j}^{n},q^{n}) } } \!\!\!\!\!\!\!\!\!\!\!\frac{P(M_{j}=m_{j}) \exp\{ -2nH(X_{j}|Q) \}   }{   \Prime^{2n+2t_{j}} \exp\{ -4n\eta-8n\eta_{1} \}\mathscr{L}_{j}(n)  }.
\nonumber
\end{eqnarray}
There exists $N_{15}(\eta_{1})\in \naturals$ such that for all $n \geq \max \left\{ N_{12}(\eta), N_{15}(\eta_{1}) \right\}$, we have \[|T_{4\eta_{1}}(U_{j},X_{j}|y_{j}^{n},q^{n})| \leq \exp \left\{ n(H(U_{j},X_{j}|Y_{j},Q)+8\eta_{1})  \right\}\mbox{ for all }(y_{j}^{n},q^{n}) \in T_{2\eta_{1}}(Y_{j},Q)\] and hence
\begin{eqnarray}
P(\epsilon_{l_{j}}^{c}\cap\epsilon_{4j}^{3})& \leq& \!\!\!\!\!\!\sum_{ \substack{ m_{j1},\hat{m}_{j1} \\m_{j1}\neq \hat{m}_{j1} } } \sum_{ \substack{   m_{jX},\hat{m}_{jX} \\m_{jX}\neq \hat{m}_{jX}  }  } \sum_{ \substack{ a^{s_{j}},\hat{a}^{s_{j}} \\ \hat{a}^{s_{j}}\neq a^{s_{j}} } } \sum_{ b_{jX},\hat{b}_{jX}   } \sum_{\substack{(u_{j}^{n},x_{j}^{n})\in \\ T_{2\eta}(U_{j},X_{j}|q^{n})}} \frac{\Prime^{-2n-2t_{j}}P(M_{j}=m_{j})\exp\left\{ -2nH(X_{j}|Q)  \right\}}{\exp\{ -n4\eta-n16\eta_{1} -nH(X_{j},U_{j}|Y_{j},Q)\}\mathscr{L}_{j}(n)}
\nonumber\\
&\leq& 2\sum_{ \substack{ m_{j1},\hat{m}_{j1} \\m_{j1}\neq \hat{m}_{j1} } } \sum_{ \substack{   m_{jX},\hat{m}_{jX} \\m_{jX}\neq \hat{m}_{jX}  }  } \frac{\Prime^{s_{j}}P(M_{j}=m_{j})\exp\left\{ -nH(X_{j}|Q)+n16\eta_{1} \right\}}{\Prime^{n+t_{j}}\exp\{ -n8\eta -nH(X_{j},U_{j}|Y_{j},Q)-nK_{j}\}}
\nonumber\\
&\leq& 2\exp \left\{ -n \left[ \left( \substack{\log\Prime + H(X_{j}|Q)-\\ H(X_{j},U_{j}|Y_{j},Q)}\right) -\left( \substack{K_{j}+L_{j} +\\
\label{Eqn:3To1ICDecoder2ErrorEventFinalBoundOnEpsilon24j}
S_{j}\log\Prime}\right)-\left(\substack{9\eta_{1}+16\eta_{1}\\+\log\Prime\eta_{1}}\right) \right]  \right\} \leq 2\exp \left\{ -n\left( \delta -\left(9\eta+16\eta_{1}\right)  \right) \right\}.
\nonumber \end{eqnarray}
We now collect all the upper bounds derived in (\ref{Eqn:3to1ICDecoderErrorEventFirstTermEpsilon4j1AfterTediousStepsAndSubstitutingUpperBoundOnsjByn}), (\ref{Eqn:3to1ICDecoder2ErrorEventEpsilon4j2UpperBound}) and (\ref{Eqn:3To1ICDecoder2ErrorEventFinalBoundOnEpsilon24j}). For $n \geq \max \left\{ N_{14}(\eta),N_{16}(\eta) \right\}$, we have 
\begin{equation}
\label{Eqn:3To1ICDecoder2ErrorEventFinalBoundFromAppendix}
P((\tilde{\epsilon}_{1}\cup  \epsilon_{3})^{c}\cap \epsilon_{4j}) \leq 10 \exp \left\{ -n\left( \delta -\left(9\eta+16\eta_{1}\right)  \right) \right\}
\end{equation}
% }\end{comment}

\section{Proof of proposition \ref{Prop:ORICStrictSub-OptimalityOfHK}}
\label{AppSec:UpperBoundOnUSBRateRegionFor3To1ORIC}

We begin by stating the conditions for sub-optimality of $\UHK$technique.
\begin{lemma}
 \label{Lem:ConditionsForSub-OptimalityOfHanKobayashiTechnique}
Consider example \ref{Ex:A3To1-OR-IC} with $\delta \define \delta_{2}=\delta_{3} \in (0,\frac{1}{2})$ and $\tau \define \tau_{2}=\tau_{3} \in (0,\frac{1}{2})$. Let $\beta \define \delta_{1}*(2\tau -\tau^{2})$. The rate triple $(h_{b}(\tau_{1} * \delta_{1})-h_{b}(\delta_{1}),h_{b}(\tau * \delta)-h_{b}(\delta),h_{b}(\tau * \delta)-h_{b}(\delta)) \notin \alpha_{u}(\ulinecost)$ if
\begin{equation}
 \label{Eqn:3To1-OR-ICConditionForStrictSubOptimalityOfHanKobayashi}
h_{b}(\tau_{1} * \delta_{1})-h_{b}(\delta_{1})+2(h_{b}(\tau * \delta)-h_{b}(\delta))> h_{b}(\tau_{1}(1-\beta)+(1-\tau_{1})\beta)-h_{b}(\delta_{1})
\end{equation}
In particular, if (\ref{Eqn:3To1-OR-ICConditionForStrictSubOptimalityOfHanKobayashi}) is true, $\alpha_{u}(\ulinecost) \subsetneq \beta(\underline{\tau},\underline{\delta})$, where $\beta(\underline{\tau},\underline{\delta})$ is defined in (\ref{Eqn:StrictSubOptimalityOfUSBTechnique3To1ICOuterBound}).
\end{lemma}
\begin{proof}
 We prove this by contradiction. Suppose $(h_{b}(\tau_{1} * \delta_{1})-h_{b}(\delta_{1}),h_{b}(\tau * \delta)-h_{b}(\delta),h_{b}(\tau * \delta)-h_{b}(\delta)) \in $\newline$\cocl(\alpha^{\three-1}_{f}(p_{QU_{2}U_{3}\ulineInputRV\ulineOutputRV}))$ for some $p_{QU_{2}U_{3}\ulineInputRV\ulineOutputRV} \in \mathbb{D}_{3-1}(\tau_{1},\tau,\tau)$. In the sequel, we characterize such a $p_{QU_{2}U_{3}\ulineInputRV\ulineOutputRV}$ and employ the same to derive a contradiction. Our first claim is that $p_{X_{2}|Q}(1|q)=p_{X_{3}|Q}(1|q)=\tau$ for all $q \in \TimeSharingRVSet$.

From (\ref{Eqn:3To1ICUpperBoundOnR1}) we have 
\begin{eqnarray}
 \label{Eqn:PluggingInBoundOnRjFrom3To1ICHanKobayashiConditions}
&&R_{j} \leq I(U_{j}X_{j};Y_{j}|Q) =H(Y_{j}|Q)-H(Y_{j}|X_{j}U_{j}Q) = H(Y_{j}|Q)-h_{b}(\delta) = \sum_{q \in \TimeSharingRVSet}p_{Q}(q)H(Y_{j}|Q=q) - h_{b}(\delta)\nonumber\\
\label{Eqn:BoundOnRateOfUserjAsASum}
&&= \sum_{q \in \TimeSharingRVSet}p_{Q}(q)H(X_{j}\oplus N_{j}|Q=q) - h_{b}(\delta) \mbox{ for }j=2,3.
\end{eqnarray}
If $\tau_{q}\define p_{X_{j}|Q}(1|q)$, then independence of the pair
$N_{j}$ and $(X_{j},Q)$ implies $p_{X_{j}\oplus
  N_{j}|Q}(1|q)=\tau_{q}(1-\delta)+(1-\tau_{q})\delta=\tau_{q}(1-2\delta)+\delta$. Substituting
the same in (\ref{Eqn:BoundOnRateOfUserjAsASum}), we have 
\begin{eqnarray}
 \label{Eqn:PluggingInBoundOnRjFrom3To1ICHanKobayashiConditions}
R_{j} \leq \sum_{q \in \TimeSharingRVSet}p_{Q}(q)h_{b}(\tau_{q}(1-2\delta)+\delta) - h_{b}(\delta) \leq h_{b}([p_{X_{j}}(1)(1-2\delta)+\delta])-h_{b}(\delta)\nonumber
\end{eqnarray}
from Jensen's inequality. Since $p_{X_{j}}(1) \leq \tau <\frac{1}{2}$,
we have $p_{X_{j}}(1)(1-2\delta)+\delta \leq \tau(1-2\delta)+\delta <
\frac{1}{2}(1-2\delta)+\delta = \frac{1}{2}$.\footnote{Here we have
  used the positivity of $(1-2\delta)$, or equivalently $\delta$ being
  in the range $(0,\frac{1}{2})$.} The term
$h_{b}([p_{X_{j}}(1)(1-2\delta)+\delta])$ is therefore strictly
increasing in $p_{X_{j}}(1)$ and is at most $h_{b}(\tau *
\delta)$. Moreover, the condition for equality in Jensen's inequality
implies $R_{j} = h_{b}(\tau * \delta)-h_{b}(\delta)$ if and only if
$p_{X_{j}|Q}(1|q)=\tau$ for all $q \in \TimeSharingRVSet$ that
satisfies $p_{\TimeSharingRV}(\timeshare)>0$. We have therefore proved
our first claim. 

Our second claim is an analogous statement for $p_{X_{1}|Q}(1|q)$. In particular, our second claim is that 
$p_{X_{1}|Q}(1|q)=\tau_{1}$ for each $q \in \TimeSharingRVSet$ of positive probability. We begin with the upper bound on $R_{1}$ in (\ref{Eqn:3To1ICUpperBoundOnR1}). As in proof of proposition \ref{Prop:StrictSub-OptimalityOfHanKobayashi}, we let $\tilde{\TimeSharingRVSet} \define \TimeSharingRVSet \times \SemiPrivateRVSet_{2} \times \SemiPrivateRVSet_{3}$, $\tilde{q} = (q,u_{2},u_{3}) \in \tilde{\TimeSharingRVSet}$ denote a generic element and $\tilde{Q} \define (\TimeSharingRV,\SemiPrivateRV_{2},\SemiPrivateRV_{3})$. The steps we employ in proving the second claim borrows steps from proof of proposition \ref{Prop:StrictSub-OptimalityOfHanKobayashi} and the proof of the first claim presented above. Note that
\begin{eqnarray}
 \label{Eqn:3To1ORICProvingTheSecondClaimOnR1}
\lefteqn{R_{1} \leq I(X_{1};Y_{1}|\tilde{\TimeSharingRV}) = H(Y_{1}|\tilde{\TimeSharingRV}) -
H(Y_{1}|\tilde{\TimeSharingRV}X_{1})}\nonumber\\
\ifThesis{&=&\sum_{\tilde{\timeshare}}p_{\tilde{\TimeSharingRV}}(\tilde{\timeshare})H(Y_{1}|\tilde{\TimeSharingRV}=\tilde{\timeshare}
)-\sum_{x_{ 1 },\tilde{\timeshare}}p_{\tilde{\TimeSharingRV}X_{1}}(\tilde{\timeshare},x_{1})
H(Y_{1}|X_{1}=x_{1},\tilde{\TimeSharingRV}=\tilde{\timeshare})\nonumber\\
&=&\sum_{\tilde{\timeshare}}p_{\tilde{\TimeSharingRV}}(\tilde{\timeshare})H(X_{1}\oplus
N_{1}\oplus  (X_{2}\vee X_{3}) |\tilde{\TimeSharingRV}=\tilde{\timeshare}
)-\sum_{x_{ 1
},\tilde{\timeshare}}p_{X_{1}\tilde{\TimeSharingRV}}(x_{1,}\tilde{\timeshare})
H(x_{1}\oplus N_{1}\oplus  (X_{2}\vee X_{3})|X_{1}=x_{1},\tilde{\TimeSharingRV}=\tilde{\timeshare})\nonumber\\}\fi
\label{Eqn:3To1ORICX2X3N1IndependentOfX1GivenQ}
&=&\sum_{\tilde{\timeshare}}p_{\tilde{\TimeSharingRV}}(\tilde{\timeshare})H(X_{1}\oplus
N_{1}\oplus  (X_{2}\vee X_{3}) |\tilde{\TimeSharingRV}=\tilde{\timeshare})-\sum_{x_{ 1
},\tilde{\timeshare}}p_{X_{1}\tilde{\TimeSharingRV}}(x_{1,}\tilde{\timeshare})
H(N_{1}\oplus  (X_{2}\vee X_{3})
|\tilde{\TimeSharingRV}=\tilde{\timeshare})\\
&\leq&\sum_{\tilde{\timeshare}}p_{\tilde{\TimeSharingRV}}(\tilde{\timeshare})H(X_{1}\oplus
N_{1}|\tilde{\TimeSharingRV}=\tilde{\timeshare}
)\label{Eqn:3To1ORICInequalityDueToUncertainityInX2PlusX3GivenU2U3}-\sum_{\tilde{\timeshare}} p_ { \tilde{\TimeSharingRV} } (\tilde{\timeshare} )
H(N_{1}|\tilde{\TimeSharingRV}=\tilde{\timeshare})=\sum_{\timeshare}p_{\tilde{\TimeSharingRV}}(\tilde{\timeshare})H(X_{1}\oplus N_{1}|\tilde{\TimeSharingRV}=\tilde{\timeshare})-h_{b}(\delta_{1})\\
\label{Eqn:3To1ORICJensen'sInequalityAndCostConstraint}
&=&\sum_{\tilde{\timeshare}}p_{\tilde{\TimeSharingRV}}(\tilde{\timeshare})h_{b}(\tau_{1\tilde{q}}*\delta_{1})- h_{b}(\delta_{1})\leq h_{b}(\Expectation_{\tilde{\TimeSharingRV}}[ \tau_{1\tilde{q}}*\delta_{1} ])- h_{b}(\delta_{1})= h_{b}(p_{X_{1}}(1)*\delta_{1})-h_{b}(\delta_{1}),
\end{eqnarray}
where (i) (\ref{Eqn:3To1ORICInequalityDueToUncertainityInX2PlusX3GivenU2U3}) follows from substituting $p_{X_{1}\oplus N_{1}|\tilde{\TimeSharingRV}}(\cdot|\tilde{\timeshare})$ for $p_{Z_{1}}$,
$p_{N_{1}|\tilde{\TimeSharingRV}}(\cdot|\tilde{\timeshare})$ for $p_{Z_{2}}$ and $p_{X_{2}\vee
X_{3}|\tilde{\TimeSharingRV}}(\cdot|\tilde{\timeshare})$ for
$p_{Z_{3}}$ in lemma
\ref{Lem:AddingAnIndependentRandomVariableReducesDifferenceInEntropies},
(iii) the first inequality in
(\ref{Eqn:3To1ORICJensen'sInequalityAndCostConstraint}) follows from
Jensen's inequality. Since $p_{X_{1}}(1)\leq \tau_{1}< \frac{1}{2}$,
we have
$p_{X_{1}}(1)*\delta_{1}=p_{X_{1}}(1-\delta_{1})+(1-p_{X_{1}}(1))\delta_{1}
=p_{X_{1}}(1)(1-2\delta_{1})+\delta_{1} \leq
\tau_{1}(1-2\delta_{1})+\delta_{1} \leq
\frac{1}{2}(1-2\delta_{1})+\delta_{1}=\frac{1}{2}$. Therefore
$h_{b}(p_{X_{1}}(1)*\delta_{1})$ is increasing\footnote{This also
  employs the positivity of $1-2\delta_{1}$, or equivalently
  $\delta_{1}$ being in the range $(0,\frac{1}{2})$.} in
$p_{X_{1}}(1)$ and is bounded above by
$h_{b}(\tau_{1}*\delta_{1})$. Moreover, the condition for equality in
Jensen's inequality implies $R_{1} = h_{b}(\tau_{1} *
\delta_{1})-h_{b}(\delta_{1})$ if and only if
$p_{X_{1}|\tilde{\TimeSharingRV}}(1|\tilde{q})=\tau_{1}$ for all
$\tilde{q} \in \tilde{\TimeSharingRVSet}$. We have therefore proved
our second claim.\footnote{We have only proved
  $p_{X_{1}|QU_{2}U_{3}}(1|q,u_{2},u_{3}=\tau_{1})$ for all
  $(q,u_{2},u_{3}) \in \TimeSharingRVSet \times \SemiPrivateRVSet_{2}
  \times \SemiPrivateRVSet_{3}$ of positive probability. The claim now
  follows from conditional independence of $X_{1}$ and $U_{2},U_{3}$
  given $\TimeSharingRV$.} 

Our third claim is that either $H(X_{2}|\TimeSharingRV,\SemiPrivateRV_{2}) >0$ or $H(X_{3}|\TimeSharingRV,\SemiPrivateRV_{3}) >0$. Suppose not, i.e., $H(X_{2}|\TimeSharingRV,\SemiPrivateRV_{2})=H(X_{3}|\TimeSharingRV,\SemiPrivateRV_{3}) =0$. In this case, the upper bound on $R_{1}+R_{2}+R_{3}$ in (\ref{Eqn:SumRateBoundOn3To1IC}) is 
\begin{eqnarray}
\label{Eqn:3To1ORICUpperBoundOnSumOfTheUseRates}
R_{1}+R_{2}+R_{3} &\leq& I(X_{2},X_{3},X_{1};Y_{1}|Q)=H(Y_{1}|Q)-H(Y_{1}|Q,X_{1},X_{2},X_{3})\nonumber\\
&=&h_{b}(\tau_{1}(1-\beta)+(1-\tau_{1})\beta)-h_{b}(\delta_{1}),\nonumber
\end{eqnarray}
where the last equality follows from substituting $p_{X_{j}|Q}:j=1,2,3$ derived in the earlier two claims.\footnote{$\beta \define (1-\tau)^{2}\delta_{1}+(2\tau-\tau^{2})(1-\delta_{1})$ is as defined in the statement of the lemma.} The hypothesis (\ref{Eqn:3To1-OR-ICConditionForStrictSubOptimalityOfHanKobayashi}) therefore precludes $(h_{b}(\tau_{1} * \delta_{1})-h_{b}(\delta_{1}),h_{b}(\tau * \delta)-h_{b}(\delta),h_{b}(\tau * \delta)-h_{b}(\delta)) \in \alpha^{\three-1}_{f}(p_{QU_{2}U_{3}\ulineInputRV\ulineOutputRV})$ if $H(X_{2}|\TimeSharingRV,\SemiPrivateRV_{2})=H(X_{3}|\TimeSharingRV,\SemiPrivateRV_{3}) =0$. This proves our third claim.

Our fourth claim is $H(X_{2}\vee X_{3}|\TimeSharingRV,\SemiPrivateRV_{2},\SemiPrivateRV_{3})>0$. The proof of this claim rests on each of the earlier three claims. Note that we have either $H(X_{2}|\TimeSharingRV,\SemiPrivateRV_{2}) >0$ or $H(X_{3}|\TimeSharingRV,\SemiPrivateRV_{3}) >0$. Without loss of generality, we assume $H(X_{2}|\TimeSharingRV,\SemiPrivateRV_{2}) >0$. We therefore have a $u_{2}^{*} \in \SemiPrivateRVSet_{2}$ such that $p_{U_{2}|Q}(u_{2}^{*}|q^{*})>0$ and $H(X_{2}|U_{2}=u_{2}^{*},Q=q^{*})>0$. This implies $p_{X_{2}|U_{2}Q}(x_{2}|u_{2}^{*},q^{*}) \notin \{0,1\}$ for each $x_{2} \in \{0,1\}$.
Since $p_{Q}(q^{*})>0$, from the first claim we have $0<1-\tau=p_{X_{3}|Q}(0|q^{*})= \sum_{u_{3} \in \SemiPrivateRVSet_{3}}p_{X_{3}U_{3}|Q}(0,u_{3}|q^{*})$.
This guarantees existence of $u_{3}^{*} \in \SemiPrivateRVSet_{3}$ such that $p_{X_{3}U_{3}|Q}(0,u_{3}^{*}|q^{*})>0$. We therefore have $p_{U_{3}|Q}(u_{3}^{*}|q^{*})>0$ and $1\geq p_{X_{3}|U_{3}Q}(0|u_{3}^{*},q^{*})>0$.

We have therefore identified $(q^{*},u_{2}^{*},u_{3}^{*}) \in \TimeSharingRVSet \times \SemiPrivateRVSet_{2} \times \SemiPrivateRVSet_{3}$ such that $p_{Q}(q^{*})>0$, $p_{U_{2}|Q}(u_{2}^{*}|q^{*})>0$, $p_{U_{3}|Q}(u_{3}^{*}|q^{*})>0$, $p_{X_{2}|U_{2}Q}(x_{2}|u_{2}^{*},q^{*}) \notin \{0,1\}$ for each $x_{2} \in \{0,1\}$ and $1\geq p_{X_{3}|U_{3}Q}(0|u_{3}^{*},q^{*})>0$. By conditional independence of the pairs $(X_{2},U_{2})$ and $(X_{3},U_{3})$ given $Q$, we also have $p_{X_{2}|U_{2}U_{3}Q}(x_{2}|u_{2}^{*},u_{3}^{*},q^{*}) \notin \{0,1\}$ for each $x_{2} \in \{0,1\}$ and $1\geq p_{X_{3}|U_{2}U_{3}Q}(0|u_{2}^{*},u_{3}^{*},q^{*})>0$. The reader may now verify $p_{X_{2}\vee X_{3}|U_{2}U_{3}Q}(x|u_{2}^{*},u_{3}^{*},q^{*}) \notin \{0,1\}$ for each $x \in \{0,1\}$. Since $p_{QU_{2}U_{3}}(q^{*},u_{2}^{*},u_{3}^{*})=p_{Q}(q^{*})p_{U_{2}|Q}(u_{2}^{*}|q^{*})p_{U_{3}|Q}(u_{3}^{*}|q^{*})>0$, we have proved the fourth claim.

Our fifth and final claim is $R_{1}< h_{b}(\tau_{1}*\delta_{1})-h_{b}(\delta_{1})$. This follows from a sequence of steps employed in proof of the second claim herein, or in the proof of proposition \ref{Prop:StrictSub-OptimalityOfHanKobayashi}. Denoting $\tilde{\TimeSharingRV} \define
(\TimeSharingRV,\SemiPrivateRV_{2},\SemiPrivateRV_{3})$ and a generic element
$\tilde{\timeshare} \define (\timeshare,u_{2},u_{3}) \in \tilde{\TimeSharingRVSet}
\define \TimeSharingRVSet \times \SemiPrivateRVSet_{2} \times \SemiPrivateRVSet_{3}$, we
observe that
\begin{eqnarray}
 \label{Eqn:3To1ORICSecondProvingTheSecondClaimOnR1}
\lefteqn{R_{1} \leq I(X_{1};Y_{1}|\tilde{\TimeSharingRV}) =\sum_{\tilde{\timeshare}}p_{\tilde{\TimeSharingRV}}(\tilde{\timeshare})H(X_{1}\oplus
N_{1}\oplus  (X_{2}\vee X_{3}) |\tilde{\TimeSharingRV}=\tilde{\timeshare}
)-\sum_{\tilde{\timeshare}}p_{\tilde{\TimeSharingRV}}(\tilde{\timeshare})
H(N_{1}\oplus (X_{2}\vee X_{3})
|\tilde{\TimeSharingRV}=\tilde{\timeshare})} \nonumber\\
&<&\sum_{\tilde{\timeshare}}p_{\tilde{\TimeSharingRV}}(\tilde{\timeshare})H(X_{1}\oplus
N_{1}|\tilde{\TimeSharingRV}=\tilde{\timeshare}
)\label{Eqn:3To1ORICSecondInequalityDueToUncertainityInX2PlusX3GivenU2U3}-\sum_{\tilde{\timeshare}} p_ { \tilde{\TimeSharingRV} } (\tilde{\timeshare} )
H(N_{1}|\tilde{\TimeSharingRV}=\tilde{\timeshare})=\sum_{\timeshare}p_{\tilde{\TimeSharingRV}}(\tilde{\timeshare})H(X_{1}\oplus N_{1}|\tilde{\TimeSharingRV}=\tilde{\timeshare})-h_{b}(\delta_{1})\\
\label{Eqn:3To1ORICSecondJensen'sInequalityAndCostConstraint}
&=&\sum_{\tilde{\timeshare}}p_{\tilde{\TimeSharingRV}}(\tilde{\timeshare})h_{b}(\tau_{1\tilde{q}}*\delta_{1})- h_{b}(\delta_{1})\leq h_{b}(\Expectation_{\tilde{\TimeSharingRV}}\left\{ \tau_{1\tilde{q}}*\delta_{1} \right\})- h_{b}(\delta_{1})= h_{b}(p_{X_{1}}(1)*\delta_{1})-h_{b}(\delta_{1}),
\end{eqnarray}
where (i) (\ref{Eqn:3To1ORICSecondInequalityDueToUncertainityInX2PlusX3GivenU2U3}) follows from existence of a $\tilde{\timeshare}^{*} \in \tilde{\TimeSharingRVSet}$ for which $H(X_{2}\vee X_{3}|\tilde{\TimeSharingRV}=\tilde{\timeshare}^{*})>0$ and substituting $p_{X_{1}\oplus N_{1}|\tilde{\TimeSharingRV}}(\cdot|\tilde{\timeshare}^{*})$ for $p_{Z_{1}}$,
$p_{N_{1}|\tilde{\TimeSharingRV}}(\cdot|\tilde{\timeshare}^{*})$ for $p_{Z_{2}}$ and $p_{X_{2}\vee
X_{3}|\tilde{\TimeSharingRV}}(\cdot|\tilde{\timeshare}^{*})$ for $p_{Z_{3}}$ in lemma \ref{Lem:AddingAnIndependentRandomVariableReducesDifferenceInEntropies}, (iii) the first inequality in (\ref{Eqn:3To1ORICSecondJensen'sInequalityAndCostConstraint}) follows from Jensen's inequality. Since $p_{X_{1}}(1)*\delta_{1}=p_{X_{1}}(1-\delta_{1})+(1-p_{X_{1}}(1))\delta_{1} =p_{X_{1}}(1)(1-2\delta_{1})+\delta_{1} \leq \tau_{1}(1-2\delta_{1})+\delta_{1} \leq \frac{1}{2}(1-2\delta_{1})+\delta_{1}=\frac{1}{2}$. Therefore $h_{b}(p_{X_{1}}(1)*\delta_{1})$ is increasing\footnote{This also employs the positivity of $1-2\delta_{1}$, or equivalently $\delta_{1}$ being in the range $(0,\frac{1}{2})$.} in $p_{X_{1}}(1)$ and is bounded above by $h_{b}(\tau_{1}*\delta_{1})$. We therefore have $R_{1}< h_{b}(\tau_{1}*\delta_{1})-h_{b}(\delta_{1})$.
\end{proof}
We now derive conditions under which $\alpha^{3-1}_{f}(\tau_{1},\tau,\tau)=\mathbb{C}(\tau_{1},\tau,\tau)$. Clearly, $\mathbb{C}(\tau_{1},\tau,\tau) \subseteq \beta(\underline{\tau},\underline{\delta})$ where $\underline{\tau}=(\tau_{1},\tau,\tau)$ and $\underline{\delta}=(\delta_{1},\delta,\delta)$. It therefore suffices to derive conditions under which $(h_{b}(\tau_{1} * \delta_{1})-h_{b}(\delta_{1}),h_{b}(\tau * \delta)-h_{b}(\delta),h_{b}(\tau * \delta)-h_{b}(\delta)) \in \alpha^{\three-1}_{f}(\tau_{1},\tau,\tau)$.

\begin{lemma}
 \label{Lem:3To1ORICConditionsUnderWhichRateTripleIsAchievableUsingCosetCodes}
Consider example \ref{Ex:A3To1-OR-IC} with $\delta \define \delta_{2}=\delta_{3} \in (0,\frac{1}{2})$ and $\tau \define \tau_{2}=\tau_{3} \in (0,\frac{1}{2})$. Let $\beta \define \delta_{1}*(2\tau -\tau^{2})$. The rate triple $(h_{b}(\tau_{1} * \delta_{1})-h_{b}(\delta_{1}),h_{b}(\tau * \delta)-h_{b}(\delta),h_{b}(\tau * \delta)-h_{b}(\delta)) \in \alpha^{\three-1}_{f}(\tau_{1},\tau,\tau)$ i.e., achievable using coset codes, if,
\begin{eqnarray}
 \label{Eqn:3To1-OR-ICConditionForAchievabilityUsingCosetCodes}
h_{b}(\tau * \delta)-h_{b}(\delta)\leq  \theta ,
\end{eqnarray}
where $\theta =h_{b}(\tau)-h_{b}((1-\tau)^{2})-(2\tau-\tau^{2})h_{b}(\frac{\tau^{2}}{2\tau-\tau^{2}})-h_{b}(\tau_{1} * \delta_{1})+h_{b}(\tau_{1}*\beta)$. We therefore have $\alpha^{3-1}_{f}(\tau_{1},\tau,\tau)=\mathbb{C}(\tau_{1},\tau,\tau)$ if (\ref{Eqn:3To1-OR-ICConditionForAchievabilityUsingCosetCodes}) holds.
\end{lemma}
\begin{proof}
The proof only involves identifying the appropriate test channel $p_{QU_{2}U_{3}\ulineInputRV\ulineOutputRV} \in \mathbb{D}_{f}^{3-1}(\tau_{1},\tau,\tau)$. Let $\TimeSharingRVSet = \phi$ be empty, $\SemiPrivateRVSet_{2}=\SemiPrivateRVSet_{3}=\{ 0,1,2 \}$. Let $p_{X_{1}}(1)=1-p_{X_{1}}(0)=\tau_{1}$. Let $p_{U_{j}X_{j}}(0,0)=1-p_{U_{j}X_{j}}(1,1)=1-\tau$ and therefore $P(U_{j}=2)=P(X_{j} \neq U_{j})=0$ for $j=2,3$. It is easily verified that $p_{QU_{2}U_{3}\ulineInputRV\ulineOutputRV} \in \mathbb{D}_{f}^{3-1}(\tau_{1},\tau,\tau)$, i.e, in particular respects the cost constraints. The choice of this test channel, particularly the ternary field, is motivated by $H(X_{2}\vee X_{3}|U_{2}\oplus_{3} U_{3})=0$. The decoder $1$ can reconstruct the interfering pattern after having decoded the ternary sum of the codewords.
\end{proof}

\section{Proof of proposition \ref{Prop:ResultForOrExampleWithMACAsAProp}}
\label{AppSec:3To1ORNonAddICStrictSubOptimalityOfUnstructuredCodes}
We prove proposition \ref{Prop:ResultForOrExampleWithMACAsAProp} by splitting the same into the two following lemmas.
\begin{lemma}
 \label{Lem:StrictSubOptimalityOfUnstructuredCodesForNonAdditice3To1IC}
Consider example \ref{Ex:3To1ORICCoupledThroughNonAdditiveMAC} and let $\underline{C}^{*}, C_{1}, \mathcal{D}(\ulinecost),p^{*}_{\ulineInputRV\ulineOutputRV}$ be defined as above. If
\begin{eqnarray}
 \label{Eqn:3To1OrNonAddICUnstructuredCodesSubOptimal}
C_{1}+2(h_{b}(\tau*\delta)-h_{b}(\delta))=I(X_{1};Y_{1}|X_{2}\vee X_{3})+2(h_{b}(\tau*\delta)-h_{b}(\delta)) > I(\ulineInputRV ;Y_{1}),
\end{eqnarray}
where the mutual information terms $I(X_{1};Y_{1}|X_{2}\vee X_{3}), I(\ulineInputRV ;Y_{1})$ are evaluated with respect to $p^{*}_{\ulineInputRV\ulineOutputRV}$, then $\underline{C}^{*} \notin \alpha_{u}(\ulinecost)$.
\end{lemma}
The reader will recognize that above lemma is the counterpart of lemma \ref{Lem:ConditionsForSub-OptimalityOfHanKobayashiTechnique} for example \ref{Ex:3To1ORICCoupledThroughNonAdditiveMAC}.
\begin{proof}
The proof here closely mimics proof of lemma \ref{Lem:ConditionsForSub-OptimalityOfHanKobayashiTechnique}. In fact, we allude to appendix \ref{AppSec:UpperBoundOnUSBRateRegionFor3To1ORIC} to avoid restating certain elements.

We assume $\underline{C}^{*} \in \alpha_{u}(\ulinecost)$, and derive a contradiction. Suppose $\underline{C}^{*} \in \cocl(\alpha_{u}(p_{QU_{2}U_{3}\ulineInputRV\ulineOutputRV}))$ for some $p_{QU_{2}U_{3}\ulineInputRV\ulineOutputRV} \in \mathbb{D}_{u}(\ulinecost)$\footnote{Recall $\ulinecost \define (\tau_{1},\tau,\tau)$.}. In the sequel, we characterize such a $p_{QU_{2}U_{3}\ulineInputRV\ulineOutputRV}$ and employ the same to derive a contradiction. Our first claim, as in appendix \ref{AppSec:UpperBoundOnUSBRateRegionFor3To1ORIC}, is $p_{X_{j}|Q}(1|q)=\tau$ for $j=2,3$ and every $q \in \mathcal{Q}$. Since the corresponding arguments in appendix \ref{AppSec:UpperBoundOnUSBRateRegionFor3To1ORIC} hold verbatim, we allude to the same for a proof of this claim. We conclude the triplet $(Q,X_{1}),X_{2},X_{3}$ to be mutually independent, and in particular $X_{1},X_{2},X_{3}$ to be mutually independent. We conclude that for any $p_{QU_{2}U_{3}\ulineInputRV\ulineOutputRV} \in \mathbb{D}_{u}(\ulinecost)$ for which $\underline{C}^{*} \in \cocl(\alpha_{u}(p_{QU_{2}U_{3}\ulineInputRV\ulineOutputRV}))$, we have its corresponding marginal $p_{\ulineInputRV\ulineOutputRV} \in \mathcal{D}(\ulinecost)$.

Our second claim is $p_{X_{1}|Q}(1|q)=p^{*}_{X_{1}}(1)$ for every $q \in \mathcal{Q}$ for which $p_{Q}(q)>0$. We begin with the upper bound on $R_{1}$ in (\ref{Eqn:3To1ICUpperBoundOnR1}). Denoting
\begin{eqnarray}
 \label{Eqn:3To1ORAndNonAdditiveICAlternateNotationForMutInf}
 I(p_{A|C}(\cdot|c);p_{B|A,C}(\cdot|\cdot,c)) \define I(A;B|C=c)\mbox{ for any random variables $A,B,C$, we have,}\nonumber
\end{eqnarray}
\begin{eqnarray}
\label{Eqn:3To1ORNonAddICConsOfX2OrX3EnteringMAC}
\lefteqn{I(X_{1};Y_{1}|Q,U_{2},U_{3}) \leq I(X_{1};Y_{1}|Q,X_{2}\vee X_{3})}\\
&=&\sum_{s}p_{X_{2}\vee X_{3}}(s) \sum_{q}p_{Q|X_{2}\vee X_{3}}(q|s)I\left( p_{X_{1}|Q,X_{2}\vee X_{3}}(\cdot|q,s);p_{Y_{1}|X_{1}Q,X_{2}\vee X_{3}}(\cdot|\cdot,q,s) \right) \nonumber\\
\label{Eqn:3To1ICOrNonAddPX1NatureOfChannel}
&=&\sum_{s}p_{X_{2}\vee X_{3}}(s) \sum_{q}p_{Q|X_{2}\vee X_{3}}(q|s)I\left( p_{X_{1}|Q,X_{2}\vee X_{3}}(\cdot|q,s);p_{Y_{1}|X_{1},X_{2}\vee X_{3}}(\cdot|\cdot,s) \right) \\
\label{Eqn:3To1OrNonAddICJensenIneqInDerivingPX1}
&\leq&\sum_{s}p_{X_{2}\vee X_{3}}(s) I\left( \sum_{q}p_{Q|X_{2}\vee X_{3}}(q|s)p_{X_{1}|Q,X_{2}\vee X_{3}}(\cdot|q,s);p_{Y_{1}|X_{1},X_{2}\vee X_{3}}(\cdot|\cdot,s) \right) \\
\label{Eqn:3To1OrNonAddICPx1DefnOfC1}
&=&\sum_{s}p_{X_{2}\vee X_{3}}(s)I\left( p_{X_{1}|X_{2}\vee X_{3}}(\cdot|s);p_{Y_{1}|X_{1},X_{2}\vee X_{3}}(\cdot|\cdot,s) \right)=I(X_{1};Y_{1}|X_{2}\vee X_{3}) \leq C_{1}
\end{eqnarray}
where (i) (\ref{Eqn:3To1ORNonAddICConsOfX2OrX3EnteringMAC}) follows from the Markov chains $(U_{2},U_{3})-(X_{2}\vee X_{3})-Y_{1}$ and $(U_{2},U_{3})-(X_{1},X_{2}\vee X_{3})-Y_{1}$, (ii) (\ref{Eqn:3To1ICOrNonAddPX1NatureOfChannel}) follows from the Markov chain $Q-X_{1},X_{2}\vee X_{3} - Y_{1}$ resulting from the nature of the channel from the inputs to $Y_{1}$, (iii) (\ref{Eqn:3To1OrNonAddICJensenIneqInDerivingPX1}) follows from Jensen's inequality, and (iv) (\ref{Eqn:3To1OrNonAddICPx1DefnOfC1}) follows from $p_{\ulineInputRV\ulineOutputRV} \in \mathcal{D}(\ulinecost)$ and definition of $C_{1}$. The strict concavity of $I(p_{A}(\cdot);p_{B|A}(\cdot|\cdot))$ in $p_{A}(\cdot)$ implies equality holds in (\ref{Eqn:3To1OrNonAddICJensenIneqInDerivingPX1}) if and only if $p_{X_{1}|Q,X_{2}\vee X_{3}}(1|q,s)=p_{X_{1}|Q}(1|q)$ is invariant with $q$ for every $q \in \mathcal{Q}$ for which $p_{Q|X_{2}\vee X_{3}}(q|s) = p_{Q}(q) > 0$.\footnote{We have proved in our first claim $Q$ and $(X_{2},X_{3})$ are independent.} By the uniqueness of $p^{*}_{\ulineInputRV\ulineOutputRV}$, and in particular $p_{X_{1}}^{*}$, we conclude $p_{X_{1}|Q}(1|q)=p^{*}_{X_{1}}(1)$ for every $q \in \mathcal{Q}$ for which $p_{Q}(q)>0$.

Our first and second claims imply that if $\underline{C}^{*} \in \cocl(\alpha_{u}(p_{QU_{2}U_{3}\ulineInputRV\ulineOutputRV}))$ for some $p_{QU_{2}U_{3}\ulineInputRV\ulineOutputRV} \in \mathbb{D}_{u}(\ulinecost)$, then $\underset{q,u_{2},u_{3}}{\sum}p_{QU_{2}U_{3}\ulineInputRV,\ulineOutputRV}(q,u_{2},u_{3},\ulineinput,\ulineoutput)=p^{*}_{\ulineInputRV\ulineOutputRV}(\ulineinput,\ulineoutput) \in \mathcal{D}(\ulinecost)$, and furthermore, $Q$ is independent of $\ulineInputRV$. We therefore reiterate that any entropy or mutual information terms involving random variables in $\ulineInputRV,\ulineOutputRV$, stated in the sequel, is evaluated with respect to $p^{*}_{\ulineInputRV\ulineOutputRV}$. 

Our third claim is that either $H(X_{2}|Q,U_{2})>0$ or $H(X_{3}|Q,U_{3})>0$. Suppose not, i.e., $H(X_{2}|\TimeSharingRV,\SemiPrivateRV_{2})=H(X_{3}|\TimeSharingRV,\SemiPrivateRV_{3}) =0$. In this case, the upper bound on $R_{1}+R_{2}+R_{3} = C_{1}+2(h_{b}(\tau * \delta) - h_{b}(\delta))$ in (\ref{Eqn:SumRateBoundOn3To1IC}) is 
\begin{eqnarray}
\label{Eqn:3To1ORNonAddICUpperBoundOnSumOfTheUseRates}
R_{1}+R_{2}+R_{3} = C_{1}+2(h_{b}(\tau * \delta) - h_{b}(\delta)) \leq I(X_{2},X_{3},X_{1};Y_{1}|Q)=I(\ulineInputRV;Y_{1})
\end{eqnarray}
where the last equality follows from independence of $Q$ and $\ulineInputRV$ and thereby implying independence of $Q$ and $(\ulineInputRV,\ulineOutputRV)$. (\ref{Eqn:3To1ORNonAddICUpperBoundOnSumOfTheUseRates}) contradicts the hypothesis (\ref{Eqn:3To1OrNonAddICUnstructuredCodesSubOptimal}) of the lemma.

Our fourth claim is $H(X_{2}\vee X_{3}|Q,U_{2},U_{3})>0$. The proof of this claim is identical to the proof of the corresponding claim in appendix \ref{AppSec:UpperBoundOnUSBRateRegionFor3To1ORIC} and the reader is alluded to the same. As a consequence of $H(X_{2}\vee X_{3}|\tilde{Q})>0$, where $\tilde{Q} \define (Q,U_{2},U_{3})$, there exists $\tilde{q}^{*} \define (q^{*},u_{2}^{*},u_{3}^{*}) \in \tilde{\mathcal{Q}} \define \mathcal{Q} \times \mathcal{U}_{2}\times \mathcal{U}_{3}$ for which $p_{\tilde{Q}}(\tilde{q}^{*})>0$ and $H(X_{2}\vee X_{3}|\tilde{Q}=\tilde{q}^{*})>0$. 

Our fifth claim and final claim is that $H(X_{2}\vee X_{3}|Q,U_{2},U_{3})>0$ implies $C_{1} < I(X_{1};Y_{1}|X_{2}\vee X_{3})$ thereby contradicting the definition of $C_{1}$ (\ref{Eqn:TestChannelsThatEnsureUsers2And3AchieveCapacity}). The reader will recognize that our proof for the fifth claim in appendix \ref{AppSec:UpperBoundOnUSBRateRegionFor3To1ORIC} \textit{cannot} be employed here. We employ a more powerful technique that we will have opportunity to use in our study of example \ref{Ex:3To1ICGroupAdditiveExample}. The upper bound (\ref{Eqn:3To1ICUpperBoundOnR1}) on $R_{1}$ implies
\begin{eqnarray}
\label{Eqn:3To1ORNonAddICUppeBoundOnUnstructuredCodesFifthClaimUsingBound}
\lefteqn{C_{1}=R_{1} \leq I(X_{1};Y_{1}|\tilde{Q})= \sum_{\tilde{q}}p_{\tilde{Q}}(\tilde{q})I(p_{X_{1}|\tilde{Q}}(\cdot|\tilde{q});p_{Y_{1}|X_{1}\tilde{Q}}(\cdot|\cdot,\tilde{q}))} \nonumber\\
\label{Eqn:3To1ORNonAddICUppeBoundOnUnstructuredCodesFifthClaimUsingBound1}
&=& \sum_{\tilde{q}}p_{\tilde{Q}}(\tilde{q})I\left(p_{X_{1}|\tilde{Q}}(\cdot|\tilde{q});\sum_{s}p_{Y_{1}|X_{1},X_{2}\vee X_{3}\tilde{Q}}(\cdot|\cdot,s,\tilde{q})p_{X_{2}\vee X_{3}|\tilde{Q}}(s|\tilde{q}) \right) \nonumber\\
\label{Eqn:3To1ORNonAddICUppeBoundOnUnstructuredCodesFifthClaimUsingBound2}
&<&\sum_{\tilde{q}}p_{\tilde{Q}}(\tilde{q})\sum_{s}p_{X_{2}\vee X_{3}|\tilde{Q}}(s|\tilde{q})I\left(p_{X_{1}|\tilde{Q}}(\cdot|\tilde{q});p_{Y_{1}|X_{1},X_{2}\vee X_{3}\tilde{Q}}(\cdot|\cdot,s,\tilde{q}) \right)\\
\label{Eqn:3To1ORNonAddICUppeBoundOnUnstructuredCodesFifthClaimUsingBound3}
&=&\sum_{s,\tilde{q}}p_{\tilde{Q},X_{2}\vee X_{3}}(\tilde{q},s)I(p_{X_{1}}(\cdot);p_{Y_{1}|X_{1},X_{2}\vee X_{3}}(\cdot|\cdot,s))\\
\label{Eqn:3To1ORNonAddICUppeBoundOnUnstructuredCodesFifthClaimUsingBound4}
&=&\sum_{s,\tilde{q}}p_{\tilde{Q},X_{2}\vee X_{3}}(\tilde{q},s)I(p_{X_{1}|X_{2}\vee X_{3}}(\cdot|s);p_{Y_{1}|X_{1},X_{2}\vee X_{3}}(\cdot|\cdot,s))=I(X_{1};Y_{1}|X_{2}\vee X_{3}) \leq C_{1},
\end{eqnarray}
where (i) (\ref{Eqn:3To1ORNonAddICUppeBoundOnUnstructuredCodesFifthClaimUsingBound2}) follows from strict convexity of the mutual information in the conditional distribution (channel transition probabilities), the presence of $\tilde{q}^{*} \in \tilde{\mathcal{Q}}$ for which $p_{X_{2}\vee X_{3}|\tilde{Q}}(\cdot|\tilde{q}^{*})$ is non-degenerate and $p_{Y_{1}|X_{1},X_{2}\vee X_{3},\tilde{Q}}(\cdot|\cdot,s,\tilde{q}^{*})$ distinct, (ii) (\ref{Eqn:3To1ORNonAddICUppeBoundOnUnstructuredCodesFifthClaimUsingBound3}) follows from conditional independence of $X_{1}$ and $(U_{2},U_{3})$ given $Q$, the second claim above, and the Markov chain $\tilde{Q}-X_{1},X_{2}\vee X_{3} - Y_{1}$ induced by the nature of the channel, and (iii) (\ref{Eqn:3To1ORNonAddICUppeBoundOnUnstructuredCodesFifthClaimUsingBound4}) follows from $X_{1},X_{2},X_{3}$ being mutually independent, $p_{\ulineInputRV\ulineOutputRV} \in \mathcal{D}(\ulinecost)$ and the definition of $C_{1}$. We have thus derived a contradiction $C_{1}<C_{1}$.
\end{proof}

\begin{lemma}
 \label{Lem:OptimalityOfPCCForNonAdditive3To1OrNonAddIC}
Consider example \ref{Ex:3To1ORICCoupledThroughNonAdditiveMAC}. Let $C_{1}, p^{*}_{\ulineInputRV\ulineOutputRV}$ be as defined above. If $h_{b}(\tau^{2})+(1-\tau^{2})h_{b}(\frac{(1-\tau)^{2}}{1-\tau^{2}})+H(Y_{1}|X_{2}\vee X_{3})-H(Y_{1}) \leq \min\{ H(X_{2}|Y_{2})H(X_{3}|Y_{3})\}$, where the entropies are evaluated with respect to $p^{*}_{\ulineInputRV\ulineOutputRV}$, then $(C_{1},h_{b}(\delta * \tau)-h_{b}(\delta),h_{b}(\delta * \tau)-h_{b}(\delta)) \in \alpha_{f}^{3-1}(\ulinecost)$.
\end{lemma}
\begin{proof}
As in proof of lemma \ref{Lem:3To1ORICConditionsUnderWhichRateTripleIsAchievableUsingCosetCodes}, we identify an appropriate test channel $p_{QU_{2}U_{3}\ulineInputRV\ulineOutputRV} \in \mathbb{D}_{f}(\ulinecost)$ for which $(C_{1},h_{b}(\delta * \tau)-h_{b}(\delta),h_{b}(\delta * \tau)-h_{b}(\delta)) \in \alpha_{f}^{3-1}(p_{QU_{2}U_{3}\ulineInputRV\ulineOutputRV})$. Let $\TimeSharingRVSet = \phi$ be empty, $\SemiPrivateRVSet_{2}=\SemiPrivateRVSet_{3}=\{ 0,1,2 \}$. Let $p_{\ulineInputRV}=p^{*}_{\ulineInputRV}$. Let $p_{U_{j}X_{j}}(0,0)=1-p_{U_{j}X_{j}}(1,1)=1-\tau$ and therefore $P(U_{j}=2)=P(X_{j} \neq U_{j})=0$ for $j=2,3$. It is easily verified that $p_{QU_{2}U_{3}\ulineInputRV\ulineOutputRV} \in \mathbb{D}_{f}^{3-1}(\ulinecost)$, i.e, in particular respects the cost constraints.

It maybe verified that the hypothesis $h_{b}(\tau^{2})+(1-\tau^{2})h_{b}(\frac{(1-\tau)^{2}}{1-\tau^{2}})+H(Y_{1}|X_{2}\vee X_{3})-H(Y_{1}) = H(U_{2}\oplus_{3}U_{3}) + H(Y_{1}|X_{2}\vee X_{3})-H(Y_{1}) = H(U_{2}\oplus_{3}U_{3}) +H(Y_{1}|U_{2}\oplus_{3}U_{3})-H(Y_{1})=H(U_{2}\oplus_{3}U_{3}|Y_{1})$. we therefore have $H(U_{2}\oplus_{3}U_{3}|Y_{1}) \leq \min\{ H(X_{2}|Y_{2})H(X_{3}|Y_{3})\}$. This implies (i) $H(U_{j}) \geq H(U_{2}\oplus U_{3}|Y_{1})$ and (ii) $H(U_{j})-H(U_{2}\oplus U_{3}|Y_{1}) \geq H(U_{j})-H(U_{j}|Y_{j})=I(U_{j};Y_{j})=I(X_{j};Y_{j})=h_{b}$. Employing these in bounds characterizing $\alpha_{f}^{3-1}(p_{QU_{2}U_{3}\ulineInputRV\ulineOutputRV})$ and the marginal $p_{\ulineInputRV\ulineOutputRV}=p^{*}_{\ulineInputRV\ulineOutputRV}$, it can be verified that $(C_{1},h_{b}(\delta * \tau)-h_{b}(\delta),h_{b}(\delta * \tau)-h_{b}(\delta)) \in \alpha_{f}^{3-1}(p_{QU_{2}U_{3}\ulineInputRV\ulineOutputRV})$.

\begin{comment}{
 The proof only involves identifying the appropriate test channel $p_{QU_{2}U_{3}\ulineInputRV\ulineOutputRV} \in \mathbb{D}_{f}^{3-1}(\tau_{1},\tau,\tau)$. 

The choice of this test channel, particularly the ternary field, is motivated by $H(X_{2}\vee X_{3}|U_{2}\oplus_{3} U_{3})=0$. The decoder $1$ can reconstruct the interfering pattern after having decoded the ternary sum of the codewords. It maybe verified that for this test channel $p_{QU_{2}U_{3}\ulineInputRV\ulineOutputRV}$, $\alpha_{f}^{\three-1}(p_{QU_{2}U_{3}\ulineInputRV\ulineOutputRV})$ is defined as the set of rate triples $(R_{1},R_{2},R_{3}) \in [0,\infty)^{3}$ that satisfy
\begin{eqnarray}
 \label{Eqn:3To1ORICDescriptionOfTheRateRegionForTernaryTestChannelChosenForCosetCodes}
&R_{1} < \min\left\{0,\theta \right\}+h_{b}(\tau_{1}*\delta_{1})-h_{b}(\delta_{1}), ~~~
R_{j} < h_{b}(\tau * \delta)-h_{b}(\delta):j=2,3\nonumber\\
&R_{1}+R_{j} < h_{b}(\tau_{1}*\delta_{1})-h_{b}(\delta_{1})+\theta,
\end{eqnarray}
where $\theta=h_{b}(\tau)-h_{b}((1-\tau)^{2})-(2\tau-\tau^{2})h_{b}(\frac{\tau^{2}}{2\tau-\tau^{2}})-h_{b}(\tau_{1} * \delta_{1})+h_{b}(\tau_{1}(1-\beta)+(1-\tau_{1})\beta)$ is as defined in the statement of the lemma. Clearly, $(h_{b}(\tau_{1} * \delta_{1})-h_{b}(\delta_{1}),h_{b}(\tau * \delta)-h_{b}(\delta),h_{b}(\tau * \delta)-h_{b}(\delta)) \in \cocl(\alpha^{\three-1}_{f}(p_{U_{2}U_{3}\ulineInputRV\ulineOutputRV}))$ if (\ref{Eqn:3To1-OR-ICConditionForAchievabilityUsingCosetCodes}) is satisfied.
}\end{comment}

\end{proof}

\section{Proof of Theorem \ref{Thm:AchievableRateRegionFor3To1ICUsingGroupCosetCodes}}
\label{AppSec:abeliangroupsachievability}

We provide an illustration of the main arguments of the proof without
giving complete details. In view of our detailed proof of
theorem \ref{Thm:AchievableRateRegionFor3To1ICUsingCosetCodes}, the
interested reader can fill in the details. We begin with an alternate
characterization of
$\alpha_{g}^{\three-1}(p_{\ulineInputRV\ulineOutputRV})$ in terms of
the parameters of the code. 

\begin{definition}
 \label{Defn:3To1ICAchRegionSimplifiedPreFourierMotzkinBounds}
 Consider $(p_{\TimeSharingRV\SemiPrivateRV_{2}\SemiPrivateRV_{3}\ulineInputRV\ulineOutputRV},w) \in
\SetOfDistributions_{g}(\ulinecost)$ and let
 $G\define \SemiPrivateRVSet_{2}=\SemiPrivateRVSet_{3}$. Let 
$\tilde{\alpha}^{\three-1}_{g}(p_{\TimeSharingRV\SemiPrivateRV_{2}\SemiPrivateRV_{3}
\ulineInputRV\ulineOutputRV},w)$ be defined as the set of rate triples $(R_{1},R_{2},R_{3})
\in [0,\infty)^{3}$ for which $\underset{\delta > 0}{\cup}\tilde{\mathcal{S}}(\underlineR,p_{\TimeSharingRV\SemiPrivateRV_{2}\SemiPrivateRV_{3}
\ulineInputRV\ulineOutputRV},w,\delta)$ is non-empty, where $\tilde{\mathcal{S}}(\underlineR,p_{\TimeSharingRV\SemiPrivateRV_{2}\SemiPrivateRV_{3}
\ulineInputRV\ulineOutputRV},w,\delta)$ is defined as the collection
of  vectors $(S_{2},T_{2},L_{2},S_{3},T_{3},L_{3},R_g) \in [0,\infty )^{9}$ that satisfy
for $j=2,3$, with $Z=U_2 \oplus U_3$.
\begin{eqnarray}
 \label{Eqn:3-to-1ICBoundsInTermsOfCodeParameters}
&R_{j}=T_{j} +L_{j},~~(S_{j}-T_{j}) > \log |G|-
H(U_{j}|Q)+\delta,
~R_g > S_j +\delta \nonumber\\ 
&S_j > S_w^G(U_j;0|Q) + \log |G| -H(U_j|Q)+\delta, 
~~T_{j}>\delta,~L_{j}>\delta, ~L_{j} <
I(X_{j};Y_{j}|U_{j}\TimeSharingRV)-\delta, \nonumber \\
&S_{j} +L_{j} < \log |G|
+I(X_j;Y_j|U_j Q)+C_w^G(U_j;Y_j|Q) -H(U_j|Q) 
-\delta, ~~
R_{1}  < I(X_{1};Y_{1}|Z \TimeSharingRV)-\delta\nonumber\\
&R_{1}+R_{g}< \log |G| + I(X_1;Y_1|ZQ)+C_w^G(Z;Y_1|Q)-H(Z|Q) 
-\delta\nonumber
 \end{eqnarray}
\end{definition}
\begin{lemma}
 \label{Lem:3To1ICSimplifiedPreFourierMotzkinBoundEquality}
$\tilde{{\alpha}}^{\three-1}_{g}(p_{\TimeSharingRV\SemiPrivateRV_{2}\SemiPrivateRV_{3}
\ulineInputRV\ulineOutputRV},w)=\alpha^{\three-1}_{g}(p_{\TimeSharingRV\SemiPrivateRV_{2}\SemiPrivateRV_{3}
\ulineInputRV\ulineOutputRV},w)$.
\end{lemma}
\begin{proof}
 The proof follows from Fourier-Motzkin elimination.
\end{proof}
Choose the parameters  $(R_1,S_{2},T_{2},L_{2},S_{3},T_{3},L_{3},R_g) \in [0,\infty
)^{10}$. The coding technique is exactly the same as that considered in the case of  finite fields
and is given in the proof of Theorem \ref{Thm:AchievableRateRegionFor3To1ICUsingCosetCodes}, The main exception is that the PCCs are built on the
abelian group $G$. Instead of constructing vector spaces of
$\mathcal{F}^n$, we construct subgroups of $G^n$. The cloud center
codebook $\lambda_j$ of user $j$ is characterized as follows. Let
\begin{align*}
J_j=\bigoplus_{(p,r)\in\mathcal{Q}(G)} \mathbb{Z}_{p^r}^{s_jw_{p,r}}, \ \ \ 
J=\bigoplus_{(p,r)\in\mathcal{Q}(G)} \mathbb{Z}_{p^r}^{sw_{p,r}}, \ \ \ 
\end{align*}
for $j=2,3$ with $s=\max\{s_2,s_3\}$, where $s_j$ will be specified shortly.
Note that $J_j \leq J$ for $j=2,3$. Let $\phi$ be a homomorphism from $J$ into $G^n$. 
Let $\phi_j$ be the restriction of $\phi$ to $J_j$ for $j=2,3$. 
It is shown in \cite[Equation 11]{201305TITarXiv_SahPra} that $\phi$ has the following representation
\begin{align*}
\phi(a)=\bigoplus_{(p,r,m)\in\mathcal{G}(G^n)} \overbrace{\sum}^{(\mathbb{Z}_{p^r})}_{(q,t,l)\in\mathcal{G}(J)}
a_{(q,t,l)} g_{(q,t,l)\rightarrow (p,r,m)} 
\end{align*}
where $g_{(q,t,l)\rightarrow (p,r,m)}=0$ for $p\ne q$ and
$g_{(q,t,l)\rightarrow (p,r,m)}$ is uniformly distributed over
$p^{|r-t|^+}\mathbb{Z}_{p^r}$ for $p=q$. The code $\lambda_j$ is given
by $\phi_j(J_j) \oplus b_j^n$, where $b_j^n$ is a bias vector in $G^n$. 
Choose $s_2$, $s_3$ and $s$ such that 
\[
s_2= \frac{nS_2}{\sum_{(p,r) \in \mathcal{Q}(G)} r
w_{p,r} \log p }   \ \ \ 
s_3= \frac{nS_3 }{\sum_{(p,r) \in \mathcal{Q}(G)} r
w_{p,r} \log p }, \ \ \ 
s= \frac{nR_g}{\sum_{(p,r) \in \mathcal{Q}(G)} r
w_{p,r} \log p }  
\]
Note that 
\[
\frac{1}{n} \log |J|= \frac{s}{n} \sum_{(p,r) \in \mathcal{Q}(G)} r
w_{p,r} \log p = R_g, \ \ \ 
\frac{1}{n} \log |J_j| =  S_j: j=2,3.
\]
The binning
functions $i_j$ are defined analogously: $i_j: J_j \rightarrow
|G|^{t_j}$, where $t_j \log |G|=nT_j$, for $j=2,3$.
The encoding and decoding operations are defined analogously. 
This implies that $|\mathcal{M}_1|=2^{nR_1}$,
$|\mathcal{M}_{j1}| = |G|^{t_j}$ for $j=2,3$. 
The homomorphism and the bias vectors are chosen independently and
with uniform probability over their ranges.

For any  $a,\tilde{a} \in J$, and
$(q,s,l) \in \mathcal{G}(J)$ , let
$\hat{\theta}_{q,s,l} \in \{1,2,\ldots,s\}$ be such that 
$\tilde{a}_{q,s,l}-a_{q,s,l} \in
q^{\hat{\theta}_{q,s,l}} \mathbb{Z}_{q^s} \backslash
q^{\hat{\theta}_{q,s,l}+1} \mathbb{Z}_{q^s}$.
and any $(p,r) \in \mathcal{Q}(G)$, define
\[
\boldsymbol{\theta}_{p,r}(a,\tilde{a})= \min_{(p,s,l) \in \mathcal{G}(J)} |r-s|^+
+\hat{\theta}_{q,s,l}. 
\]
Define for any $a \in J$, and any 
$\theta=(\theta_{p,r})_{(p,r) \in \mathcal{Q}(G)}$, the set
$T_{J,\theta}(a)=\{\tilde{a} \in J:  \forall
(p,r) \in \mathcal{Q}(G), \boldsymbol{\theta}_{p,r}(a,\tilde{a})=\theta_{p,r} \}.$

It can be shown  that the expected value of the probability
of all the error events over the ensemble approach zero as the block
length increases if the parameters of the code  belong to 
$\tilde{{\alpha}}^{\three-1}_{g}(p_{\TimeSharingRV\SemiPrivateRV_{2}\SemiPrivateRV_{3}
\ulineInputRV\ulineOutputRV},w)$.
For conciseness, we give proofs of the elements in
this argument that
are new as compared to the analysis done in the case of fields. 

\textit{Upper bound on $P(\epsilon_{l_2})$}:- Given a message $m_2$
that indexes the bin in the cloud center codebook, 
define
\begin{align*}
\psi_2(m_{21})=\sum_{a\in J_2} \sum_{u_2\in T_{2\eta}(U_2)} \mathds{1}_{\{\phi(a)+b_2=u_2,i_2(a)=m_{21}\}}
\end{align*}
We have
\begin{align*}
\mathbb{E}\left\{\psi_2(m_{21})\right\} 
&=\sum_{a\in J_2} \sum_{u_2\in
T_{2\eta}(U_2|q^n)} \frac{1}{|G|^n}\cdot \frac{1}{|G|^{t_2}}
=\frac{|J_2| \cdot |T_{2\eta}(U_2|q^n)|}{|G|^n \cdot |G|^{t_2}}
\end{align*}
and let $\pmb{r}$, $\pmb{0}$ be vectors whose components are indexed by
$(p,r)\in\mathcal{Q}(G)$, and whose $(p,r)\textsuperscript{th}$
component is equal to $r$ and $0$, respectively. Then,
\begin{align*}
\mathbb{E}\left\{\psi_2(m_{21})^2\right\} 
&=\sum_{\theta\in\Theta} \sum_{a\in J_2} \sum_{\tilde{a}\in
T_{J_2,\theta}(a)} \sum_{u_2\in
T_{2\eta}^n(U_2|q^n)}  \sum_{\substack{\tilde{u}_2\in
T_{2\eta}^n(U_2|q^n)\\ \tilde{u}_2\in
u_2+H_{\theta}^n}} \frac{1}{|G|^n}\cdot \frac{1}{|H_{\theta}|^n}\cdot
P\left(I_2(a)=m_{21},I_2(\tilde{a})=m_{21}\right)\\ 
&=\sum_{a\in J_2} \sum_{u_2\in T_{2\eta}^n(U_2|q^n)} \frac{1}{|G|^n} \cdot \frac{1}{|G|^{t_2}}\\
&+\sum_{\substack{\theta\in\Theta\\\theta\ne \pmb{r}}} \sum_{a\in
J_2} \sum_{\tilde{a}\in T_{J_2,\theta}(a)} \sum_{u_2\in
T_{2\eta}^n(U_2|q^n)}  \sum_{\substack{\tilde{u}_2\in
T_{2\eta}^n(U_2|q^n)\\ \tilde{u}_2\in
u_2+H_{\theta}^n}} \frac{1}{|G|^n}\cdot \frac{1}{|H_{\theta}|^n}\cdot \frac{1}{|G|^{2t_2}}\\ 
&\le \frac{|J_2| \cdot |T_{2\eta}^n(U_2|q^n)|}{|G|^n \cdot |G|^{t_2}}
+ \sum_{\substack{\theta\in\Theta\\\theta\ne \pmb{r}}} \sum_{a \in J_2}\frac{
|T_{J_2,\theta}(a)|\cdot
|T_{2\eta}^n(U_2|q^n)| \cdot \left|T_{2\eta}^n(U_2|q^n)\cap
(u_2+H_{\theta}^n)\right| }{|G|^n \cdot |H_{\theta}|^n\cdot
|G|^{2 t_2}}\\ 
\end{align*}
Using \cite[Lemma IX.2]{201305TITarXiv_SahPra}, we
get 
\begin{align*}
\frac{\Var\left\{\psi_2(m_{21})^2\right\}}{\mathbb{E}^2(\psi_2(m_{21}))}
 &\le \frac{|G|^n \cdot |G|^{t_2}}{|J_2|\cdot
2^{n[H(U_2|Q)-\eta]}}
+ \sum_{\substack{\theta\in\Theta\\\theta\ne \pmb{0}, \theta\ne \pmb{r}}} \sum_{a \in
J_2} \frac{|G|^n \cdot
|T_{J_2,\theta}(a)| \cdot 2^{n[H(U_2|[U_2]_{\theta}Q)+\eta]} }{|J_2|^2
|H_{\theta}|^n \cdot 2^{n[H(U_2|Q)-\eta]}}\\ 
\end{align*}
Using \cite[Lemma IX.2]{201305TITarXiv_SahPra} we have 
$|T_{J_2,\theta}(a)| \leq 2^{n(1-w_{\theta})(S_2 +\eta_3)}$, and hence
for the probability of error to go to zero, we require 
\begin{align*}
&(S_2-T_2) > \log|G| - H(U_2|Q), \ \ \  
S_2 > \max_{\substack{\theta\in\Theta\\\theta\ne \pmb{0}}} \frac{1}{\omega_{\theta}} [\log |G:H_{\theta}|-H([U_2]_{\theta}|Q)].
\end{align*}

\textit{Upper bound on
$P \left( (\epsilon_{11} \cup \epsilon_{l_2} \cup \epsilon_{l_3}
\cup \epsilon_{2} \cup \epsilon_3 )^c \cap \epsilon_{41} \right)$}:
This probability can be decomposed into two parts: (i) the first,  $P_1$,
is the probability of the event that $X_1^n$ and $U_2^n+U_3^n$ are both
decoded incorrectly and (ii) the second, $P_2$, is the probability of the event
that $X_1^n$ is decoded incorrectly but $U_2^n+U_3^n$ is decoded
correctly.  In the following we provide an upper bound only on the
first part. For a fixed code we have,
\begin{align*}
P_1
&\!\le\! \frac{1}{|\mathcal{M}_1|} \! \sum_{m_1} \sum_{x_1,u_2,u_3 \in
T_{2\eta_2}(X_1,U_2,U_3|q^n)} \!\!\!\!\!\!\!\!\mathds{1}_{\{X_1^n(m_1)=x_1\}} 
\frac{1}{|G|^{t_2}}\sum_{m_{21},m_{31}} \frac{1}{|G|^{t_3}}\frac{2}{\mathbb{E}\{\psi_2(m_{21})\}} \sum_{u_3 \in
T_{2\eta_2}^n(U_3)} \frac{2}{\mathbb{E}\{\psi_3(m_{31})\}}  \nonumber \\
&\qquad \sum_{y_1\in\mathcal{Y}_1^n}
p_{Y_1|X_1,U_2,U_3}^n(y_1|x_1,u_2,u_3,q^n) 2^{-2
n \eta_4}   \sum_{\tilde{m}_1 \ne m_1} \sum_{a \in
J_2, b \in
J_3}\mathds{1}_{\{\phi_2(a)+b_2=u_2,\phi_3(b)+b_3=u_3,i_2(a)=m_{21},i_3(b)=m_{31}\}} \nonumber \\
&\qquad \sum_{\substack{(\tilde{x_1},\tilde{z})\in
T_{4\eta_4}(X_1,Z|y_1,q^n)\\\tilde{z}\ne z}}  \mathds{1}_{\{\tilde{x}_1=X_1^n(\tilde{m}_1)\}} 
\mathds{1}_{\{\exists \tilde{c}\in
J: \phi(\tilde{c})+b_2+b_3=\tilde{z},\tilde{c} \neq a+b\}} 
\end{align*}
Taking expectation and using the union
bound we get
\begin{align*}
\mathbb{E}\{P_1\}
&\le \sum_{\substack{\theta\in\Theta\\\theta\ne \pmb{r}}}  \sum_{a \in
J_2} \sum_{b \in J_3} \frac{2^{-2 \eta_4}
2^{nR_1}\cdot 2^{n[H(X_1|Z,Y_1Q)+\eta]}\cdot 2^{n[H(Z|[Z]_{\theta}Y_1Q)+\eta]} \cdot 
|T_{J,\theta}(a+b)|}{ |J_2||J_3|2^{n[H(X_1|Q)-\eta]}\cdot |H_{\theta}^n|}
\end{align*}
Using \cite[Lemma IX.2]{201305TITarXiv_SahPra},  note that
$|T_{J,\theta}(a+b)| \leq 2^{n(1-\omega_{\theta})R_g}$. Therefore, it suffices to have
\begin{align*}
R_1+(1-\omega_{\theta}) R_g <I(X_1;Y_1|ZQ)+\log |H_{\theta}|-H(Z|[Z]_{\theta}Y_1Q)
\end{align*}
for $\theta\ne \pmb{r}$. For optimum weights
$\{w_{p,r}\}_{(p,r)\in \mathcal{Q}(G)}$, the condition $R_1+R_g < I(X_1;Y_1|ZQ)+C^G_w(Z;Y_1|Q)+ \log |G|  -H(Z|Q)$ implies
\begin{align*}
R_g 
&\stackrel{(a)}{<} \min_{\substack{\theta\in\Theta\\\theta\ne \pmb{r}}} \frac{1}{1-\omega_{\theta}}\left[I(X_1;Y_1|ZQ)-R_1\right]
+ \min_{\substack{\theta\in\Theta\\\theta\ne \pmb{r}}} \frac{1}{1-\omega_{\theta}}\left[\log
|H_{\theta}|-H(Z|[Z]_{\theta}Y_1Q)\right]\\ 
&\le \min_{\substack{\theta\in\Theta\\\theta\ne \pmb{r}}} \frac{1}{1-\omega_{\theta}}\left[ 
I(X_1;Y_1|ZQ)-R_1+\log |H_{\theta}|-H(Z|[Z]_{\theta}Y_1Q)\right]
%&\stackrel{(a)}{=} \min_{\substack{\theta\in\Theta\\\theta\ne \pmb{r}}}
%(I(X_1;Y_1)-R_1) \min_{\substack{\theta\in\Theta\\\theta\ne \pmb{r}}} \left\(\frac{1}{1-\omega_{\theta}}\left[\log
%|H_{\theta}|-H(Z|[Z]_{\theta}X_1Y_1)\right]\right)\\ 
%&\le \min_{\substack{\theta\in\Theta\\\theta\ne \pmb{r}}} \left\( I(X_1;Y_1)-R_1+\frac{1}{1-\omega_{\theta}}\left[\log |H_{\theta}|-H(Z|[Z]_{\theta}X_1Y_1)\right]\right)\\
\end{align*}
which is the desired condition. In the above equations, $(a)$ follows
since the maximum of $1-\omega_{\theta}$ is attained for
$\theta=\pmb{0}$ and is equal to $1$. We have thus proved the bounds provided in definition
(\ref{Defn:3To1ICAchRegionSimplifiedPreFourierMotzkinBounds}) suffice to drive the
probability of incorrect decoding exponentially down to $0$.

\section{Proof of proposition \ref{thm:3To1ICGroupAddition}}
\label{AppSec:abeliangroupoptimality}

 We first note that for any
 $p_{QU_{2}U_{3}\ulineInputRV\ulineOutputRV} \in
 \mathbb{D}_{u}(\tau,0,0)$ with $H(X_{j}|Q,U_{j})=0$ for $j=2,3$, we
 have $R_{1}+R_{2}+R_{3} < I(\ulineInputRV;Y_{1}) \leq
 \underset{p_{X_{1}}p_{X_{2}}p_{X_{3}}}{\sup}
 I(\ulineInputRV;Y_{1})$. This follows from substituting the corresponding quantities in
 (\ref{Eqn:SumRateBoundOn3To1IC}). It can be easily verified that
 $\underset{p_{X_{1}}p_{X_{2}}p_{X_{3}}}{\sup} I(\ulineInputRV;Y_{1})
 = 2-h_{b}(\delta_{1})-\delta_{1}\log_{2}3$ which is achieved for all
 those distributions $p_{X_{1}}p_{X_{2}}p_{X_{3}}$ that ensure $Y_{1}$
 is uniformly distributed. Condition
 (\ref{Eqn:3To1GroupICUpperBoundUnstructuredCodes}) therefore implies
 $(C^{*},2-h_{b}(\delta)-\delta\log_{2}3,2-h_{b}(\delta)-\delta\log_{2}3)
 \notin \alpha_{u}(p_{QU_{2}U_{3}\ulineInputRV\ulineOutputRV})$  
 % \begin{equation}
%   \label{Eqn:3To1GroupICrateTripleNotAchievableWithparTestCh}
%   (C^{*},2-h_{b}(\delta)-\delta\log_{2}3,2-h_{b}(\delta)-\delta\log_{2}3) \notin \alpha_{u}(p_{QU_{2}U_{3}\ulineInputRV\ulineOutputRV}) \nonumber
%  \end{equation}
if $H(X_{j}|Q,U_{j})=0$ for $j=2,3$. Hence either $H(X_{2}|Q,U_{2})>0$
or $H(X_{3}|Q,U_{3})>0$. Assume $H(X_{j}|Q,U_{j})>0$ and
$\{j,\msout{j}\}=\{2,3\}$. By the conditional independence of
$(U_{2},X_{2})$ and $(U_{3},X_{3})$ given $Q$, we have
$0<H(X_{j}|Q,U_{j}) = H(X_{j}|Q,U_{j},U_{\msout{j}},X_{\msout{j}}) =
H(X_{j}\GroupSum
X_{\msout{j}}|Q,U_{j},U_{\msout{j}},X_{\msout{j}})=H(X_{2}\GroupSum
X_{3}|Q,U_{2},U_{3},X_{\msout{j}}) \leq H(X_{2}\GroupSum
X_{3}|Q,U_{2},U_{3})$. 
We only need to prove $H(X_{2}\GroupSum X_{3}|Q,U_{2},U_{3}) > 0$ implies
 $I(X_{1};Y_{1}|Q,U_{2},U_{3}) < C^{*}$. For this, we allude to the proof of fifth claim in appendix \ref{AppSec:3To1ORNonAddICStrictSubOptimalityOfUnstructuredCodes}. Therein, we have proved an analogous statement for example \ref{Ex:3To1ORICCoupledThroughNonAdditiveMAC}. The statement herein can be proved through an analogous sequence of steps and we let the reader fill in these details.

We now show that user 1 can achieve rate equal to
$C^{*}$ exploiting the fact that user 2 and 3 use  group codes.
We also derive the condition (\ref{Eqn:3To1GroupICUpperBoundUnstructuredCodes}) in terms of
parameters $\delta_{1},\tau,\delta$. 
Note that the channel between $X_2 \oplus_{4} X_3$ and $Y_1$ is additive with
noise given by $X_1 \oplus_{4} N_1$. Let us choose $p_{X_{1}}(x_{1})=\frac{\tau}{3}$ for $x_{1} \in \{1,2,3
\}$. The resulting distribution of $X_1\oplus_{4} N_1$ is given by
$p_{X_1\oplus_{4} N_1}(a)= \beta/3$ for $a \in \{1,2,3\}$.
Using concavity of entropy once again, we get
\begin{align}
C_w^G(X_2 \oplus_{4} X_3 ;Y)=\min\{2-h_b(\beta)-\beta \log_2(3),
2+2h_b(2\beta/3)-2h_b(\beta)-2\beta \log_2(3) \}=
2-h_b(\beta)-\beta \log_2(3). \nonumber
\end{align}
Note that for $\delta_1 \in (0, \frac{1}{4})$ and $\tau <
\frac{3}{4}$, using the fact that $X_1$ and $N_1$ are independent, we get
 $\beta \in (0, \frac{3}{4})$. Note also that 
$2-h_b(\beta)-\beta \log_2(3)$ is monotone
decreasing for $\beta \in (0,3/4)$. Hence if 
$\beta \leq  \delta$, the signal $X_2 \oplus_{4} X_3$ can be decoded
at decoder 1, and user 1 can communicate at the rate $C^*$. 
A simple calculation yields  $C^{*}=h_{b}(\beta)+\beta\log_{2}3-h_{b}(\delta_{1})-\delta_{1}\log_{2}3$.

\bibliographystyle{../sty/IEEEtran}
{
\bibliography{IEEEabrv,wisl}

\begin{thebibliography}{10}
\providecommand{\url}[1]{#1}
\csname url@rmstyle\endcsname
\providecommand{\newblock}{\relax}
\providecommand{\bibinfo}[2]{#2}
\providecommand\BIBentrySTDinterwordspacing{\spaceskip=0pt\relax}
\providecommand\BIBentryALTinterwordstretchfactor{4}
\providecommand\BIBentryALTinterwordspacing{\spaceskip=\fontdimen2\font plus
\BIBentryALTinterwordstretchfactor\fontdimen3\font minus
  \fontdimen4\font\relax}
\providecommand\BIBforeignlanguage[2]{{%
\expandafter\ifx\csname l@#1\endcsname\relax
\typeout{** WARNING: IEEEtran.bst: No hyphenation pattern has been}%
\typeout{** loaded for the language `#1'. Using the pattern for}%
\typeout{** the default language instead.}%
\else
\language=\csname l@#1\endcsname
\fi
#2}}

\bibitem{197801TIT_Car}
A.~Carleial, ``Interference channels,'' \emph{Information Theory, IEEE
  Transactions on}, vol.~24, no.~1, pp. 60--70, 1978.

\bibitem{198101TIT_HanKob}
T.~Han and K.~Kobayashi, ``A new achievable rate region for the interference
  channel,'' \emph{Information Theory, IEEE Transactions on}, vol.~27, no.~1,
  pp. 49 -- 60, jan 1981.

\bibitem{200808TIT_AliMotKha}
M.~Maddah-Ali, A.~Motahari, and A.~Khandani, ``Communication over {M}{I}{M}{O}
  {X} channels: Interference alignment, decomposition, and performance
  analysis,'' \emph{Information Theory, IEEE Transactions on}, vol.~54, no.~8,
  pp. 3457--3470, Aug 2008.

\bibitem{200808TIT_CadJaf}
V.~Cadambe and S.~Jafar, ``Interference alignment and degrees of freedom of the
  k -user interference channel,'' \emph{IEEE Trans. on Info. Th.}, vol.~54,
  no.~8, pp. 3425--3441, 2008.

\bibitem{201009TIT_BreParTse}
G.~Bresler, A.~Parekh, and D.~Tse, ``The approximate capacity of the
  many-to-one and one-to-many {G}aussian interference channels,''
  \emph{Information Theory, IEEE Transactions on}, vol.~56, no.~9, pp. 4566
  --4592, sept. 2010.

\bibitem{201408TIT_HonCai}
S.-N. Hong and G.~Caire, ``On interference networks over finite fields,''
  \emph{Information Theory, IEEE Transactions on}, vol.~60, no.~8, pp.
  4902--4921, Aug 2014.

\bibitem{201407TIT_KriJaf}
S.~Krishnamurthy and S.~Jafar, ``On the capacity of the finite field
  counterparts of wireless interference networks,'' \emph{Information Theory,
  IEEE Transactions on}, vol.~60, no.~7, pp. 4101--4124, July 2014.

\bibitem{6516907}
U.~Niesen and M.~Maddah-Ali, ``Interference alignment: From degrees of freedom
  to constant-gap capacity approximations,'' \emph{Information Theory, IEEE
  Transactions on}, vol.~59, no.~8, pp. 4855--4888, Aug 2013.

\bibitem{197903TIT_KorMar}
J.~K\"orner and K.~Marton, ``How to encode the modulo-two sum of binary sources
  (corresp.),'' \emph{{IEEE} Trans. Inform. Theory}, vol.~25, no.~2, pp. 219 --
  221, Mar 1979.

\bibitem{200906TIT_PhiZam}
T.~Philosof and R.~Zamir, ``On the loss of single-letter characterization: The
  dirty multiple access channel,'' \emph{IEEE Trans. on Info. Th.}, vol.~55,
  pp. 2442--2454, June 2009.

\bibitem{201108TIT_PhiZamEreKhi}
T.~Philosof, R.~Zamir, U.~Erez, and A.~Khisti, ``Lattice strategies for the
  dirty multiple access channel,'' \emph{Information Theory, IEEE Transactions
  on}, vol.~57, no.~8, pp. 5006--5035, Aug 2011.

\bibitem{200812IEEEI_KocKhiEreZam}
Y.~Kochman, A.~Khina, U.~Erez, and R.~Zamir, ``Rematch and forward: Joint
  source/channel coding for communications,'' in \emph{Electrical and
  Electronics Engineers in Israel, 2008. IEEEI 2008. IEEE 25th Convention of},
  Dec 2008, pp. 779--783.

\bibitem{200911TIT_KocZam}
Y.~Kochman and R.~Zamir, ``Joint wyner-ziv/dirty-paper coding by modulo-lattice
  modulation,'' \emph{Information Theory, IEEE Transactions on}, vol.~55,
  no.~11, pp. 4878--4889, Nov 2009.

\bibitem{201103TIT_KriPra}
D.~Krithivasan and S.~Pradhan, ``Distributed source coding using abelian group
  codes: A new achievable rate-distortion region,'' \emph{Information Theory,
  IEEE Transactions on}, vol.~57, no.~3, pp. 1495--1519, March 2011.

\bibitem{200710TIT_NazGas}
B.~Nazer and M.~Gastpar, ``Computation over multiple-access channels,''
  \emph{IEEE Trans. on Info. Th.}, vol.~53, no.~10, pp. 3498 --3516, oct. 2007.

\bibitem{200607ISIT_AliMotKha}
M.~Maddah-Ali, A.~Motahari, and A.~Khandani, ``Signaling over {M}{I}{M}{O}
  multi-base systems: Combination of multi-access and broadcast schemes,'' in
  \emph{Information Theory, 2006 IEEE International Symposium on}, July 2006,
  pp. 2104--2108.

\bibitem{200809Allerton_SriJafVisJafSha}
S.~Sridharan, A.~Jafarian, S.~Vishwanath, S.~Jafar, and S.~Shamai, ``A layered
  lattice coding scheme for a class of three user {G}aussian interference
  channels,'' in \emph{2008 46th Annual Allerton Conference Proceedings on},
  sept. 2008, pp. 531 --538.

\bibitem{201109TITarXiv_JafVis}
A.~Jafarian and S.~Vishwanath, ``Achievable rates for k-user {G}aussian
  interference channels,'' submitted to {I}{E}{E}{E} {T}rans. of {I}nformation
  theory 2011, available at {\tt http://arxiv.org/abs/1109.5336}.

\bibitem{6283726}
O.~Ordentlich, U.~Erez, and B.~Nazer, ``The approximate sum capacity of the
  symmetric {G}aussian {K}-user interference channel,'' in \emph{Information
  Theory Proceedings (ISIT), 2012 IEEE International Symposium on}, July 2012,
  pp. 2072--2076.

\bibitem{201105TIT_BanGam}
B.~Bandemer and A.~El~Gamal, ``Interference decoding for deterministic
  channels,'' \emph{Information Theory, IEEE Transactions on}, vol.~57, no.~5,
  pp. 2966--2975, 2011.

\bibitem{201301arXivMACDSTx_PadPra}
A.~Padakandla and S.~Pradhan, ``Achievable rate region based on coset codes for
  multiple access channel with states,'' submitted to IEEE Trans. on Info. Th.,
  available at {\tt http://arxiv.org/abs/1301.5655}.

\bibitem{200912arXiv_CadJaf}
V.~R. Cadambe and S.~A. Jafar, ``Interference alignment and a noisy
  interference regime for many-to-one interference channels,'' available at
  {\tt http://arxiv.org/abs/0912.3029}.

\bibitem{198305TIT_AhlHan}
R.~Ahlswede and T.~Han, ``On source coding with side information via a
  multiple-access channel and related problems in multi-user information
  theory,'' \emph{IEEE Trans. on Info. Th.}, vol.~29, no.~3, pp. 396 -- 412,
  may 1983.

\bibitem{Gal-ITRC68}
R.~G. Gallager, \emph{Information Theory and Reliable Communication}.\hskip 1em
  plus 0.5em minus 0.4em\relax New York: John Wiley \& Sons, 1968.

\bibitem{201110TIT_NazGas}
B.~Nazer and M.~Gastpar, ``Compute-and-forward: Harnessing interference through
  structured codes,'' \emph{Information Theory, IEEE Transactions on}, vol.~57,
  no.~10, pp. 6463--6486, Oct 2011.

\bibitem{201304arXivTIT_KriZaf}
S.~R. Krishnamurthy and S.~A. Jafar, ``On the capacity of the finite field
  counterparts of wireless interference networks,'' available at {\tt
  http://arxiv.org/abs/1304.7745}.

\bibitem{1972MMPCT_GacKor}
P.~G\'acs and J.~K\"orner, ``Common information is far less than mutual
  information,'' \emph{Problems of Control and Information Theory}, vol.~2,
  no.~2, pp. 119--162, 1972.

\bibitem{197501SJAM_Wit}
H.~S. Witsenhausen, ``On sequences of pairs of dependent random variables,''
  \emph{SIAM Journal of Applied Mathematics}, vol.~28, no.~1, pp. 100--113,
  January 1975.

\bibitem{Hal-TTOG59}
M.~Hall, \emph{The theory of groups}.\hskip 1em plus 0.5em minus 0.4em\relax
  New York: Macmillan, 1959.

\bibitem{1965MMAMS_Hoe}
W.~Hoeffding, ``Asymptotically optimal tests for multinomial distributions,''
  \emph{Annals of Mathematical Statistics}, vol.~36, no.~2, pp. 369--401, 1965.

\bibitem{1957MMMS_San}
I.~Sanov, ``On the probability of large deviations of random variables,''
  \emph{Matematicheskii Sbornik}, vol. 42(84), pp. 11--44, 1957, translated by
  Dana E. A. Quade, Institute of Statistics, Mimeograph Series No. 192, March
  1958, available at.

\bibitem{201305TITarXiv_SahPra}
A.~Sahebi and S.~Pradhan, ``Abelian group codes for source coding and channel
  coding,'' submitted to {I}{E}{E}{E} {T}rans. of {I}nformation theory, April
  2013, available at {\tt http://arxiv.org/abs/1305.1598}.

\end{thebibliography}
}
\end{document}